\documentclass[a4paper,oneside,11pt]{article}

\usepackage[utf8]{inputenc}        % utf8 encoding
\usepackage[T1]{fontenc}             %modern font encoding
\usepackage{amsmath,amssymb,amsthm,bm}
\usepackage[english]{babel}
\usepackage{enumitem}
\usepackage{tikz}
\usepackage{microtype}
\usepackage{hyperref}
\usepackage{algorithm}
\usepackage[noend]{algpseudocode}

\usepackage{xspace} %For spacing after custom commands
\usepackage{placeins}
\usetikzlibrary{shapes,calc,arrows}
\usetikzlibrary{arrows.meta}

\tikzstyle{graphnode}=[circle,draw,minimum size=20pt,scale=0.8]
\tikzstyle{dnode}=[scale=0.8]
\tikzset{
  font={\fontsize{10pt}{12}\selectfont}}
  
\usepackage{subcaption}
\usepackage{ifthen}

\usepackage[margin=2.5cm]{geometry}
\newtheorem{theorem}{Theorem}
\newtheorem{definition}[theorem]{Definition}
\newtheorem{corollary}[theorem]{Corollary}
\newtheorem{lemma}[theorem]{Lemma}
\newtheorem{remark}{Remark}

\newcommand{\agent}{A}
\newcommand{\agentset}{\mathcal{A}}
\newcommand{\turing}{M}
\newcommand{\bigO}{\mathcal{O}}
\newcommand{\N}{\mathbb{N}}
\newcommand{\regExt}[1]{\overline{#1}}

\newcommand{\ceil}[1]{\left\lceil #1 \right\rceil}
\renewcommand{\vec}[1]{\boldsymbol #1}

\newcommand{\pebblemachine}{T}

\newcommand{\startvertex}{v_0}

\newcommand{\Agentstate}{\Sigma}
\newcommand{\agentstate}{\sigma}
\newcommand{\agentstartstate}{\agentstate^*}

\makeatletter
\providecommand*{\udotdot}{%
  \mathbin{\mathpalette\@udotdot{}}%
}
\providecommand*{\@udotdot}[2]{%
  \sbox0{$#1.\m@th$}%
  \setbox2=\hbox to \ht0{\hss\copy0\hss}
  \dimen2=.5\dimexpr\wd0-\wd2\relax
  \sbox4{$#1=\m@th$}%
  \sbox6{$#1\vcenter{}$}%
  \dimen@=\dimexpr(\ht4-\ht6)*2 + \ht0\relax
  \kern\dimen2 % side bearing
  \vcenter to \dimen@{%
    \hbox to \dimen@{\hfill\copy2}%
    \nointerlineskip
    \vfill
    \hbox to \dimen@{\copy2\hfill}%
  }%
  \kern\dimen2 % side bearing
}
\makeatother

\usepackage{authblk}
\author[1]{Yann Disser\thanks{Supported by the `Excellence Initiative' of the German Federal and State Governments and the Graduate School~CE at TU~Darmstadt.}}
\author[2]{Jan Hackfeld}
\author[3]{Max Klimm}
\affil[1]{Graduate School CE, TU Darmstadt, Germany \protect\\
 \texttt{disser@mathematik.tu-darmstadt.de}}
\affil[2]{School of Business and Economics, HU Berlin, Germany \protect\\
\texttt{\{jan.hackfeld|max.klimm\}@hu-berlin.de}}

\begin{document}
\title{Tight bounds for undirected graph exploration\\with pebbles and multiple agents\thanks{Results concerning exploration with pebbles appeared in preliminary form in~\cite{disser15}.}}

\maketitle

\begin{abstract} 
We study the problem of deterministically exploring an undirected and initially unknown graph with $n$ vertices either by a single agent equipped with a set of pebbles, or by a set of collaborating agents. The vertices of the graph are unlabeled and cannot be distinguished by the agents, but the edges incident to a vertex have locally distinct labels. The graph is explored when all vertices have been visited by at least one agent. In this setting, it is known that for a single agent without pebbles $\Theta(\log n)$ bits of memory are necessary and sufficient to explore any graph with at most~$n$ vertices. We are interested in how the memory requirement decreases as the agent may mark vertices by dropping and retrieving distinguishable pebbles, or when multiple agents jointly explore the graph. We give tight results for both questions showing that for a single agent with constant memory $\Theta(\log \log n)$ pebbles are necessary and sufficient for exploration. We further prove that using collaborating agents instead of pebbles does not help as $\Theta(\log \log n)$ agents with constant bits of memory each are necessary and sufficient for exploration.

For the upper bounds, we devise an algorithm for a single agent with constant memory that explores any $n$-vertex graph using $\bigO(\log \log n)$ pebbles, even when $n$ is not known a priori. The algorithm terminates after polynomial time and returns to the starting vertex. Since an additional agent is at least as powerful as a pebble, this implies that $\bigO(\log \log n)$ agents with constant memory can explore any $n$-vertex graph.
For the lower bound, we show that the number of agents needed for exploring any graph with at most~$n$ vertices is already  $\Omega(\log \log n)$ when we allow each agent to have at most $\bigO( \log n ^{1-\varepsilon})$ bits of memory for any $\varepsilon>0$. This also implies that a single agent with sublogarithmic memory needs $\Theta(\log \log n)$ pebbles to explore any $n$-vertex graph.

\end{abstract}

%!TEX root = 0_pebbles_main.tex
\section{Introduction}

The exploration of unknown graphs subject to space or time bounds is a fundamental problem with applications to robot navigation, Internet crawling, and image recognition. Its pivotal role within the theory of computation stems from the fact that it is a natural abstraction of a process of computing where nodes correspond to states, edges correspond to possible state transitions, and the goal is to find an accepting state when starting in a given initial state. Single agent exploration then corresponds to computing with a single processing unit while exploration with multiple agents corresponds to parallel computing. In this context, graph exploration and traversal problems have proven to be useful to study the relationship between probabilistic and deterministic space-bounded algorithms \cite{savitch73}.

The time and space complexity of undirected graph exploration by a \emph{single} agent is fairly well understood. Aleliunas et al.~\cite{aleliunas79} showed that a random walk of length $n^5 \log n$ visits all vertices of any $n$-vertex graph with high probability. When $n$ is known, it is thus possible to use a counter that keeps track of the number of steps taken to obtain a probabilistic algorithm that explores any graph of size $n$ in polynomial time and $\bigO(\log n)$ space with high probability. In a breakthrough result, Reingold~\cite{reingold08} showed that the same time and space complexity can be achieved by a \emph{deterministic} algorithm. Both the randomized algorithm of Aleliunas et al.\ and the deterministic algorithm of Reingold work for anonymous graphs where vertices are indistinguishable. Logarithmic memory is in fact necessary to explore all anonymous graphs with $n$ vertices; see Fraigniaud et al.~\cite{fraigniaud05}.

Already the early literature on graph exploration problems is rich with examples where exploration is made feasible or the time or space complexity of exploration by a single agent can be decreased substantially by either allowing the agent to mark vertices with pebbles or by cooperating with other agents. For instance, two-dimensional mazes can be explored by a single agent with finite memory using two pebbles \cite{shah74,blum77,blum78}, or by two cooperating agents with finite memory \cite{blum78}, while a single agent with finite memory (and even a single agent with finite memory and a single pebble) does not suffice \cite{budach78,hoffmann81}. Directed anonymous graphs can be explored in polynomial time by two cooperating agents \cite{bender94} or by a single agent with $\Theta(\log \log n)$ indistinguishable pebbles and $\bigO(n^2 \Delta \log n)$ bits of memory \cite{bender02,fraigniaud04}, where $\Delta$ is the maximum out-degree in the graph. Note that a single agent needs at least 
 $\Omega(n \log \Delta)$ bits of memory in this setting even if it is equipped with a linear number of  indistinguishable pebbles \cite{fraigniaud04} and it needs
 exponential time for exploration if it only has a constant number of pebbles and no upper bound on the number of vertices is known \cite{bender02}. 

Less is known regarding the complexity of general undirected graph exploration by more than one agent or an agent equipped with pebbles. The only result in this direction is due to Rollik~\cite{rollik80} showing that there are finite graphs, henceforth called \emph{traps}, that a finite set of $k$ agents each with a finite number $s$ of states cannot explore. Fraigniaud et al.~\cite{fraigniaud06} revisited Rollik's construction and observed that the traps have $\tilde{\bigO}(s \uparrow\uparrow (2k+1))$ vertices, where $a \uparrow \uparrow b := \smash{a^{a^{\udotdot^a}}}$ with~$b$ levels in the exponent and $\tilde{\bigO}$ suppresses lower order terms. Fraigniaud et al.~also gave an improved upper bound of $\tilde{\bigO}(s \uparrow \uparrow (k+1))$. 
While it is a rather straightforward observation that an agent with $s$ states and $p$ pebbles is less powerful than a set of $p+1$ agents with $s$ states each, no better bounds for a single agent with pebbles were known. 
Even more striking is the lack of any non-trivial upper bounds for the exploration with several agents or the single agent exploration with pebbles for undirected graphs. 
Specifically, there was no algorithm known that explores an undirected graph with sublogarithmic space when more than one agent and/or pebbles are allowed.

\subsection{Our results}

We give tight bounds for both the space complexity of undirected graph exploration by a single agent with pebbles, as well as by a set of cooperating agents.

\subsubsection{Results for exploration with pebbles}

For the exploration of a graph by a single agent with constant memory, we show that $\Theta(\log \log n)$ distinguishable pebbles are necessary and sufficient to explore all undirected anonymous graphs with at most $n$ vertices. 

Our proof of the upper bound is constructive, i.e., we devise an algorithm that explores any graph with $n$ vertices using $\bigO(\log \log n)$ pebbles (Corollary~\ref{cor:constant_memory_pebbles}). Our algorithm terminates after having explored the graph and returns to the starting vertex. We further show that the exploration time, i.e., the number of edge traversals of the agent, is polynomial in the size of the graph. Our algorithm does not require $n$ to be known and gradually increases the number of used pebbles during the course of the algorithm such that for any $n$-vertex graph at most $f(n)$ pebbles are used where $f(n) \in \bigO(\log \log n)$.

For a lower bound we show that a single agent with sublogarithmic memory (more precisely $\bigO((\log n)^{1-\epsilon})$ bits of memory for an arbitrary constant $\epsilon > 0$) already needs $\Omega(\log \log n)$ pebbles for exploring every graph with $n$ vertices (Corollary~\ref{cor_pebbles_lower_bound}). 
Our results fully characterize the tradeoff between the memory and the number of pebbles of an agent needed for exploration. It turns out that this tradeoff is governed by two thresholds. When the agent has $\Omega(\log n)$ bits of memory, no pebbles are needed at all. But as soon as the memory is $\bigO((\log n)^{1-\epsilon})$ already $\Omega(\log \log n)$ pebbles are needed to explore all $n$-vertex graphs. On the other hand, with $\Omega(\log \log n)$ pebbles already a constant number of bits of memory are sufficient for exploration.

\subsubsection{Results for multi-agent exploration}
\label{sec:results-multi}

For collaborative graph exploration, we show that $\Theta(\log \log n)$ agents with constant memory are necessary and sufficient to explore all undirected anonymous graphs with at most $n$ vertices. 

We first make the rather straightforward observation that a set of $p+1$ agents with a constant set of states each can reproduce the exploration by a single agent with a constant number of states and $p$ pebbles in the following way. One of the agents replicates the moves of the single agent while the others do not move independently and simply act as pebbles (Lemma~\ref{lem:agent>pebble}). This observation allows to rephrase our single agent exploration algorithm with $\smash{\bigO(\log \log n)}$ pebbles as a multi-agent exploration algorithm with $\bigO(\log \log n)$ agents and constant memory each (Corollary~\ref{cor:constant_memory_agents}). As a perhaps surprising result, we show that this is  optimal in terms of the asymptotic number of agents. To prove this lower bound, we construct a family of graphs with $\smash{\bigO(s^{2^{5k}})}$ vertices that trap any set of $k$ agents with $s$ states each (Theorem~\ref{theo_size_of_trap}). 
Our construction exhibits dramatically smaller traps with only a doubly exponential number of vertices compared to the traps of size $\tilde{\bigO}(s\uparrow \uparrow (2k+1))$ and $\tilde{\bigO}(s \uparrow \uparrow (k+1))$ due to Rollik~\cite{rollik80} and Fraigniaud et al.~\cite{fraigniaud06}, respectively. As a consequence of our improved bound on the size of the trap, we are able to show that, even if we allow $\bigO((\log n)^{1-\epsilon})$ bits of memory for an arbitrary constant $\epsilon>0$ for every agent, the number of agents needed for exploration is at least $\Omega(\log \log n)$ (Theorem~\ref{theo_pebbles_lower_bound}). This construction also yields the lower bound for a single agent with pebbles, as $p+1$ agents with $\bigO((\log n)^{1-\epsilon})$ bits of memory each are more powerful than one agent with $\bigO((\log n)^{1-\epsilon})$ bits of memory and $p$ pebbles. Our results again allow to fully describe the tradeoff between the number of agents and the memory of each agent. When agents have $\Omega(\log n)$ memory, a single agent explores all $n$-vertex graphs. For agents with $\bigO((\log n)^{1-\epsilon})$ memory, $\Omega(\log \log n)$ agents are needed. On the other hand, when $\Omega(\log \log n)$ agents are available it is sufficient that each of them has only constant memory.

\subsection{Related work}
Exploration algorithms were first designed for mazes which are finite anonymous subgraphs of the two-dimensional grid where edges are labeled with their cardinal direction. Shannon~\cite{shannon51} constructed an actual physical device -- Shannon's mouse --  that explores a $5 \times 5$ grid and uses constant memory per vertex. Budach~\cite{budach78} gave a proof that no agent with finite memory can explore any finite maze. Shah~\cite{shah74} proved that exploration is possible when the agent can use five pebbles. Blum and Sakoda~\cite{blum77} reduced the number of pebbles to four, and Blum and Kozen~\cite{blum78} showed that in fact two pebbles suffice. This result is tight as Hoffmann~\cite{hoffmann81} proved that one pebble does not suffice.  Blum and Kozen~\cite{blum78} also showed that  any maze can be explored by two agents with finite memory. As a contrast, they showed that there are finite planar cubic graphs that cannot be explored by three agents. Rollik~\cite{rollik80} strengthened the latter result showing that for any finite set of agents with finite memory there is a planar graph -- a so-called \emph{trap} -- that cannot be explored. 
For~$k$ agents, the trap is of order $\smash{\tilde{\bigO}(s\uparrow \uparrow (2k+1))}$, where $\smash{a \uparrow \uparrow b := a^{a^{\udotdot^a}}}$ with~$b$ levels in the exponent. 
This bound was improved by Fraigniaud et al.~\cite{fraigniaud06} to $\smash{\tilde{\bigO}(s \uparrow \uparrow (k+1))}$, and is further improved in this paper (see \S~\ref{sec:results-multi}).

Aleliunas et al.~\cite{aleliunas79} showed that a random walk of length $n^5 \log n$ explores any $n$-vertex graph with high probability. By the probabilistic method, this implies the existence of a universal traversal sequence (UTS) of polynomial length for regular graphs. A $(n,d)$-UTS is a sequence of port numbers in $\{0,\dots,d-1\}$ such that an agent starting in an arbitrary vertex and following the  port numbers of the sequence explores every $d$-regular graph with $n$ vertices. 
Using a counter for the number of steps, this yields a log-space randomized algorithm constructing a UTS when the number of vertices is known. Using Nisan's derandomization technique \cite{nisan92}, this gives a deterministic algorithm with $\bigO(\log^2 n)$ memory. The length of the sequence, however, increases to $\smash{2^{\bigO(\log^2 n)}} = \bigO(n^{\log n})$. Explicit constructions of UTS are only known for cycles (Istrail~\cite{istrail88}) and to date it remains open whether an UTS of polynomial length can be constructed deterministically in log-space for general graphs.

As a remedy for the perceived difficulty of constructing a UTS, Kouck\'y~\cite{koucky02} introduced the related concept of universal exploration sequences (UXS) where the next port number depends on the port number of the edge to the previously visited vertex. 
Reingold~\cite{reingold08} showed that an $(n,d)$-UXS can be constructed deterministically in $\bigO(\log n)$ space. 
A slight modification of his algorithm allows to explore any (not necessarily) regular graph, thus, solving the undirected $s$-$t$ connectivity problem in log-space. Fraigniaud et al.~\cite{fraigniaud05} showed that $\Omega(\log n)$ memory is necessary to explore all graphs with up to $n$ vertices. The set of graphs $\mathcal{G}_k$ that can be explored by an agent with $k$ states is increasing in the sense that there is polynomial $h\colon \N \to \N$ such that $\mathcal{G}_{h(i)}$ is strictly included in $\mathcal{G}_i$, see Fraigniaud et al.~\cite{fraigniaud08}.

Kouck\'y~\cite{koucky02} noted that $1^{2n}$ is a UXS for trees. As remarked by Diks et al.~\cite{diks02}, this gives rise to a \emph{perpetual} tree exploration algorithm that runs forever, and eventually visits all vertices of the tree. 
For the cases that the exploration algorithm is required to terminate they showed that $\Omega(\log \log \log n)$ bits of memory are needed. If the algorithm is even required to terminate at the same vertex where it started, $\Omega(\log n)$ bits of memory are needed. 
A matching upper bound for the latter problem was given by Ambühl et al.~\cite{ambuhl11}.

Regarding the exploration time, Dudek et al.~\cite{dudek91} showed that an agent provided with a pebble can map an undirected graph in time $\bigO(n^2 \Delta)$. In a similar vein, Chalopin et al.~\cite{chalopin10} showed that if the starting node can be recognized by the agent, the agent can explore the graph in time $\bigO(n^3 \Delta)$ using $\bigO(n \Delta \log n)$ memory.
Exploration with the objective of minimizing the exploration time also has been studied in terms of competitive analysis. 
In this context an exploration algorithm for an edge weighted graph is called $c$-competitive if the sum of the weights of the edges traversed by the algorithm is at most $c$ times that of an optimal offline walk. If the task is to visit all vertices of a vertex-labeled weighted graph and an agent learns about all neighbors when arriving at a node, a nearest neighbor greedy approach is $\Theta(\log n)$-competitive \cite{rosenkrantz77}. Algorithms with constant competitive ratios are only known for planar graphs \cite{kalyana94} and, more generally, graphs with bounded genus \cite{megow12}. 

Exploration becomes considerably more difficult for directed graphs, where random walks may need  exponential time to visit all vertices. Without any constraints on memory, Bender et al.~\cite{bender02} gave an $\bigO(n^8 \Delta^2 )$-time algorithm that uses one pebble and explores (and maps) am unlabeled  directed graph with maximum degree~$\Delta$, when an upper bound on the number~$n$ of vertices is known. For the case that such an upper bound is not available, they proved that $\Theta(\log \log n)$ pebbles are both necessary
and sufficient. Concerning the space complexity of directed graph exploration, Fraigniaud and Ilcinkas~\cite{fraigniaud04} showed that $\Omega(n \log \Delta)$ bits of memory are necessary to explore any directed graph with $n$ vertices and maximum degree~$\Delta$, even with a linear number of pebbles. 
They provided an $\bigO(n \Delta \log n)$-space algorithm for terminating exploration with an exponential running time using a single pebble. 
They also gave an $\bigO(n^2 \Delta \log n)$-space algorithm running in polynomial time and using $\bigO(\log \log n)$ indistinguishable pebbles. 
According to Bender et al.~\cite{bender02} at least~$\Omega(\log \log n)$ pebbles are necessary in this setting.
For the exploration time of labeled directed graphs, where the task is to traverse all edges, the competitive ratios achievable by online algorithms are closely related to the deficiency of the graph \cite{deng99,albers00,fleischer05}.  
If, on the other hand, the agent learns about all neighbors when arriving at a node and has to visit all vertices of the graph, the best possible competitive ratio is $\Theta(n)$, even for Euclidean planar graphs \cite{foerster16}.

Further related is the problem of exploring geometric structures in the plane, see, e.g.,  \cite{chalopin13,chalopin15} for the exploration of simple polygons and \cite{blum97,deng98} for the exploration of other geometric terrains.

\subsection{Techniques and outline of the paper}

In \S~\ref{sec:prelim}, we fix notation and introduce the agent model. We give formal proofs of the intuitive facts that for undirected graph exploration an additional agent (with two states) is more powerful than a pebble, and a pebble is more powerful than a bit of memory (Lemmas~\ref{lem:pebble>bit} and~\ref{lem:agent>pebble}). 

Our main positive result is presented in \S~\ref{sec:pebbles} where we give a single agent exploration algorithm that explores any $n$ vertex graph with $\bigO(\log \log n)$ pebbles and $\bigO(\log \log n)$ bits of memory. In light of the lemmata presented in \S~\ref{sec:prelim}, we obtain as direct corollaries that (a) a single agent with $\bigO(\log \log n)$ pebbles and constant memory can explore any $n$-vertex graph and (b) a set of $\bigO(\log \log n)$ agents with constant memory each can explore any $n$-vertex graph.

For the algorithm, we use the concept of universal exploration sequences due to Kouck\'y~\cite{koucky02}. 
One of our main building blocks is the algorithm of Reingold~\cite{reingold08} that takes $n$ and $d$ as input and deterministically constructs an exploration sequence universal to all $d$-regular graphs using $\bigO(\log n)$ bits of memory. 
The general idea of our algorithm is to run Reingold's algorithm with a smaller amount of seed memory $a$. As the seed memory is substantially less than $\bigO(\log n)$, the algorithm will, in general, fail to explore the whole graph. 
We show, however, that the algorithm will visit $2^{\Omega{(a)}}$ distinct vertices (Lemma~\ref{lemma_uxs_general_graph}). 
Reinvoking Reingold's algorithm allows us to \emph{deterministically} walk along these vertices in the order of exploration of Reingold's algorithm. 
Using this traversal, we encode additional memory by placing a subset of pebbles on the vertices along the walk. 
Having boosted our memory this way, we run again Reingold's algorithm, this time with more memory, and recurse. 
At some recursion depth, running Reingold's algorithm with $a^*$ bits of memory will visit less than $2^{\Omega{(a^*)}}$ distinct vertices. 
We show that this can only happen when the graph is fully explored which allows to terminate the algorithm when this event occurs. Unwrapping all levels of recursion then also allows to return to the starting vertex.
The ability of our algorithm to terminate and return to the starting vertex after successful exploration, stands in contrast to Reingold's algorithm that is only able to terminate when being given the number $n$ of vertices as input.

There are a couple of technical difficulties to make these ideas work. 
The main challenge is that the memory generated by placing pebbles along a walk in the graph is implicit and can only be accessed and altered locally. 
To still make use of the memory, we do not work with Reingold's algorithm directly but consider an implementation of Reingold's algorithm on a Turing machine with logarithmically bounded working tape. We show that the tape operations on the working tape can be reproduced by the agent by placing and retrieving the pebbles on the walk. This allows to use the memory encoded by the pebble positions for further runs of Reingold's algorithm. In each recursion, we only need a constant number of pebbles and additional states. We further show that $\bigO(\log \log n)$ recursive calls are necessary to explore an $n$-vertex graph so that the total number of pebbles needed is $\bigO(\log \log n)$. 

A second challenge is that Reingold's algorithm produces a UXS for regular graphs which our graph need not be. A natural approach to circumvent this issue is to apply the technique of Kouck\'y~\cite{koucky03} that allows to locally view vertices with degree $d$ as cycles of $3d$ subvertices with degree $3$ each. 
Unfortunately, this approach requires $\bigO(\log d)$ bits of memory if we keep track of the current subvertex which may exceed the memory of our agent.
To circumvent this issue, we store the current subvertex only implicitly and navigate the graph in terms of subvertex index offsets instead of the actual subvertex indices.

Our general lower bound is presented in \S~\ref{sec:lower_bound}. Specifically, we show that for a set of cooperative agents with sublogarithmic memory of $\bigO((\log n)^{1-\varepsilon})$ for some constant $\varepsilon > 0$, $\Omega(\log \log n)$ agents are needed to explore any undirected graph with $n$ vertices. In light of our reduction presented in \S~\ref{sec:prelim}, this implies that an agent with sublogarithmic memory needs $\Omega(\log \log n)$ pebbles to explore any $n$-vertex graph.

To prove the lower bound, we use the concept of an $r$-barrier. Informally, an $r$-barrier is a graph with two special entry points such that any subset of up to $r$ agents with $s$ states cannot reach one entry point when starting from the other.
Moreover, a set of  $r+1$ agents can explore an $r$-barrier, but the agents can only leave the barrier via the same entry point. 
We recursively construct an $r$-barrier by replacing edges by $(r-1)$-barriers. 
A set of $r$ agents traversing this graph needs to stay close to each other to be able to traverse the barriers and make progress. 
However, if the agents stay close to each other, the states and relative positions of the agents become periodic relatively quickly, and we can use this fact to build an $r$-barrier. By carefully bounding the size of the $r$-barriers in our
recursive construction, we obtain a trap of size $\bigO(s^{2^{5 k}})$ for any given set of $k$ agents with at most $s$ states each. 
The size of the trap directly implies that the number of agents with at most $\bigO((\log n)^{1-\varepsilon})$ bits of memory needed for exploring any graph of size $n$ is at least 
$\Omega(\log \log n)$.

\section{Terminology and model}
\label{sec:prelim}

\subsection{The graph}

Let $G = (V,E)$ be an undirected graph with $n$ vertices to be explored by a single agent or a group of agents. All agents start at the same vertex $\startvertex \in V$. We say that $G$ is explored when each vertex of $G$ has been visited by at least one agent. Every graph considered in this paper is assumed to be connected as otherwise exploration is not possible. We further assume that the graph is anonymous, i.e., the vertices of the graph are unlabeled and therefore cannot be identified or distinguished by the agents.  
To enable sensible navigation, the edges incident to a vertex~$v$ have distinct labels $0,\ldots,d_v-1$, where $d_v$ is the degree of~$v$.
This way every edge~$\{u,v\}$ has two -- not necessarily identical -- labels called \emph{port numbers}, one at~$u$ and one at~$v$.

\subsection{The agents}
\label{sec:prelim:agent}

We model an agent as a tuple $A = (\Agentstate, \bar{\Agentstate}, \delta, \agentstartstate)$ where $\Agentstate$ is its set of states, $\bar{\Agentstate}\subseteq{\Agentstate}$ is its set of halting states, $\agentstartstate \in \Agentstate$ is its starting state, and $\delta$ is its transition function. 
The transition function governs the actions of the agent and its transitions between states based on its local observations. Its exact specifics depend on the problem considered, i.e., whether we consider a single agent or a group of agents and whether we allow the agents to use pebbles. 
Exploration terminates when a halting state is reached by all agents.

\subsubsection{A single agent without pebbles}
\label{sec:prelim:agents:single}
The most basic model is that of a single agent $A$ without any pebbles. In each step, the agent observes its current state $\agentstate \in \Agentstate$, the degree $d_v$ of the current vertex $v$ and the port number at $v$ of the edge from which $v$ was entered. The transition function $\delta$ then specifies a new state $\agentstate' \in \Agentstate$ of the agent and a move $l' \in  \{0,\dots,d_v-1\} \cup \{\bot\}$. If $l' \in \{0,\dots,d_v-1\}$ the agent enters the edge with the local port number $l'$, whereas for $l'=\bot$ it stays at $v$. Formally, the transition function is a partial function
\begin{align*}
\delta \colon \Agentstate  \times \N \times \N &\to \Agentstate \times (\N \cup \{\bot\}),\\ 
(\agentstate, d_v, l) &\mapsto (\agentstate', l').
\end{align*}
Note that the transition function only needs to be defined for $l$ with $l<d_v$ and degrees $d_v$ that actually appear in the class of graphs considered.
It is standard to define the space requirement of an an agent with states $\Agentstate$ as $\log |\Agentstate|$ as this is the number of bits needed to encode every state, see, e.g., Cook and Rackoff~\cite{cook1980}.

\subsubsection{A single agent with pebbles}
\label{sec:prelim:agents:pebbles}
We may equip the agent $A$ with a set $P = \{1,\dots,p\}$ of unique and distinguishable pebbles. At the start of the exploration the agent is carrying all of its pebbles. As before, the agent observes in each step the degree $d_v$ of the current vertex $v$ and the port number from which $v$ was entered. In addition, the agent  has the ability to observe the set of pebbles $P_{A}$ that it carries and the set of pebbles $P_v$ present at the current vertex $v$. The transition function $\delta$ then specifies the new state $\agentstate' \in \Agentstate$ of the agent, and a move $l' \in \{0,\dots,d_v-1\} \cup \{\bot\}$ as before. In addition, the agent may drop any subset $P_{\text{drop}} \subseteq P_A$ of carried pebbles and pick up any subset of pebbles $P_{\text{pick}} \subseteq P_v$ that were located at $v$, so that after the transition the set of carried pebbles is $P_{A}' = \left( P_{A} \setminus P_{\text{drop}}\right) \cup P_{\text{pick}}$ and the set of pebbles present at $v$ is $P_v' = \left(P_v \setminus P_{\text{pick}}\right) \cup P_{\text{drop}}$. Formally, we have
\begin{align*}
\delta \colon \Agentstate \times \N \times \N \times 2^P 	\times 2^P &\to \Agentstate \times (\N \cup \{\bot\}) \times 2^P \times 2^P,\\
(\agentstate,d_v,l,P_{A},P_{v}) &\mapsto (\agentstate',l',P_{A}',P_{v}').
\end{align*}
The transition function $\delta$ is partial  as it is only defined for $P_A \cap P_v = \emptyset$.
We assume that the pebbles are actual physical devices dropped at the vertices so that no space is needed to manage the pebbles, thus, the space requirement of the agent is again $\log|\Agentstate|$.

\subsubsection{A set of agents without pebbles}
\label{sec:prelim:agents:set}
Consider a set of $k$ cooperative agents $\smash{A_1 = (\Agentstate_1,\bar{\Agentstate}_1,\delta_1,\agentstartstate_1)}$, \dots, $\smash{A_k = (\Agentstate_k,\bar{\Agentstate}_k,\delta_k,\agentstartstate_k)}$ jointly exploring the graph. 
We assume that all agents start at the same vertex $v_0$. 
In each step, all agents synchronously determine the set of agents they share a location with, as well as the states of these agents. Then, all agents move and alter their states synchronously according to their transition functions $\delta_1,\dots,\delta_k$. The transition function of agent~$i$ determines a new state $\agentstate'$ and a move $l'$ as before. Formally, let 
$$\vec \Agentstate_{-i} = (\Agentstate_1 \cup \{\bot\}) \times \cdots \times (\Agentstate_{i-1} \cup \{\bot\}) \times (\Agentstate_{i+1} \cup \{\bot\}) \times \cdots \times (\Agentstate_k \cup \{\bot\})$$
denote the states of all agents potentially visible to agent $A_i$ where a $\bot$ at position $j$ (or $(j-1)$ if $j \geq i$) stands for the event that agent $A_i$ and agent $A_j$ are located on different vertices. Then, the transition function $\delta_i$ of agent~$A_i$ is a partial function
\begin{align*}
\delta_i : \Agentstate_i \times \vec \Agentstate_{-i} \times \N \times \N &\to \Agentstate_i	\times (\N \cup \{\bot\}),\\
(\agentstate_i, \vec \agentstate_{-i},d_v,l) &\mapsto (\agentstate_i',l_i').
\end{align*}
The overall memory requirement is $\sum_{i=1}^k \log |\Agentstate_i|$.

\subsection{Relationship between agent models}

In order to compare the capability of an agent $A$ with $s$ states and $p$ pebbles to another agent $A'$ with $s'$ states and $p'$ pebbles or a set of agents $\agentset$, we use the following notion: We say that the walk of an agent $A$ is \emph{reproduced} by an agent $A'$ in a graph $G$, if the sequence of edges traversed by $A$ is a subsequence of the edges visited by $A'$ in $G$. Put differently, $A$ traverses the same edges as $A'$ in the same order, but for every edge traversal of $A$ the agent $A'$ can do an arbitrary number of intermediate edge traversals. Similarly, we say that a set of agents $\agentset$ reproduces the walk of an agent $A$ in $G$, if there is an agent $A'\in \agentset$ such that $A'$ reproduces the walk of $A$ in $G$.

We first formally show the intuitive fact that pebbles are more powerful than memory bits. 

\begin{lemma}
\label{lem:pebble>bit}
Let $A$ be an agent with $s$ states and $p$ pebbles exploring a set of graphs $\mathcal{G}$. 
Then there is an agent $A'$ with six states and $p+\lceil \log s \rceil$ pebbles that reproduces the walk of $A$ on every $G \in \mathcal{G}$ and performs at most three edge traversals for every edge traversal of $A$.
\end{lemma}

\begin{proof}
As the set of graphs $\mathcal{G}$ that can be explored by an agent with $s$ states and $p$ pebbles is non-decreasing in $s$, it suffices to show the claimed result for the case that $s$ is an integer power of two. 
Let $A = (\Agentstate, \bar{\Agentstate}, \delta, \agentstartstate)$ be an agent with a set of $p$ pebbles $P$ and $s = |\Agentstate| = 2^r$, $r \in \mathbb{N}$ states exploring all graphs $G \in \mathcal{G}$. In the following, we construct an agent $A' = (\Agentstate', \bar{\Agentstate}', \delta', {\agentstartstate}')$ with six states $\Agentstate' = \{{\agentstartstate}',\agentstate_{\text{comp}}, \bar{\agentstate}_{\text{halt}}, \agentstate_{\text{back}-1}, \agentstate_{\text{back}-2}, \agentstate_{\text{swap}}\}$, one halting state $\bar{\Agentstate}' = \{\bar{\agentstate}_{\text{halt}}\}$, and a set $P'$ of $|P'| = p+r$ pebbles. 
The general idea is to let $A'$ store the state of $A$ by dropping and retrieving the additional $r$ pebbles. 
To this end, we identify $p$ of the pebbles of $A'$ with the $p$ pebbles of $A$ and call the additional set of $r$ pebbles  $P'_{\Agentstate}$, i.e., $P'=P \cup P'_{\Agentstate}$ with $|P| = p$ and $|P'_{\Agentstate}| = r$, respectively. 
Since $|P'_{\Agentstate}|=r$ and $|\Agentstate| = s = 2^r$, there is a canonical bijection $\smash{f : \Agentstate \to 2^{P'_{\Agentstate}}}$.
Every edge traversal of agent $A$ in a state $\sigma$, will be simulated
by agent $A'$ in the computation state $\agentstate_{\text{comp}}$ while
carrying the set of pebbles $f(\sigma)$ plus the additional pebbles that $A$ is carrying. We need the additional states $\agentstate_{\text{back}-1}$, $\agentstate_{\text{back}-2}$, $\agentstate_{\text{swap}}$ to move all pebbles in $P'_{\Agentstate}$ encoding the state of $A$ to the next vertex in some intermediate steps. 

At the start of the exploration, $A'$ remains at the starting vertex and stores the starting state~$\agentstartstate$ of agent $A$ by dropping the set of pebbles $(P'_{\Agentstate} \setminus f(\agentstartstate))$.
Formally, we define the transition from the starting state~${\agentstartstate}'$ of agent $A'$ as
\begin{align*}
\delta'({\agentstartstate}', d_v, l, P', \emptyset) = (\agentstate_{\text{comp}}, \bot, f(\agentstartstate) \cup P, (P'_{\Agentstate} \setminus f(\agentstartstate)).	
\end{align*}
for all $d_v, l \in \N$. 

Next, we define the transition function $\delta'$ of $A'$ for the case that $A'$ is in its computing state $\agentstate_{\text{comp}}$, i.e., we want to simulate the change of state of $A$ and traverse the same edge. 
If $\agentstate = f^{-1}(P_{A'} \cap P'_{\Agentstate})$ is the current state of agent $A$ and agent $A$ transitions according to
\begin{align}
\label{eq:define_q_comp_1}
\delta(\agentstate, d_v, l, P_{A'} \cap P, P_v \cap P) &= (\agentstate',l',P_A',P_v')
\end{align}
with $\agentstate' \in \Agentstate$, $l' \in \N$ and $P_A', P'_v \in 2^{P'}$, then we define
\begin{align}
\label{eq:define_q_comp_2}
\delta'(\agentstate_{\text{comp}},d_v,l,P_{A'},P_v) &=
\begin{cases}
(\agentstate_{\text{comp}},\;\; l', P_A' \cup f(\agentstate'),P_v' \cup (P'_{\Agentstate} \setminus f(\agentstate')) & \text{ if } l' = \bot \text{ and } \agentstate' \notin \bar{\Agentstate},\\
%(\agentstate_{\text{halt}},\;\;\;\; l', P_A' \cup f(\agentstate'),P_v' \cup (P'_{\Agentstate} \setminus f(\agentstate')) & \text{ if } l' = \bot \text{ and } \agentstate' \in \bar{\Agentstate},\\
%(\agentstate_{\text{back}-1}, l', P_A' \cup f(\agentstate'),P_v' \cup (P'_{\Agentstate} \setminus f(\agentstate')) & \text{ else. } 
(\agentstate_{\text{back}-1}, l', P_A' \cup f(\agentstate'),P_v' \cup (P'_{\Agentstate} \setminus f(\agentstate')) & \text{ if } l' \neq \bot \text{ and } \agentstate' \notin \bar{\Agentstate},\\
(\agentstate_{\text{halt}},\;\;\;\; l', P_A' \cup f(\agentstate'),P_v' \cup (P'_{\Agentstate} \setminus f(\agentstate')) & \text{ else. } 
\end{cases}
\end{align}
Note that before and after this transition the subset of pebbles from
$P'_{\Agentstate}$ carried by $A'$ encodes the state of $A$ via the bijection $f$. 
However, if $A$ traverses an edge without entering a halting state, 
we also need to fetch the remaining pebbles from $P'_{\Agentstate}$ from the previous vertex to be able to encode the state of $A$ in the future. To this end, $A'$ switches to the state $\sigma_{\text{back}-1}$. 
The fetching will be done in three steps: First, $A'$ drops all pebbles
in  $f(\agentstate')$, moves to the previous vertex and changes its state to $\agentstate_{\text{back}-2}$. Formally, this means 
\begin{align*}
\delta'(\agentstate_{\text{back}-1}, d_v, l, P_{A'}, P_v) = \left(\agentstate_{\text{back}-2}, l, P_{A'} \setminus P'_{\Agentstate}, P_v \cup \left(P'_{\Agentstate}  \cap P_{A'} \right) \right) 
\end{align*}
for all $d_v, l \in \N$ and $P_{A'},P_v \in 2^{P'}$ with $P_{A'} \cap P_v=\emptyset$. 
Then it picks up the pebbles in  $\left(P'_{\Agentstate} \setminus f(\agentstate')\right)$, returns to the current vertex of $A$ and changes its state to
$\agentstate_{\text{swap}}$, i.e.,
\begin{align*}
\delta'(\agentstate_{\text{back}-2}, d_v, l, P_{A'}, P_v) = \left(\agentstate_{\text{swap}}, l, P_{A'} \cup \left(P'_{\Agentstate}  \cap P_v\right), P_v \setminus P'_{\Agentstate}\right)
\end{align*}
for all $d_v, l \in \N$ and $P_{A'},P_v \in 2^{P'}$ with $P_{A'} \cap P_v=\emptyset$. Lastly, agent $A'$ swaps the set of carried pebbles  $P'_{\Agentstate} \setminus f(\agentstate')$ and the set $f(\agentstate')$ of pebbles on the current vertex by performing the transition
\begin{align*}
\delta'(\agentstate_{\text{swap}},d_v,l,P_{A'},P_v) &=
\left(\agentstate_{\text{comp}}, \bot, P_{A'} \cup \left( P'_{\Agentstate} \cap P_v \right), P_v \cup \left( P'_{\Agentstate} \cap P_{A'} \right) \right)
\end{align*}
 for all $d_v, l \in \N$ and $P_{A'},P_v \in 2^{P'}$ with $P_{A'} \cap P_v=\emptyset$. 
 
A simple inductive proof establishes that the state $\agentstate$ of $A$ in every step of the exploration of a graph  $G \in \mathcal{G}$
corresponds to the set of pebbles in $P'_{\Agentstate}$ carried by $A'$ in its computation state $\agentstate_{\text{comp}}$, i.e., $\agentstate=f^{-1}\left(P_{A'} \cap  P'_{\Agentstate}\right)$. Moreover, if agent $A$ in state $\agentstate$ traverses an edge $\{v,w\}$ from a vertex $v$ to a vertex $w$ and does not move to a halting state, then $A'$ will traverse the edge $\{v,w\}$ three times and afterwards again the set of
pebbles carried by $A$ will correspond to $P_{A'}\cap P$ and the state of $A$ to  $\agentstate=f^{-1}\left(P_{A'} \cap  P'_{\Agentstate}\right)$. If $A$ remains at the same vertex
or moves to a halting state then this transition is mirrored by a single transition of agent $A'$.
In particular, agent $A'$ visits exactly the same vertices as $A$ in every graph $G \in \mathcal{G}$ while performing at most three times the number of edge traversals. 
\end{proof}

Next, we show the intuitive result that an additional agent is more powerful than a pebble.

\begin{lemma}
\label{lem:agent>pebble}
Let $A$ be an agent with $s$ states and $p$ pebbles exploring a set $\mathcal{G}$ of graphs. Then, there is a set $\mathcal{A} = (A_0,\dots,A_p)$ of $p+1$ agents,  where $A_0$ has $s$ states and all other agents have two states, that reproduce the walk of $A$ in every graph $G\in \mathcal{G}$. Moreover, for every edge traversal of $A$ each agent in $\mathcal{A}$ performs at most one edge traversal.
\end{lemma}

\begin{proof}
Let $A = (\Agentstate, \bar{\Agentstate}, \delta, \agentstartstate)$ be an agent with $|\Agentstate| = s$ and a set $P = \{1,\dots,p\}$ of $p$ pebbles exploring all graphs $G \in \mathcal{G}$. We proceed to construct a set $\mathcal{A} = \{A_0,\dots,A_p\}$ of $p+1$ agents $A_i = (\Agentstate_i, \bar{\Agentstate}_i, \delta_i, \agentstartstate_i)$, $i \in \{0,\dots,p\}$ that reproduces the walk of $A$ on all graphs $G \in \mathcal{G}$. In this construction, agent~$A_0$ represents the original agent $A$ while every agent $A_i$ for $i>0$ represents a pebble.

For agent~$A_0$, we set $\Agentstate_0 = \Agentstate$, $\bar{\Agentstate}_0 = \bar{\Agentstate}$, and $\agentstate_0^* = \agentstate^*$. For every agent $A_i$ with $i \in P$, we set $\Agentstate_i = \{c_i,d_i\}$, $\bar{\Agentstate}_i = \Agentstate_i$, and $\agentstartstate_i = c_i$. 
Intuitively, the state $c_i$ simulates that pebble~$i$ is carried and $d_i$ simulates that the pebble is dropped. In every step, we let agent
$A_0$ and the agents $A_i$ corresponding to a carried pebble do the same transitions as agent $A$. Agents that are not sharing their current vertex with $A_0$ remain at their vertex and in their state.
Let $\agentstate_{-i,j}$ for $i,j \in \{0,\ldots, p\}$ with $i \neq j$ denote the
state of agent $A_j$ visible to agent $A_i$, i.e., $\agentstate_{-i,j}=\agentstate_j$ if $A_i$ and $A_j$ share the same vertex and $\agentstate_{-i,j}=\bot$ otherwise.
Specifically, to define the transition functions $\delta_i(\agentstate_i, \vec \agentstate_{-i}, d_v, l)$ for $i \in \{0,\dots,p\}$, $\agentstate_i \in \Agentstate_i$, $\vec \agentstate_{-i} \in \vec \Agentstate_{-i}$ and $d_v,l \in \N$, we first compute
 \begin{align*}
\delta(\agentstate_0, d_v, l, \{i \in P : \agentstate_{-0,i} = c_i\}, \{i \in P : \agentstate_{-0,i} = d_i\}) = (\agentstate'_0, l', P_A', P_v')
\end{align*}
with $\agentstate'_0 \in \Agentstate$, $l' \in \N$, and $P_A', P_v' \in 2^P$. We then set
\begin{align*}
\delta_0(\agentstate_0, \vec{\agentstate}_{-0}, d_v, l) &= (\agentstate_0', l')
\intertext{
and
}
\delta_j(\agentstate_j, \vec{\agentstate}_{-j}, d_v, l) &=
\begin{cases}
(\agentstate_j,\bot) & \text{ if $\agentstate_0 = \bot$ }\\	
(c_j,l')   & \text{ if $\agentstate_0 \neq \bot$ and $j \in P_A'$}\\
(d_j,\bot)   & \text{ if $\agentstate_0 \neq \bot$ and $j \in P_v'$}
\end{cases}
\end{align*}
for all $j \in P$.

To finish the proof, fix a graph $G \in \mathcal{G}$ and consider the transitions of agent $A$ and the set of agents $\mathcal{A}$ in $G$.
A simple inductive proof shows that after $i$ transitions, the state and position of agent $A$ equals the state and position of agent $A_0$, the position of agent $A_j$ equals the position of pebble $j$ and $\sigma_j=c_j$ if and only if pebble $j$ is carried by $A$ for all $j \in P$. This implies the claim.
\end{proof}

Note that, for ease of presentation, we allow agents to make transitions even when they are in one of their halting states. We need this property in the proof above to show that two-state agents are more powerful than pebbles (cf.~Lemma~\ref{lem:agent>pebble}) in general. 
However, this reduction only needs agents to make transitions from their halting states to other halting states, and only when colocated with another agent that has not yet reached a halting state.
Furthermore, our main algorithm for single-agent exploration with pebbles that we devise in \S~\ref{sec:pebbles} has the special property that the agent $A_0$ returns to the starting vertex carrying all pebbles after having explored the graph. Thus, for our algorithm it is not necessary that agents can make transitions from halting states as we could add an additional halting state to the two-state agents to which they transition once exploration is complete and $A_0$ has returned to the starting vertex.

%!TEX root = 0_pebbles_main.tex

\section{Exploration algorithms} \label{sec:pebbles}

In this section, we devise an agent exploring any graph on at most $n$ vertices with $\bigO(\log \log n)$ pebbles and $\bigO(\log \log n)$ memory. By the reductions between the agents' models given in \S~\ref{sec:prelim} this implies that (a) an agent with $\bigO(\log \log n)$ pebbles and constant memory can explore any $n$-vertex graph and (b) that a set of $\bigO(\log \log n)$ agents with constant memory each can explore any $n$-vertex graph.

For our algorithm, we use the concept of \emph{exploration sequences} (Kouck{\'{y}} \cite{koucky02}). An exploration sequence is a sequence of integers $e_0,e_1,e_2,\ldots$ that guides the walk of an agent
through a graph~$G$ as follows: 
Assume an agent starts in a vertex $v_0$ of $G$ and let $l_0=0$.
Let~$v_i$ denote the agent's location in step $i$ and~$l_i$ the port number of the edge at~$v_i$ leading back to the previous location.
Then, the agent \emph{follows} the exploration sequence $e_0,e_1, e_2,\dots$ if, in each step~$i$, it traverses the edge with port number ${(l_i+e_i) \bmod d_{v_i}}$
at~$v_i$ to the next vertex $v_{i+1}$,  where $d_{v_i}$ is the degree of~$v_i$. 
An exploration sequence is \emph{universal} for a class of undirected, connected, locally edge-labeled 
graphs $\mathcal{G}$ if an agent following it explores every graph $G \in \mathcal{G}$ for any starting vertex in $G$, i.e., for any starting vertex it visits all vertices of $G$. For a set $M$, we further use the notation $M^*:=\bigcup_{i=1}^\infty M^i$ to denote the set of finite sequences with elements in $M$.
The following fundamental result of Reingold~\cite{reingold08} establishes that universal exploration sequences can be constructed in logarithmic space.

\begin{theorem}[\cite{reingold08}, Corollary 5.5]\label{theo_reingold}
There exists an algorithm taking $n$ and $d$ as input and producing in $\bigO(\log n)$ space an exploration sequence universal for all connected $d$-regular graphs on $n$ vertices.
\end{theorem}

Reingold's result implies in particular that there is an agent without pebbles and~$\bigO(n^c)$ states for some constant $c$ that explores any $d$-regular graph with $n$ vertices when both $n$ and $d$ are known. We further note that Reingold's algorithm can be implemented on a Turing machine that has a read/write tape of length $\bigO(\log n)$ as work tape and writes the exploration sequence to a write only output tape, see \cite[\S~5]{reingold08} for details. 
For formal reasons the Turing machine in~\cite{reingold08} additionally has a read-only input tape from which it reads the values of $n$ and $d$ encoded in unary so that the space complexity of the algorithm is actually logarithmic in the input length. 
For our setting, it is sufficient to assume that $n$ and $d$ are given as binary encoded numbers on the working tape of length $\bigO(\log n)$, as we care only about the space complexity of exploration in terms of the number of vertices $n$. 

 As a first step, we show in Lemma~\ref{lem_uxs_closed_3_reg_graph} how to modify Reingold's algorithm for 3-regular graphs to yield a closed walk containing an exponential number of vertices in terms of the memory used. Afterwards, we extend this result to general graphs in Lemma~\ref{lemma_uxs_general_graph}.

\begin{lemma}\label{lem_uxs_closed_3_reg_graph}
For any $z \in \N$, there exists a $\bigO(\log z)$-space algorithm producing an exploration sequence 
${w\in \{0,1,2\}^*}$ such that for all connected 3-regular graphs $G$ with $n$ vertices the following hold:
\begin{enumerate}
\item an agent following $w$ in $G$ explores at least $\min\{z,n\}$ distinct vertices,
\item $w$ yields a closed walk in $G$,
\item the length of $w$ is bounded by $z^{\bigO(1)}$.
\end{enumerate}
\end{lemma}

\begin{proof}
By  Theorem~\ref{theo_reingold}, there is a Turing machine~$\turing_0$ with a tape of length $\bigO( \log z)$ producing a universal exploration sequence $e_1,e_2,\ldots $ for any 3-regular graph on exactly $4 z$~vertices. 
Let $c_{\turing_0}$ be the number of configurations of $\turing_0$ and $a := 12  z   c_{\turing_0}+1$. 
Here the number of configurations of $\turing_0$ is the number of possible combinations of Turing state, tape contents and head position of~$\turing_0$.

\begin{algorithm}[tb]
\floatname{algorithm}{Algorithm}
\caption{Turing machine $\turing$ computing exploration sequence for 3-regular graphs.}
\label{alg_tm_log_space_3_regular}

\begin{algorithmic}
\For{$t \in \{1, \ldots, 2 a \}$} 
	\If{ $t \leq a$}
		\State  run $\turing_0$ for $t$ steps to obtain element $e_t$ of the exploration sequence generated by~$\turing_0$
		\State output $e_t$
	\ElsIf{$t=a+1$}
		\State output $0$
	\ElsIf{$t \geq a+2$}
		\State  run $\turing_0$ for $2a+2-t$ steps to obtain  element $e_{2a+2-t}$ of  exploration sequence of $\turing_0$
		\State output $-e_{2s+2-t} \bmod 3$
	\EndIf
\EndFor
\end{algorithmic}
\end{algorithm}

The Turing machine $\turing$ producing an exploration sequence $w$ with the desired properties is given in Algorithm~\ref{alg_tm_log_space_3_regular}. 
By construction, the sequence $w$ produced by $\turing$ is
\begin{align*}
e_1, e_2, \ldots, e_a, 0, (-e_a \bmod 3), (-e_{a-1} \bmod 3), \ldots, (- e_2 \bmod 3).
\end{align*}
We first show that this sequence corresponds to a closed walk in any 3-regular graph. 
Let an agent $A$ start at a vertex~$v_0$ in some graph 3-regular $G$, follow the exploration sequence $w$, and, for $i \in \{1,\ldots, a\}$, let $v_i$ be the vertex reached after following $w$ up to $e_i$.
Then the offset 0 takes the agent back from $v_a$ to $v_{a-1}$ and afterwards $-e_i \bmod 3$ takes agent $A$ from $v_{i-1}$ to $v_{i-2}$. 
Thus, at the end the agent returns to $v_0$, which yields the second claim. 

Moreover, the number of configurations $c_{\turing_0}$ of the Turing machine $\turing_0$, i.e., the number of possible combinations of state, head position, and tape contents, is bounded by $z^{\bigO(1)}$, because the working tape has length $\bigO(\log z)$. 
Hence, the length of $w$, i.e., $2a= 2 \cdot (12  z  c_{\turing_0}+1)$, is also bounded by $z^{\bigO(1)}$, which yields the third claim.
As the auxiliary variable $t$ ranges from $1$ to~$2a$ and running the Turing machine $\turing_0$ for $t$ steps can be implemented in $\bigO(\log z)$ space, the Turing machine $\turing$ can be implemented to run in $\bigO(\log z)$ space.

It is left to show is that an agent following $w$ in an arbitrary connected 3-regular graph with~$n$ vertices explores at least $\min\{z,n\}$ vertices.
For the sake of contradiction, assume there exists some 3-regular graph~$G$ on~$n$~vertices so that an agent~$\agent$ starting in a vertex~$v_0$ and following the exploration sequence $w$ produced by $\turing$ only visits a set of vertices $V_0$ with $|V_0|< \min\{z,n\}$. 
Let~$G_0$ be the subgraph of $G$ induced by $V_0$. 
Note that, since $|V_0|<n$ by assumption, at least one vertex in $G_0$ has degree less than 3. 
We now extend $G_0$ to a connected 3-regular graph with~$4z$ vertices as follows. 
First, we let $G_1$ be the graph $G_0$ after adding
an isolated vertex if $V_0$ is odd and we let~$V_1$ be the vertex set of $G_1$.
We further let $G_2$ be a cycle of length $4z-|V_1|$ with opposite vertices connected by an edge. 
Note that $4z$ and $|V_1|$ are even and $G_2$ is 3-regular. 
As long as $G_1$ contains at least one vertex of degree less than 3, we delete an edge $\{w,w'\}$ connecting opposite vertices in the cycle in $G_2$ and for $w$ and then $w'$ add an edge from this vertex to a vertex of degree less than 3 in $G_1$ (possibly the same). 
This procedure terminates when all vertices in~$G_1$ have degree~3, since $G_2$ contains $4z-|V_1|\geq 3 z \geq 3 |V_1|$ vertices and there cannot be a single vertex of degree~2 left in~$G_1$, as this would mean that the sum of all vertex degrees in~$G_1$ is odd.
 The labels in~$\{0,1,2\}$ at both endpoints of every edge not in~$G_0$ are chosen arbitrarily.
Let~$H$ be the resulting 3-regular graph with $4z$ vertices containing~$G_0$ as induced subgraph. 

By construction, the walk of an agent~$\agent$ starting in~$H$ at~$v_0$ and following~$w$ is the same as the walk in $G$ starting in~$v_0$ and following $w$. In particular, 
the agent~$\agent$ does not explore $H$. Let now $\agent_0$ be an agent following the exploration sequence $w_0$ produced by $M_0$ starting in vertex $v_0$ in $H$. 
As the first $a$ values of $w$ and $w_0$ coincide, the walk of agent~$\agent_0$ in $H$ up to step~$a$ is the same as that of agent~$\agent$. Recall that $a=3 \cdot 4  z   c_{\turing_0}+1$.
This implies that in the first $a$ steps there must be a vertex $v$ in $H$ visited twice
by agent $\agent_0$ (there are $4 z$ vertices in $H$) and in both visits, the label to the previous vertex (there are 3 possible labels) is the same and the Turing machine $\turing_0$ producing the exploration sequence $w_0$ is in the same configuration (there are $c_{\turing_0}$ possible configurations) in both visits. But this implies that the behaviour of $\agent_0$ in $H$ becomes periodic
and it only visits the set of vertices already visited in the first $a$ steps, i.e., the set of vertices $V_0$. 
We conclude that $\agent_0$ does not explore $H$, contradicting that $w_0$ is a universal exploration sequence for all 3-regular connected graphs on $4 z$ vertices.
\end{proof}

We proceed to give a similar result for non-regular graphs.

\begin{lemma}\label{lemma_uxs_general_graph}
For any $z \in \N$, there exists a $\bigO(\log z)$-space algorithm producing an exploration sequence 
${w\in \{-1,0,1\}^*}$ such that for all connected graphs $G$ with $n$ vertices the following hold:
\begin{enumerate}
\item an agent following $w$ in $G$ explores at least $\min\{z,n\}$ distinct vertices,
\item $w$ yields a closed walk in $G$,
\item the length of $w$ is bounded by $z^{\bigO(1)}$.
\end{enumerate}
\end{lemma}

\begin{figure}
\centering
%\fbox{
\begin{subfigure}[b]{0.2\linewidth}
\centering
\tikzstyle{graphnode}=[scale=0.8,circle,draw,minimum size=20pt]
\tikzstyle{dnode}=[scale=0.8]
\newcommand*{\x}{1.1}
\newcommand*{\y}{1.1}
\begin{tikzpicture}[scale=1.3]
\node (v) at (0,0)[graphnode] {$v$};
\node (w) at (1,0)[graphnode] {$w$};
\draw (v) -- node [above,pos=0.2, dnode] {$1$} node [below,pos=0.8, dnode] {$0$}(w);
\draw (v) -- (-0.6,0);
\draw[dotted] (v) -- (-0.8,0);
\draw (w) -- ($(1,0)+(60:0.6)$);
\draw[dotted] (w) -- ($(1,0)+(60:0.8)$);
\draw (w) -- ($(1,0)+(-60:0.6)$);
\draw[dotted] (w) -- ($(1,0)+(-60:0.8)$);
\end{tikzpicture}
~\\[1.2cm]
\subcaption{Original graph $G$.}
\end{subfigure}%}
\hspace{0.05\linewidth}
%\fbox{
\begin{subfigure}[b]{0.70\linewidth}
\centering
\tikzstyle{graphnode}=[scale=0.8,circle,draw,minimum size=20pt]
\tikzstyle{dnode}=[scale=0.8]
\newcommand*{\x}{0.9}
\newcommand*{\y}{1.15}
\newcommand*{\z}{3.5}
\begin{tikzpicture}[scale=1.3]
  [auto=left]
\foreach \i in {0,...,5}{
	\node(v\i) at (-\i*60+60:\x) [graphnode]{};
	\node(d\i) at (-\i*60+60:\x) [dnode]{$v,\i$};
	\ifthenelse{\i=3 \OR \i=5}{}{
		\node[dnode] at (-\i*60+65:\x+0.35) {2};
	}%endif
	\ifthenelse{\i=3}{
		\node[dnode] at (-\i*60+70:\x+0.30) {2};
	}{};
	\ifthenelse{\i=5}{
		\node[dnode] at (-\i*60+50:\x+0.30) {2};
	}{};
	\ifthenelse{\i=1 \OR \i=3 \OR \i=5}{}{
		\draw (v\i) -- (-\i*60+60:\x+0.6);
		\draw[dotted] (v\i) -- (-\i*60+60:\x+0.8);
	}%endif
	\node[dnode] at (-\i*60+80:\x) {1};
	\node[dnode] at (-\i*60+104:\x-0.24) {0};
}%endfor
\foreach \i [count=\j] in {0,...,4}{
	\draw (v\i) -- (v\j);
}%endfor
\draw (v5) -- (v0);
\foreach \i in {0,...,8}{
	\node(w\i) at ($(\i*40+180:\y)+(\z,0)$) [graphnode]{};
	\node(dd\i) at ($(\i*40+180:\y)+(\z,0)$) [dnode]{$w,\i$};
	\ifthenelse{\i=3}{
		\node[dnode] at ($(\i*40+172:\y+0.30)+(\z,0)$) {2};
	}{
		\node[dnode] at ($(\i*40+185:\y+0.35)+(\z,0)$) {2};
	}%endif
	\ifthenelse{\i=0 \OR \i=3 \OR \i=6}{}{
		\draw (w\i) -- ($(\i*40+180:\y+0.6)+(\z,0)$);
		\draw[dotted] (w\i) -- ($(\i*40+180:\y+0.8)+(\z,0)$);
	}%endif
	\node[dnode] at ($(-\i*40+75:\y+0.05)+(\z,0)$) {0};
	\node[dnode] at ($(-\i*40+86:\y+-0.20)+(\z,0)$) {1}; 
}%endfor
\foreach \i [count=\j] in {0,...,7}
{
	\draw (w\i) -- (w\j);
}
\draw (w8) -- (w0);

\draw (v1) -- (w0);
%\draw (v5) .. controls ($(v5)+(0,\x+0.8)$) and ($(w6)+(0,\x+0.6)$).. (w6);
\draw[rounded corners] (v5) -- ($(v5)+(60:1.35)$) -- ($(w6)+(120:1.1)$) -- (w6);
\draw[rounded corners] (v3) -- ($(v3)+(-60:1.35)$) -- ($(w3)+(-120:1.1)$) -- (w3);
\end{tikzpicture}
\subcaption{$3$-regular graph $G_{\text{reg}}$.}
\end{subfigure}%}
\caption{Example for the transformation of a graph $G$ to a 3-regular graph $G_\text{reg}$. A vertex $v$ of degree~2 is transformed to a cycle containing 6 vertices and for the edge $\{v,w\}$, three edges are added to the graph.}
\label{fig_transformation_to_3_regular}
\end{figure}
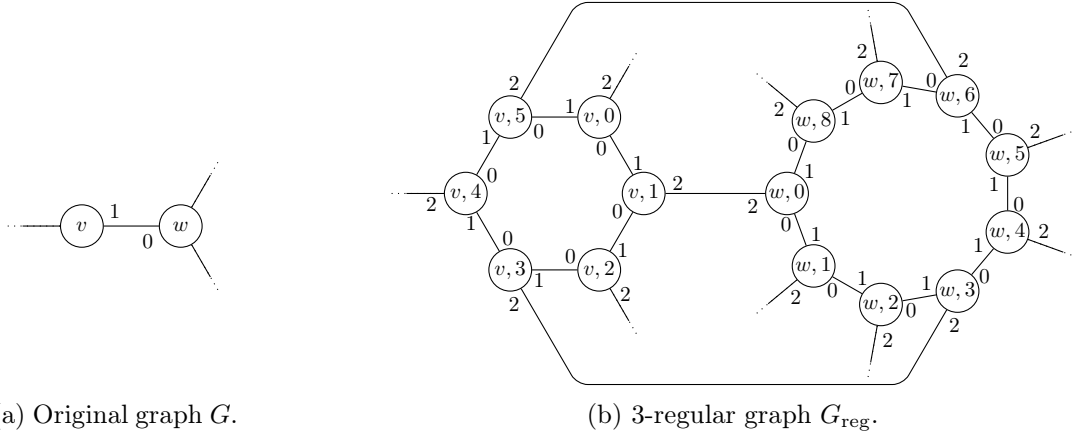

\begin{proof}
Let $\turing_{\text{reg}}$ be the Turing machine of Lemma~\ref{lem_uxs_closed_3_reg_graph}
 with a tape of length bounded by $\bigO( \log z)$ producing a universal exploration sequence $w_{\text{reg}}\in \{0,1,2\}^*$ such that
an agent following $w_{\text{reg}}$ 
 in some 3-regular graph with $n$ vertices visits at least $\min\{3 z^2, n\}$ distinct vertices.

To prove the statement, we transform this universal exploration sequence for 3-regular graphs to a universal exploration sequence universal for general graphs by using a construction taken from Kouck\'y~\cite[Theorem 87]{koucky03}. In this construction,
an arbitrary graph~$G$ with $n$~vertices is transformed into a 3-regular graph $G_\text{reg}$ as follows:
We replace every vertex $v$ of degree $d_v$ by a circle of $3 d_v$ vertices $(v,0),\ldots, (v,3d_v-1)$, where the edge $\{(v,i), (v,i+1 \bmod 3d_v)\}$ has port number~0 at $(v,i)$ and port number~1 at $(v,i+1 \bmod 3d_v)$, see also Figure~\ref{fig_transformation_to_3_regular} for an example of this construction. For any edge $\{v,w\}$ in $G$ with port number $i$ at $v$ and $j$ at $w$, we add the three edges $\{(v,i),(w,j)\},\{(v,i+d_v),(w,j+d_w)\},\{(v,i+2d_v),(w,j+2d_w)\}$ with port numbers~2 at both endpoints to $G_\text{reg}$. 

\begin{algorithm}[tb]
\floatname{algorithm}{Algorithm}
\caption{Turing machine $\turing$ computing exploration sequence for arbitrary graphs.}
\label{alg_tm_log_space}

\begin{algorithmic}[1]
\State output $0,0$ 
\State $i  \gets 0$
\While{$\turing_{\text{reg}}$ has not terminated} 
	\State obtain next offset $w_{\text{reg}}(i)$ from $\turing_{\text{reg}}$
	\State compute edge label $l_i$ in $G_\text{reg}$ \label{line:compute_label}
	\If{$l_i = 0$}
		\State output $1,0$ \label{line:sequence10}
	\ElsIf{$l_i = 1$}
		\State output $-1,0$ \label{line:sequence-10}
	\ElsIf{$l_i = 2$}
		\State output $0$ \label{line:sequence0}
	\EndIf
	\State $i  \gets i+1$
\EndWhile
\end{algorithmic}
\end{algorithm}

Observe that there are only two labelings of edges in $G_\text{reg}$, edges with port number 2 at both endpoints and edges with port numbers 0 and 1. In particular,
one port number of an edge can be deduced from the other port number. As a consequence, 
given the previous edge label and the edge offsets from the exploration sequence~$w_{\text{reg}}$ produced by $\turing_{\text{reg}}$, the next edge label can be computed without knowing the edge label of the edge by which the vertex was entered. In other words, we can transform the sequence of edge label offsets given by~$w_{\text{reg}}$ to a traversal sequence, i.e., a sequence of absolute edge labels $l_0,l_1,\ldots$ of $G_\text{reg}$.

We proceed to define the Turing machine
 $\turing$ producing an  exploration sequence $w\in\{-1,0,1\}^*$ with the desired properties as shown in Algorithm~\ref{alg_tm_log_space}.
First of all, note that the next edge label $l_i$ in~$G_\text{reg}$ can be computed from the last edge label in~$G_\text{reg}$ and the offset $w_{\text{reg}}(i)$ in constant space (line~\ref{line:compute_label} of Algorithm~\ref{alg_tm_log_space_3_regular}). 
Thus, $\turing$ can be implemented in $\bigO(\log z)$ space.
By assumption, the length of the exploration sequence produced by  $\turing_{\text{reg}}$ is bounded by $z^{\bigO(1)}$. Hence, also the length of the exploration sequence produced by $\turing$ is bounded by $z^{\bigO(1)}$.

Let $A_\text{reg}$ be an agent following $w_{\text{reg}}$  in  $G_\text{reg}$ and $A$ be an agent following the exploration sequence~$w$ produced by $\turing$ in $G$. 
What is left to show is that $A$ traverses $G$ in a closed walk and visits at least $\min\{z,n\}$ distinct vertices.
In order to show this, we first establish the following invariants that hold after every iteration $i$ of the while-loop in  Algorithm~\ref{alg_tm_log_space}:
\begin{enumerate}
\item If agent $A_\text{reg}$ is at vertex $(v_i,a_i)$ in $G_\text{reg}$ after $i$ steps, then  after following the exploration sequence output by $\turing$ up to the end of iteration $i$ agent $A$ is at $v_i$ and ${a_i \bmod d_{v_i}}$ is the label of the edge to the previous vertex.
\item If $(v_i,a_i)$ is visited by  $A_\text{reg}$ in $G_\text{reg}$, then in $G$ both $v_i$ and the neighbor incident to the edge with label~${(a_i\bmod d_{v_i})}$ are visited by $A$. 
\end{enumerate}
We show the invariants by induction. 
The starting vertex of $A_\text{reg}$ in $G_\text{reg}$ is $(v_0,0)$ and the starting vertex of $A$ in $G$ is $v_0$. 
Note that at the beginning the Turing machine $\turing$ outputs 0,0 so that in $G$ agent $A$ visits the neighbor of $v_0$ incident to the edge $0$ and then returns to $v_0$. 
Thus, both invariants hold before the first iteration of the while-loop.

Now assume that before iteration $i$ both invariants hold. We show that then they also hold after iteration $i$. 
If agent $A_\text{reg}$ is at the vertex $(v,a)$ after $i-1$ steps and the edge traversed by $A_\text{reg}$ in step~$i$ has label 0, i.e., $l_i=0$, then $A_\text{reg}$ moves to vertex $(v, (a+1) \bmod 3 d_v)$ 
by the definition of $G_{\text{reg}}$, see also Figure~\ref{fig_transformation_to_3_regular}. By assumption, agent $A$ is at vertex $v$
in $G$ and the last edge label is $a \bmod d_v$. Thus, if agent $A$ follows the exploration sequence $1,0$ output by $\turing$ in iteration $i$ (line~\ref{line:sequence10}~of~Algorithm~\ref{alg_tm_log_space}), then it first traverses the edge labeled $(a+1) \bmod d_v$ and then returns to $v$. This means that after iteration $i$, the current vertex of $A$ in $G$ is $v$ and the edge label to the previous vertex is $(a+1) \bmod d_v =  ( (a+1) \bmod 3 d_v) \bmod d_v$.
Moreover, agent $A$ visited both $v$ and the neighbor of $v$ incident to the edge with label $(a+1) \bmod d_v$. 
Thus, both invariants hold after iteration $i$ in this case.

The case that $l_i=1$ is analogous except that edges with label $l_i=1$ in $G_{\text{reg}}$ lead from a vertex $(v,a)$ to a vertex $(v, (a-1) \bmod 3d_v)$. The equivalent movement of $A$ in $G$ is achieved by the sequence $-1,0$ (line~\ref{line:sequence-10} in Algorithm~\ref{alg_tm_log_space_3_regular}). 

So assume that agent $A_\text{reg}$ in step $i$ traverses an edge with label $l_i=2$ from a vertex $(v,a)$ to a vertex $(v',a')$. This means that there is an edge $\{v,v'\}$ in $G$ with port number $a \bmod d_v$ at $v$ and port number $a' \bmod d_{v'}$ at $v'$. By assumption, at the beginning of iteration $i$ agent $A$ is at $v$ and $a \bmod d_v$ is the label of the edge to the previous vertex. So if $A$ follows the exploration sequence~$0$ output in iteration $i$ (line~\ref{line:sequence0} of Algorithm~\ref{alg_tm_log_space}), then it moves to $v'$. Now the label to the previous vertex at $v'$ is $a' \bmod d_{v'}$ and $A$ visited both $v$ and $v'$ so that both invariants hold again.

Finally, for the second property in the lemma, we know that  the traversal of agent $A_{\text{reg}}$ in $G_{\text{reg}}$ is a closed walk by Lemma~\ref{lem_uxs_closed_3_reg_graph} and hence
the traversal of $A$ in $G$ also is a closed walk by the first invariant. 

What is left to show is that $A$ visits at least $\min\{z,n\}$ distinct vertices in $G$. If  $G_{\text{reg}}$ has at most~$3 z^2$ vertices, then $A_{\text{reg}}$
visits all vertices in $G_{\text{reg}}$ by assumption and thus
  $A$ also visits all vertices in $G$ by the second invariant. Otherwise,
  we know that $A_{\text{reg}}$ visits at least $3 z^2$ distinct vertices
in~$G_{\text{reg}}$. Note that this implies $z<n$ as $G_{\text{reg}}$
contains at most $3 n (n-1)$ vertices. 
  
  Assume, for the sake of contradiction, that $A$ visits less than $z$ vertices in $G$.
Let $\bar{V}_{\text{reg}}$ be the set of vertices visited by $A_{\text{reg}}$ in $G_{\text{reg}}$.
As $|\bar{V}_{\text{reg}}|\geq 3 z^2$ by assumption, at least one
of the two following cases occurs:
\begin{enumerate}
\item The cardinality of $\bar{V}:=\{ \ v  \mid 
(v,j)\in \bar{V}_{\text{reg}} \text{ for some } j\ \}$ is at least $z$.
\item There is a vertex $\bar{v}$ in $G$ such that $M_{\bar{v}}:=\{ \ j  \mid 
(\bar{v},j)\in \bar{V}_{\text{reg}}\ \}$ has cardinality $\geq 3 z$.
\end{enumerate}
We show that both cases lead to a contradiction.

Note that by the second invariant agent $A$ visits all vertices in $\bar{V}$. 
Thus, if $|\bar{V}|\geq z$, then $A$ visits at least $z$ distinct vertices in $G$, a contradiction.

Assume the second case occurs and let $\bar{v}$ in $G$ be a vertex such that $|M_{\bar{v}}|\geq 3 z$. 
Then we have $|\{ j \mod d_{\bar{v}} \mid 
j \in M_{\bar{v}} |\geq z$ implying that agent $A$ visits at least $z$ neighbors of $\bar{v}$ in $G$ by the second invariant. This again is a contradiction. 
\end{proof}

To make the results above usable for our agents with pebbles, we need more structure regarding the memory usage of the agent. 
To this end,  we formally define a walking Turing machine with access to pebbles which we will refer to as a \emph{pebble machine}.
Formally, we can view such a walking Turing machine as a specification of the general agent model with pebbles described in \S~\ref{sec:prelim:agents:pebbles}, where the states of the agent correspond the state of the working tape, the position of the head, and the state of the Turing machine.

\begin{definition}\label{def-pebble-machine}
Let $s,p,m \in \N$. 
An \emph{$(s,p,m)$-pebble machine} $T = (Q,\bar{Q}, P, m, \delta_{\text{in}}, \delta_{\text{TM}}, \delta_{\text{out}}, q^*)$ is an agent $A = (\Agentstate, \bar{\Agentstate}, \delta, \agentstartstate)$ with a set $P = \{1,\dots,p\}$ of $p$ pebbles and the following properties:
\begin{enumerate}
	\item The set of states is $\Agentstate = Q \times \{0,1\}^m \times \{0,\dots,m-1\}$, where each state consists of a Turing state, the state of the working tape of length~$m$, and a head position on the tape.
	%consists of a set $Q$ or Turing states, and the state of a working tape of length $m$ with tape alphabet $\{0,1\}$, and a head position $0,\dots,m-1$ on the tape;
	\item In the initial state $\agentstartstate$ the Turing state is $q^*$, the head position is $0$, and the tape has $0$ at every position.
	\item The agent's transition function $\delta\colon \Agentstate \times \N \times \N \times 2^P \times 2^P \to \Agentstate \times (\N \cup \{\bot\}) \times 2^P \times 2^P$ is computed as follows:
	\begin{enumerate}
	\item The agent first observes its local environment according to the function $\delta_{\text{in}}\colon Q \times \N \times \N \times 2^P \times 2^P \to Q$ that maps a vector $(q,d_v,l,P_A,P_v)$ consisting of the current Turing state, the degree $d_v$ of the current vertex, the label $l$ of the edge leading back to the vertex last visited, the set $P_A$ of carried pebbles and the set $P_v$ of pebbles located at the current vertex to a new Turing state $q'$.
	\item The agent does computations on the working tape like a regular Turing machine according to the function $\delta_{\text{TM}}\colon Q \times \{0,1\} \to Q \times \{0,1\} \times \{\mathsf{left}, \mathsf{right}\}$ that maps the tuple consisting of the current Turing state and the symbol at the current head position $(q,a)$ to a tuple $(q',a',d)$ meaning that the machine transitions to the new state $q'$, writes $a'$ at the current position of the head and moves the head in direction $d$; this process is repeated until a halting state $\bar{q} \in \bar{Q}$ is reached (note that we only consider Turing machines that eventually halt).
\item It performs actions according to the function $\delta_{\text{out}} \colon \bar{Q} \times 2^P \times 2^P \times \N \times \N  \to 2^P \times 2^P \times \N$ that maps a tuple $(q,P_A,P_v)$ containing the current Turing state $q$, the set of carried pebbles~$P_A$ and the set of pebbles $P_v$ at the current vertex $v$ to a tuple $(P_A',P_v',l')$ meaning that it drops and retrieves pebbles such that it carries $P_A'$, leaves $P_v'$ at $v$ and takes the edge locally labeled by $l'$.
	\end{enumerate} 
\end{enumerate}	
\end{definition}

When considering a pebble machine $T =(Q, \bar{Q}, P, m, \delta_{\text{in}}, \delta_{\text{TM}}, \delta_{\text{out}})$ we will call the Turing states~$Q$ simply \emph{states} and we will call the set of states $\Agentstate$ of the underlying agent model \emph{configurations}.
As the configuration of a pebble machine is fully described by the (Turing) state $q  \in Q_T$, the head position, and the state of the working tape, it has $sm2^m$ configurations. We further call a transition of the agent according to the transition function $\delta_{\text{TM}}$ a \emph{computation step}. Note that an agent remains at the same vertex and only changes its configuration when performing a computation step.

 In the following theorem, we explain how to place pebbles on a closed walk and use them as additional memory.

\begin{theorem}\label{theo_pebble_machine_simulation}
There are constants $c, c' \in \N$, such that for every $(s,p,2m)$-pebble machine $\pebblemachine$
 there exists a $(c s ,p+c,m)$-pebble machine $\pebblemachine'$ with the following properties:
 \begin{enumerate}
 \item  For every graph $G$ with $n< 2^{m/{c'}}$ vertices, the pebble machine $T'$ explores
 $G$ in a closed walk, collects all pebbles, returns to the starting vertex
 and terminates. The overall number of edge traversals and computation steps needed by the pebble machine~$\pebblemachine'$ is bounded by $2^{\bigO(m)}$.
 \item For every graph $G$ with $n\geq 2^{m/{c'}}$ vertices, $\pebblemachine'$ reproduces the walk of $\pebblemachine$ in $G$ while the positions of $p$ of the $p+c$ pebbles correspond to the positions of the $p$ pebbles of $T$. For the initialization, $\pebblemachine'$ needs  $2^{\bigO(m)}$ edge traversals and computations steps. Afterwards, the number of edge traversals and computation steps needed by the pebble machine~$\pebblemachine'$ to reproduce one edge traversal or computation step of $\pebblemachine$ is bounded by $2^{\bigO(m)}$. 
 \end{enumerate}
\end{theorem}

\begin{proof}
The general idea of the proof is that $\pebblemachine'$ places the constant number of additional pebbles 
on a closed walk $\omega$ in order to encode the tape content of the pebble machine $\pebblemachine$.
Using these pebbles, $\pebblemachine'$ can also count the number of distinct vertices on the closed walk $\omega$. If the closed
walk is too short, then  $\pebblemachine'$ already explored the graph and the first case occurs. Otherwise,
the closed walk is long enough to allow for a sufficient number of distinct positions of the pebble and we are in the second case of the statement of the theorem.

Let $Q$ be the set of states of $\pebblemachine$. 
We define the set of states of $\pebblemachine'$ to be $Q \times Q'$ for a set~$Q'$, i.e., every state of $\pebblemachine'$ is a tuple $(q,q')$,
where~$q$ corresponds to the state of $\pebblemachine$ in the current step of the traversal.
The pebble machine $\pebblemachine'$ observes the input according to $\delta_\text{in}$ and performs actions according to $\delta_\text{out}$ just as~$\pebblemachine$, while only changing the first component of the current state. $\pebblemachine'$ uses $p$ pebbles in the same way
as $\pebblemachine$ and possesses a set $\{p_\text{start},p_\text{temp},p_\text{next}, p_0,p_1,\ldots,p_{c-4}\}$ of additional pebbles. The pebble $p_\text{start}$ is dropped by $\pebblemachine'$ right after observing the input according to $\delta_\text{in}$ in order to mark the current location of $\pebblemachine$ during the traversal. 
The purpose of the pebbles $p_\text{temp}$ and $p_\text{next}$ will be explained later. 
The other pebbles $\{p_0,p_1,\ldots,p_{c-4}\}$ are placed along a closed walk~$\omega$ to simulate the memory of $\pebblemachine$, while the states $Q'$ and the tape of $\pebblemachine'$ are used to manage this memory.

To this end, we divide the tape of $\pebblemachine'$ into a constant number $c_0$ of blocks of size~$m/c_0$ each.  In the course of the proof, we will introduce a constant number of variables to manage the simulation of the memory of $\pebblemachine$ with pebbles. Each of these variables is stored in a constant number of blocks. The constant $c_0$ is chosen large enough to accommodate all variables on the tape of $\pebblemachine'$.
By Lemma~\ref{lemma_uxs_general_graph}, there is a constant $c_1$ such that for any $r \in \N$ there is a Turing
machine $\turing$ with at most $c_1$ states and a tape of length $c_1 \cdot r$ outputting an exploration
sequence that gives a closed walk of length at most $2^{c_1 \cdot r}$ visiting at least $\min\{2^r,n\}$~vertices in any graph with $n$ vertices.
Let $m_1 :=   m/(c_0 c_1)  $ and let $m_0\in \N$ be such that for all $m' \in \N$ with $m' \geq m_0$ we have
$ c_1 \leq 2^{m'/c_0}$ and $2^{m'/c_0} > 2m'$.

In the following, we show how the simulated memory is managed by providing algorithms in pseudocode (see Algorithms \ref{alg_moving_along_walk_aux} to \ref{alg_read_write_tape}).  These can be implemented on a Turing machine with a constant number of states $c_\text{Alg}$. Let $\smash{c = \max\bigl\{2^{2 m_0}, 2 c_0 c_1+3,c_\text{Alg}\bigr\}}$ and $c':=c_0 c_1$.
 Note that $c$ only depends on the constants $c_0$, $c_1$ and $c_\text{Alg}$, but not on $m$ or $p$. It is without loss of generality to assume $m \geq m_0$, because, for $m <m_0$, we can store the configuration of the tape of $\pebblemachine$ in the states $Q'$ of $\pebblemachine'$, since $c \geq 2^{2 m_0}$. 

We proceed to show that the computations on the tape of length~$2m$ performed by $\pebblemachine$
according to the transition function~$\delta_{\text{TM}}$ can be simulated using the pebbles~$\{p_\text{start},p_\text{temp},p_0,p_1,\ldots,p_{c-4}\}$. The proof of this result proceeds along the following key claims. 

\begin{enumerate}
\item We can find a closed walk $\omega$ containing $2^{m_1}$ distinct vertices so that $c-3$ pebbles placed along this walk can encode all configurations of the tape of~$\pebblemachine$.
\item We can move along $\omega$ while keeping track of the number of steps and counting the number of distinct vertices until we have seen $2^{m_1}$ distinct vertices.
\item We can read from and write to the memory encoded by the placement of the pebbles along~$\omega$. 
\item If the closed walk $\omega$ starts at vertex $v$ and~$\pebblemachine$ moves from vertex $v$ to vertex $v'$, we can move all pebbles to a closed walk $\omega'$ starting in $v'$ while preserving the content of the memory.
\end{enumerate}

\begin{figure}
\centering
\tikzstyle{dedge}=[->]
\newcommand*{\diffX}{1.6}
\newcommand*{\diffY}{1.2}
\tikzstyle{dnode}=[scale=0.8]
\begin{subfigure}[b]{0.45\linewidth}
\centering
\begin{tikzpicture}
\renewcommand{\diffX}{1.5}
\renewcommand{\diffY}{1.0}
\draw(-0.5*\diffX,-2*\diffY)--(3.5*\diffX,-2*\diffY);
\draw(-0.5*\diffX,-2.5*\diffY)--(3.5*\diffX,-2.5*\diffY);
\foreach \i in {0,...,12}
	\draw(-0.5*\diffX+0.333*\i*\diffX,-2*\diffY)--(-0.5*\diffX+0.333*\i*\diffX,-2.5*\diffY);
\foreach \i in {0,3,6,9,12}
	\draw(-0.5*\diffX+0.333*\i*\diffX,-1.7*\diffY)--(-0.5*\diffX+0.333*\i*\diffX,-2.5*\diffY);
\foreach \i in {0,...,3}
	\node(p) at (\i*\diffX,-1.7*\diffY)[dnode]{$p_\i$};
\foreach \i in {0,2,5,7,8,9}
	\node(n) at (-0.333*\diffX+0.333*\i*\diffX,-2.25*\diffY) [dnode]{0};
\foreach \i in {1,3,4,6,10,11}
	\node(n) at (-0.333*\diffX+0.333*\i*\diffX,-2.25*\diffY) [dnode]{1};
\end{tikzpicture}
~\\[1cm]
\subcaption{Tape memory}
\end{subfigure}
\hspace{0.07\linewidth}
\begin{subfigure}[b]{0.45\linewidth}
\centering
\begin{tikzpicture}[auto=left]
\node (v0) at (0,0) [graphnode] {0};
\node (v1) at (0+\diffX,0) [graphnode] {1};
\node (v2) at (0+2*\diffX,0) [graphnode] {2};
\node (v3) at (0+3*\diffX,0) [graphnode] {3};
\node (v4) at (0+3*\diffX,-\diffY) [graphnode] {4};
\node (v5) at (0+2*\diffX,-\diffY) [graphnode] {5};
\node (v6) at (0+2*\diffX,\diffY) [graphnode] {6};
\node (v7) at (\diffX,\diffY) [graphnode] {7};
\node(p0) at (2.4*\diffX,0.4*\diffY)[dnode]{$p_0$};
\node(p1) at (2.4*\diffX,1.4*\diffY)[dnode]{$p_1$};
\node(p2) at (3.4*\diffX,-0.6*\diffY)[dnode]{$p_2$};
\node(p3) at (3.4*\diffX,0.4*\diffY)[dnode]{$p_3$};
\node(pstart) at (-0.3*\diffX,0.5*\diffY)[dnode]{$p_\text{start}$};

\draw[dedge] (v0.15) -- (v1.165);
\draw[<-] (v0.-15) -- (v1.195);

\draw[dedge] (v1.15) -- (v2.165);
\draw[<-] (v1.-15) -- (v2.195);

\draw[dedge] (v2.15) -- (v3.165);
\draw[<-] (v2.-15) -- (v3.195);

\draw[dedge] (v3.255) -- (v4.105);
\draw[<-] (v3.285) -- (v4.75);

\draw[dedge] (v4.165) -- (v5.15);
\draw[<-] (v4.195) -- (v5.-15);

\draw[dedge] (v5.105) -- (v2.255);
\draw[<-] (v5.75) -- (v2.285);

\draw[dedge] (v2.105) -- (v6.255);
\draw[<-] (v2.75) -- (v6.285);

\draw[dedge] (v6.165) -- (v7.15);
\draw[<-] (v6.195) -- (v7.-15);

\node[circle,fill=black,inner sep=0pt,minimum size=4pt] (p0-2) at (v2.north east){};
\draw[very thin] (p0) -- (p0-2);

\node[circle,fill=black,inner sep=0pt,minimum size=4pt] (p1-2) at (v6.north east){};
\draw[very thin](p1) -- (p1-2);

\node[circle,fill=black,inner sep=0pt,minimum size=4pt] (p2-2) at (v4.north east){};
\draw[very thin](p2) --(p2-2);

\node[circle,fill=black,inner sep=0pt,minimum size=4pt] (p3-2) at (v3.north east){};
\draw[very thin](p3) --(p3-2);

\node[circle,fill=black,inner sep=0pt,minimum size=4pt] (pstart-2) at (v0.north west){};
\draw[very thin](pstart) --(pstart-2);
\end{tikzpicture}
\subcaption{Memory encoded by pebbles}
\end{subfigure}

\caption{Memory encoding by pebbles on a closed walk. The state of the tape of length $2m=12$ in (a) is encoded by the position of the $c-3=4$ pebbles in (b). The number of the vertices corresponds to the order of first traversal by the closed walk $\omega$ starting in $0$. The position of each pebble encodes $m_1=3$ bits.}
\label{fig_example_exploration_alg}
\end{figure}

\emph{1. Finding a closed walk~$\omega$.}
Lemma~\ref{lemma_uxs_general_graph} yields a Turing machine $\turing_\text{walk}$ with $c_1$ states
and a tape of length $m/c_0$ that produces an exploration sequence corresponding to 
a closed walk $\omega$ that contains at least $\min\{n,2^{m_1}\}$ distinct vertices and has length
at most $2^{c_1 m_1 } = 2^{m/c_0}$. 
We use a variable $R_\mathrm{walk}$ of size $m/c_0$ for the memory of $\turing_\text{walk}$, which is initially assumed to have all bits set to 0. 
If $2^{m_1} > n$, then the exploration sequence produced by $\turing_\text{walk}$ is a walk exploring~$G$.
Note that by definition $m/{c'}=m_1$ so this happens exactly when the first case in the theorem occurs.
Below we will show how to count the number of unique vertices on the closed walk of  $\turing_\text{walk}$.
Hence, the pebble machine $T'$ can initially  walk along the closed walk~$\omega$ counting the
number of distinct vertices. If this number is smaller than~$2^{m_1}$, we know that
we have visited all vertices of $G$ so that we can collect all pebbles and return to the pebble $p_\text{start}$, which has not been moved and therefore marks the starting vertex of $T$. We show at the end of the
proof that this takes at most $2^{\bigO(m)}$ edge traversals and computation steps.

From now on, we can therefore assume that $\omega$ contains at least $2^{m_1}$ distinct vertices. 
We need to show that $c-3$ pebbles placed along the walk $\omega$ can be used to encode all of the $2^{2m}$ configurations of the tape of~$\pebblemachine$.
Figure~\ref{fig_example_exploration_alg} shows how each pebble encodes a certain part of the tape of~$\pebblemachine$. The idea is that each pebble can be placed on one of $2^{m_1}$ different vertices, thus encoding exactly $m_1$ bits. 
We divide the tape of length $2m$ into $2m/m_1 = 2 c_0 c_1$
parts of size $m_1$ each, such that the position of pebble $p_i$ encodes the bits $\{i m_1,\ldots, (i+1)m_1 -1\}$, where we assume the bits of the tape of $\pebblemachine$ to be numbered $0,1\ldots, 2m -1$.
As $c \geq 2 c_0 c_1 +3$, we have enough pebbles to encode the configuration of the tape of~$\pebblemachine$.
 
\begin{algorithm}[tb]
\floatname{algorithm}{Algorithms}
\caption{Auxiliary functions for moving along the closed walk $\omega$.}
\label{alg_moving_along_walk_aux}

\begin{algorithmic}
\Function{Step()}{}
	\State traverse edge according to value of exploration sequence output by $\turing_\text{walk}$ 
	\State $R_\mathrm{steps} \gets R_\mathrm{steps}+1$
\EndFunction
\end{algorithmic}

\begin{algorithmic}
\Function{FindPebble}{$p_i$}
\While{not \textsc{observe}($p_i$)}
	\State \textsc{Step}()
\EndWhile
\EndFunction
\end{algorithmic}

\begin{algorithmic}
\Function{Restart()}{}
\State \textsc{FindPebble}($p_\text{start}$)
\State $R_\mathrm{steps} \gets 0$
\State $R_\mathrm{id} \gets 0$
\State $R_\mathrm{walk} \gets 0$
\EndFunction
\end{algorithmic}

\end{algorithm}

 \begin{algorithm}[tb]
\floatname{algorithm}{Algorithm}
\caption{Moving along the closed walk $\omega$ while updating $R_\mathrm{steps}$ and  $R_\mathrm{id}$.}
\label{alg_moving_along_walk}

\begin{algorithmic}
\Function{NextDistinctVertex()}{}
\If{$R_\mathrm{id}= 2^{m_1} -1$ }
		\State \textsc{Restart}()
		\State \Return
\EndIf
\State $R_\mathrm{id} \gets R_\mathrm{id} + 1$
\State $R_\mathrm{steps}' \gets R_\mathrm{steps}$
\Repeat
	\State \textsc{Step}()
	\State $R_\mathrm{steps}'\gets R_\mathrm{steps}'+1$
	\State \textsc{drop}($p_\text{temp}$)
	\State $R_\mathrm{walk}' \gets R_\mathrm{walk}$
	\State \textsc{Restart}()
	\State \textsc{FindPebble}($p_\text{temp}$)
	\State \textsc{pickup}($p_\text{temp}$)
	\State $R_\mathrm{walk} \gets R_\mathrm{walk}'$ 
\Until{$R_\mathrm{steps}= R_\mathrm{steps}'$} 
	
\EndFunction
\end{algorithmic}
\end{algorithm}

\emph{2. Navigating~$\omega$.}
Let $R_\mathrm{steps}$ be a variable counting the number of steps along~$\omega$ and~$R_\mathrm{id}$ be a variable for counting the number of unique vertices visited along~$\omega$ after starting in the vertex marked by~$p_\mathrm{start}$. 
Note that $R_\mathrm{id}$ gives a way of associating a unique identifier to the first $2^{m_1}$ distinct vertices along $\omega$. As $m_1 \leq m/c_0$ holds,  $m/c_0$ tape cells suffice for counting the first $2^{m_1}$ distinct vertices along~$\omega$. The overall number of steps along the closed walk is bounded by $2^{m/c_0}$ and therefore $m/c_0$ tape cells also suffice for counting the steps along~$\omega$.

It remains to show that we can move along the closed walk $\omega$ while updating $R_\mathrm{steps}$ and $R_\mathrm{id}$, such that, starting
from the vertex marked by $p_\text{start}$, the variable  $R_\mathrm{steps}$ contains the number of steps taken and $R_\mathrm{id}$ contains the number of distinct vertices visited. 
Let \textsc{drop($p_i$)} denote the operation of dropping pebble~$p_i$ at the current location, \textsc{pickup($p_i$)} the operation of picking up~$p_i$ at the current location if possible, and let \textsc{observe($p_i$)} be ``true'' if pebble $p_i$ is located at the current position. 
Consider the auxiliary functions shown in Algorithms~\ref{alg_moving_along_walk_aux}.
The function \textsc{Step()} moves one step along~$\omega$ and updates $R_\mathrm{steps}$ accordingly.
The function \textsc{FindPebble($p_i$)} moves along~$\omega$ until it finds pebble~$p_i$. 
The function \textsc{Restart()} goes back to the starting vertex marked by  $p_\text{start}$, sets both variables $R_\mathrm{steps}$ and $R_\mathrm{id}$ to~$0$, and restarts $\turing_\text{walk}$ by setting the variable $R_\mathrm{walk}$ to~0.
Finally, the function \textsc{NextDistinctVertex()} in Algorithm~\ref{alg_moving_along_walk} does the following: If the number of distinct vertices visited is already $2^{m_1}$, then we go back to the start. 
Otherwise, we continue along~$\omega$ until we encounter a vertex we have not visited before. 
We repeatedly traverse an edge, drop the pebble $p_\text{temp}$, store the number of steps until reaching that vertex, then we restart from the beginning and check if we can reach that vertex with fewer steps. 
If not, we found a new distinct vertex. Note that we use the auxiliary variables $R_\mathrm{steps}'$ and $R_\mathrm{walk}'$, which 
both need a constant number of blocks of size~$m_0/c_0$.

\begin{algorithm}[tb]
\floatname{algorithm}{Algorithms}
\caption{Reading and changing positions of pebbles.}
\label{read_change_pebble_position}

\begin{algorithmic}
\Function{GetPebbleId}{$p_i$} 
\State \textsc{Restart}()
\While{not \textsc{observe}($p_i$)}
	\State \textsc{NextDistinctVertex}()
\EndWhile
	\Return $R_\mathrm{id}$
\EndFunction
\end{algorithmic}

\begin{algorithmic}
\Function{PutPebbleAtId}{$p_i$,id} 
\State \textsc{FindPebble}($p_i$)
\State \textsc{pickup}($p_i$)
\State \textsc{Restart}()
\While{id$>$0}
	\State $\text{id} \gets \text{id}-1$
	\State \textsc{NextDistinctVertex}()
\EndWhile
\textsc{drop}($p_i$)
\EndFunction
\end{algorithmic}

\end{algorithm}

\begin{algorithm}[tb]
\floatname{algorithm}{Algorithms}
\caption{Reading and writing one bit for the simulated memory.}
\label{alg_read_write_tape}

\begin{algorithmic}
\Function{ReadBit()}{}
\State $i \gets  \left \lfloor R_\mathrm{head}/m_1  \right \rfloor$
\State $j \gets R_\mathrm{head} - m_1 \cdot i$
\State $\text{id} \gets \text{\textsc{GetPebbleId}}(p_i)$
\State \Return $j$-th bit of id (in binary)
\EndFunction
\end{algorithmic}

\begin{algorithmic}
\Function{WriteBit}{$b$} 
\State $i \gets  \left \lfloor R_\mathrm{head}/m_1  \right \rfloor$
\State $j \gets R_\mathrm{head} - m_1 \cdot i$
\State $\text{id} \gets \text{\textsc{GetPebbleId}}(p_i)$
\If{$b = 1$ \textbf{and} \textsc{ReadBit}() $= 0$}
	\State $\text{id} \gets \text{id} + 2^j$
\ElsIf{$b = 0$ \textbf{and} \textsc{ReadBit}() $= 1$}
		\State $\text{id} \gets \text{id} - 2^j$
\EndIf
\State \textsc{PutPebbleAtId}($p_i$,id)
\EndFunction
\end{algorithmic}
\end{algorithm}

\emph{3. Reading from and writing to simulated memory.}
We show how to simulate the changes to the tape of $\pebblemachine$ by changing the positions of the pebbles along~$\omega$. 
The transition function~$\delta_{\text{TM}}$ of $\pebblemachine$ determines how $\pebblemachine$ does computations on its tape and, in particular, how $\pebblemachine$ changes its head position. 
We use a variable $R_\mathrm{head}$ of size~$m/c_0$ to store the head position. By assumption, $m \geq m_0$ and therefore $2^{m/c_0}>2m$, i.e., the size of $R_\mathrm{head}$ is sufficient to store the head position. In order to simulate one transition of $\pebblemachine$ according to $\delta_{\text{TM}}$,
we need to read the bit at the current head position and then write to the simulated memory and change
the head position accordingly. 
Reading from the simulated memory is done by the function \textsc{ReadBit()} and writing of a bit $b$ to the simulated memory by the function \textsc{WriteBit($b$)} (cf.~Algorithms~\ref{alg_read_write_tape}).

First, let us consider the two auxiliary functions \textsc{GetPebbleId($p_i$)} and \textsc{PutPebbleAtId($p_i$,} id) (cf.~Algorithms~\ref{read_change_pebble_position}).
As the name suggests, the function \textsc{GetPebbleId($p_i$)} returns the unique identifier associated to the vertex marked by~$p_i$. Recall that vertices are indistinguishable. Here, unique identifier refers to the number of distinct vertices on the walk $\omega$ before reaching the vertex marked with $p_i$ for the first time.
Given an identifier id, we can use the function \textsc{PutPebbleAtId($p_i$,} id) for placing pebble $p_i$ at the unique vertex corresponding to id.  
By the choice of our encoding, if $R_\mathrm{head}= i \cdot m_1 + j$ with $j \in \{0,\ldots, m_1-1\}$, then the $j$-bit of the binary encoding of the position of pebble $p_i$ encodes the contents of the tape cell specified by~$R_\mathrm{head}$. 
Thus, for reading from the simulated memory, we have to compute $i$ and $j$ and determine the position of the corresponding pebble in the function \textsc{ReadBit()}. For the function \textsc{WriteBit($b$)}, we also compute $i$ and $j$. Then, we move the pebble~$p_i$
by $2^j$ unique vertices forward if the bit flips to~$1$ or by $2^j$ unique vertices backward if the bit flips to~$0$.

\emph{4. Relocating~$\omega$.}
When~$\pebblemachine$ moves from a vertex $v$ to another vertex~$v'$, the walk $\omega$ and the pebbles on it need to be relocated. Recall that $\pebblemachine'$ marked the current vertex $v$ with the pebble~$p_\text{start}$. After having computed the label of the edge to~$v'$,  $\pebblemachine'$ drops the pebble $p_\text{next}$ at $v'$. Then $\pebblemachine'$ moves the pebbles placed along the walk~$\omega$ to the corresponding positions along a new walk $\omega'$ starting at~$v'$ in the following way. 
We iterate over all $c-3$ pebbles and for each pebble $p_i$ we start in~$v$, determine the identifier id of the vertex marked by~$p_i$ via \textsc{GetPebbleId($p_i$)}, pick up~$p_i$, move to $p_\text{next}$ and place~$p_i$ on $\omega'$ using the function \textsc{PutPebbleAtId($p_i$}, id).
In this call of \smash{\textsc{PutPebbleAtId($p_i$}, id)} all occurrences of $p_\text{start}$ are replaced by $p_\text{next}$. 
This way, we can carry the memory simulated by the pebbles along during the graph traversal.

Thus, we have shown that in the second case $\pebblemachine'$ can simulate the traversal of $\pebblemachine$ in $G$ while using a tape with half the length, but $c$ additional pebbles and a factor of $c$ additional states.

Finally, we bound the number of edge traversals and computation steps in both cases.
First, we bound the number of edge traversals 
that~$\pebblemachine'$ needs for simulating one computation step of $\pebblemachine$. 
Recall that $\pebblemachine'$ needs at most $2^{m/c_0}\leq 2^m$ edge traversals for moving once along the whole closed walk $\omega$. A call of the function \textsc{Step()} corresponds to one edge traversal, a call of \textsc{FindPebble($p_i$)} thus corresponds to at most $2^m$ edge traversals and also a call of \textsc{Restart()} corresponds to at most $2^m$ edge traversals. 
Moreover, one iteration of the loop in \textsc{NextDistinctVertex()} accounts for at most $2^m$ edge traversals and therefore executing the whole function results in at most $2^m \cdot 2^m = 2^{2m}$ edge traversals.
This means that one call of \textsc{GetPebbleId($p_i$)} or \textsc{PutPebbleAtId($p_i$,id)} incur at most $2^{\bigO(m)}$ edge traversals and this also holds for \textsc{ReadBit()} and \textsc{WriteBit($b$)}. 
Hence, for every computation step performed by $\pebblemachine$ according to~$\delta_{\text{TM}}$, the pebble machine $\pebblemachine'$ performs actions according to \textsc{ReadBit()} and \textsc{WriteBit($b$)} and overall does at most $2^{\bigO(m)}$ edge traversals.
The above argument also shows that at most $2^{\bigO(m)}$ edge traversals are necessary to count the number of
distinct vertices on the closed walk $\omega$ at the beginning.

Next, let us bound the number of edge traversals  that $\pebblemachine'$ needs for reproducing one edge traversal of $\pebblemachine$. This means that we need to count how many edge traversals are necessary to relocate all pebbles placed along the walk~$\omega$ to the new walk~$\omega'$. For every pebble~$p_i$, we call \textsc{GetPebbleId($p_i$)} which results in at most $2^m$ edge traversals, we pick up~$p_i$ and move to $p_\text{next}$ which again needs at most $2^m$ edge traversals, and place~$p_i$ on $\omega'$ using the function \textsc{PutPebbleAtId($p_i$}, id) which also needs $2^m$ edge traversals. Overall, this procedure is done for a constant number of pebbles and hence requires at most $2^{\bigO(m)}$ edge traversals.

Next we bound the number of computation steps of $\pebblemachine'$ by using the bounds on the number of edge traversals. 
Recall that the state of $\pebblemachine'$ is a tuple $(q,q')$, where $q$ corresponds to the state of~$\pebblemachine$. 
In the computation only the second component of the state of $\pebblemachine'$ changes and therefore there are only at most $c$ possible states. The tape length and number of possible head positions of the Turing machine is $m$. Since we may assume without loss of generality that $m \geq 2$, we can bound the number of distinct configurations of $\pebblemachine'$ in each computation by $2^{\bigO(m)}$. Hence, after every edge traversal $\pebblemachine'$ does at most $2^{\bigO(m)}$ computation steps.
This implies that in the first case of the statement of the theorem, the number of computation steps is bounded by
 $2^{\bigO(m)}$ because the number of edge traversals is bounded by $2^{\bigO(m)}$ as shown above.
Similarly, in the second case of the statement of the theorem the total number of computation steps after~$2^{\bigO(m)}$ edge traversals is bounded by $2^{\bigO(m)}$. Since $m \geq 2$ this means that also the sum of computation steps and edge traversals can be bounded by $2^{\bigO(m)}$ both for one computation step and one edge traversal of~$\pebblemachine$.
\end{proof}

Finally, we show that by recursively simulating a pebble machine by another pebble machine with half the memory but a constant number of additional pebbles we can explore any graph with at most $n$ vertices while using $\bigO(\log \log n)$ pebbles and only $\bigO(\log \log n)$ bits of memory.

%Figure Not Necessary
%\begin{figure}
%\centering
%\tikzstyle{pnode}=[rectangle,draw,minimum height=1.8em, minimum width=12em, scale=1, align=center]
%\tikzstyle{dnode}=[scale=1]
%\tikzstyle{dnode2}=[scale=0.8]
% \tikzstyle{dedge}=[->, >=latex]
% \tikzstyle{label}=[scale=0.9, color=blue]
% \newcommand*{\diffX}{1.6}
% \newcommand*{\diffY}{1.3}
% \begin{tikzpicture}[auto=left]
% \node (p1) at (0,0) [pnode] {Memory:  $c \log n$ bits};
% \node (p2) at (0,-\diffY) [pnode]  {Memory:  $c/2 \log n$ bits};
%\node (p3) at (0,-3*\diffY) [pnode] {Memory:  $c$ bits};
%\draw[dotted] (-0.0*\diffX,-1.9*\diffY)--(0.0*\diffX,-2.1*\diffY) ;
% \draw[dedge] (p1)--node [dnode2, right] {writeBit($b$)} node [dnode2, left] {readBit()}(p2);
%\draw[dedge] (p2)--node [dnode2, right] {writeBit($b$)} node [dnode2, left] {readBit()}  (0,-1.8*\diffY) ;
%\draw[dedge] (0,-2.2*\diffY)--node [dnode2, right] {writeBit($b$)} node [dnode2, left] {readBit()}  (p3) ;
%
%\node (d1) at (-2*\diffX,0.75*\diffY) [dnode,align=left] {Level of \\ recursion:};
%\node (d1) at (-2*\diffX,0) [dnode] {1};
% \node (d2) at (-2*\diffX,-\diffY) [dnode]  {2};
%\node (d3) at (-2*\diffX,-3*\diffY) [dnode] {$\log \log n$};
% 
% \end{tikzpicture}
% \caption{Recursive structure of the exploration algorithm using $O(\log \log n)$ pebbles and memory.}
%\label{fig_exploration_alg_flow_chart}
%\end{figure} 
 
\begin{theorem}
\label{theorem_loglog_algorithm}
Any connected  undirected graph on at most $n$ vertices can be explored by an agent in a polynomial number of steps using~$\bigO(\log\log n)$ pebbles and~$\bigO(\log\log n)$ bits of memory. The agent does not require $n$ as input and terminates at the starting vertex with all pebbles after exploring the graph.
\end{theorem}

\begin{proof}
Let $c,c' \in \N$ be the constants from Theorem~\ref{theo_pebble_machine_simulation}. 
Let $r \in \N$ be arbitrary and consider a $(c,0, c'2^{r+1})$-pebble machine $T^{(r)}$ that simply terminates without making a computation step or edge traversal.
Applying Theorem~\ref{theo_pebble_machine_simulation} for the pebble machine $T^{(r)}$ gives a $(c^2, c, c'2^r)$-pebble machine $\smash{T^{(r)}_r}$ with the following properties. If $n < \smash{2^{2^r}}$, then $\smash{T^{(r)}_r}$ explores the graph and returns to the starting vertex. If, on the other hand, $n \geq \smash{2^{2^r}}$, then $T_r^{(r)}$ reproduces the walk of $T^{(r)}$ (which in this case is of course trivial). Note that these properties hold even though the number $n$ of vertices is unknown and, in particular, not given as input to $\smash{T_r^{(r)}}$.

Applying Theorem~\ref{theo_pebble_machine_simulation} iteratively, we obtain a
$(c^{r+2-i}, (r+1-i)c, c'2^{i})$-pebble machine $T^{(r)}_i$ for all $i \in \{0,\dots,r-1\}$ that reproduces the walk of  $T^{(r)}_{i+1}$ or it already explores the given graph and returns to the start vertex. 
For a graph $G$ with $n< \smash{2^{2^{r}}}$, $T^{(r)}_{r}$ explores $G$ and returns to the start with all pebbles and terminates. 
Thus for such a graph $G$ it does not matter which case occurs when applying~Theorem~\ref{theo_pebble_machine_simulation},  as in both cases we can conclude that $T_i^{(r)}$ for $i \in \{0,\dots,r-1\}$ explores the graph, returns with all pebbles to the start vertex and terminates.
  
If we have $n\geq\smash{ 2^{2^{r}}}$, then $n\geq 2^{2^{i}}$ holds for all $i \in \{0,\ldots,r-1\}$ and in particular $\smash{T^{(r)}_0}$ reproduces the walk of $T^{(r)}$ in $G$, i.e., remains at the starting vertex and terminates.

The desired pebble machine $T$ exploring any graph $G$ with $\bigO(\log\log n)$ pebbles and~$\bigO(\log\log n)$ bits of memory works as follows: We have a counter $r$, which is initially 1 and is increased by one after each iteration until the given graph $G$ is explored. In iteration~$r$, pebble machine $T$ does the same as the $(c^{r+2},(r+1) c, c')$-pebble machine $T^{(r)}_0$ until it terminates. The pebble machine $T$ terminates  as soon as for some $r \in \N$ the pebble machine  $\smash{T^{(r)}_0}$ recognizes that it explored the whole graph.
This happens when $r=\ceil{\log \log n }+1$. 
Hence, $T$ uses at most $\bigO( \log \log n)$ pebbles.

Concerning the memory requirement of $T$, note that $T$ needs
to store the state of  $\smash{T^{(r)}_0}$, the tape content of  $\smash{T^{(r)}_0}$
and the current value of $r$.  
There are $c^{r+2}$ states of the pebble machine $\smash{T^{(r)}_0}$, its tape length is  $c'$ and $r \leq
 \ceil{\log \log n}+1$ in every iteration, so that $T$ can be implemented with~$\bigO(\log\log n)$ bits of memory.

It is left to show is that the number of edge traversals of $\pebblemachine$ in the exploration of a given graph $G$ with $n$ vertices is polynomial in~$n$. 
To this end, we first show that the number of edge traversals of the pebble machine $\smash{T_0^{(r)}}$ is bounded by $n^{\bigO(1)}$ for all $r \in \{1,\dots,\lceil \log \log n\rceil+1\}$. 
Let $r \in \{1,\dots, \lceil \log \log n \rceil + 1\}$ be arbitrary and let $t_i$ denote the sum of the number of edge traversals and computation steps of $\smash{\pebblemachine^{(r)}_{i}}$ in the given graph $G$. 
The pebble machine $\smash{T^{(r)}_{r}}$ has a tape of length of $m = c'2^{r}$. Applying Theorem~\ref{theo_pebble_machine_simulation}, we get that either $\smash{T^{(r)}_{r}}$ explores $G$ and uses at most $2^{\bigO(m)}$ edge traversals and computation steps or $\smash{T^{(r)}_{r}}$ simulates the walk of a pebble machine that does not make a single edge transition and uses at most $2^{\bigO(m)}$ edges traversals and computation steps. In both cases, we obtain
\begin{align*}
t_{r} \leq 2^{\bigO(2^r)} \leq 2^{\bigO(2^{\log \log n})} = 2^{\bigO(\log n)} = n^{\bigO(1)}.
\end{align*}
This shows the desired bound for $t_r$. Furthermore, 
one computation step or one edge traversal of $\smash{\pebblemachine^{(r)}_i}$ leads to at most $\smash{2^{\bigO( c' \cdot 2^{i})} =2^{\bigO(1) 2^{i}}}$ edge traversals and computation steps of $\smash{\pebblemachine^{(r)}_{i-1}}$ 
by Theorem~\ref{theo_pebble_machine_simulation}. Hence, we obtain
\begin{align}
t_{i-1} \leq 2^{\bigO(1) 2^{i}} t_i \quad \quad  \forall \ i \in \{1,\ldots, \lceil \log \log n \rceil +1 \}. \label{ineq_edge_traversals}
\end{align}
By iterative application of Inequality~\ref{ineq_edge_traversals}, we obtain
\begin{align*}
t_0 
& \leq 2^{ \bigO(1) 2^{i}} t_1 
  \leq \ldots \leq 2^{\bigO(1) \sum_{i=1}^{\lceil \log \log n \rceil+1} 2^{i}} \cdot t_{\lceil \log \log n \rceil+1}
 \leq 2^{\bigO(1) 2^{\lceil \log \log n \rceil}} \cdot n^{\bigO(1)} \leq n^{\bigO(1)}.
\end{align*}
Thus, the number of edge traversals $t_0$ of $\pebblemachine^{(r)}_0$ is polynomial in $n$.
As $T$ performs at most  $n^{\bigO(1)}$ edge traversals according to $\smash{T_0^{(r)}}$  for at most
$\ceil{\log \log n}+1$ distinct values of $r$, the overall number of edge traversals of $T$
is also bounded by  $n^{\bigO(1)}$.
\end{proof}

Since an additional pebble is more powerful than a bit of memory (Lemma~\ref{lem:pebble>bit}), we obtain the following direct corollary of Theorem~\ref{theorem_loglog_algorithm}.

\begin{corollary}
\label{cor:constant_memory_pebbles}
Any connected undirected graph on at most $n$ vertices can be explored by an agent in a polynomial number of steps using $\bigO(\log \log n)$ pebbles and constant memory. The agent does not require $n$ as input and terminates at the starting vertex with all pebbles after exploring the graph.
\end{corollary}

Since an additional agent is more powerful than a pebble (Lemma~\ref{lem:agent>pebble}), we obtain the following direct corollary of Theorem~\ref{theorem_loglog_algorithm}  and Corollary~\ref{cor:constant_memory_pebbles}.

\begin{corollary}
\label{cor:constant_memory_agents}
Any connected undirected graph on at most $n$ vertices can be explored in polynomial time by a set of $\bigO(\log \log n)$ agents with constant memory each. The agents do not require $n$ as input and terminate at the starting vertex after exploring the graph.
\end{corollary}

\begin{remark}\label{remark_pebbles_encoding_transition}
The agent in  Theorem~\ref{theorem_loglog_algorithm} requires $\bigO(\log \log n)$ bits of memory and the agents in  Corollaries~\ref{cor:constant_memory_pebbles} and~\ref{cor:constant_memory_agents} only  $\bigO(1)$ bits of memory. 
An interesting question is how much memory is necessary to fully encode the transition function 
\begin{align*}
\delta \colon \Agentstate \times \N \times \N \times 2^P 	\times 2^P &\to \Agentstate \times (\N \cup \{\bot\}) \times 2^P \times 2^P,
\end{align*}
of an agent (see  \S~\ref{sec:prelim:agents:pebbles}). 
Naively encoding it as a table with a row for every possible state, vertex degree, previous edge
label and possible combination of $\bigO(\log \log n)$ pebbles/agents at the current vertex takes $(\log n)^{\bigO(1)}$ bits of memory.

However, we can obtain a much more compact encoding by exploiting the specific structure of our algorithm:
First of all, we never explicitly use the degree of the current vertex.
Moreover, the Turing machine from Lemma~\ref{lemma_uxs_general_graph} that we internally use produces an exploration sequence of the form~$\{-1,0,1\}^*$.
This means that our transition function can be expressed more concisely if we would allow in our model to specify transitions relative to the label of the previous edge.

Furthermore, our algorithm only interacts with a constant number of pebbles in every level of the recursion (cf.~Theorem~\ref{theo_pebble_machine_simulation}).
We can express the state of $\pebblemachine$ in the proof of Theorem~\ref{theorem_loglog_algorithm} as a vector, where each component encodes the state in a different level of the recursion.
In every transition, only two consecutive entries of this vector can change, as one level of recursion only interacts with the level of recursion below to access the simulated memory.  

Since there are only a constant number of states per recursive level, and only a constant number of pebbles involved, all transitions regarding two consecutive levels can be encoded in constant memory.
If we therefore explicitly encode all~$\bigO(\log\log n)$ levels of recursion and additionally 
allow to only give the edge label offset in the transition function, the entire transition function can be encoded with~$\bigO(\log\log n)$ bits of memory.
\end{remark}

\FloatBarrier

%!TEX root = 0_pebbles_main.tex

\section{Lower bound for collaborating agents}
\label{sec:lower_bound}

The goal in this section is to obtain a lower bound on the number $k$ of $s$-state agents needed for exploring any graph on at most $n$ vertices. 
To this end, we will construct a \emph{trap} for a given set of agents, i.e., a graph that the agents are unable to explore. The number of vertices of this trap yields a lower bound on the number of agents required for exploration. The graphs involved in our construction are 3-regular and allow a labeling such that the two port numbers at both endpoints of any edge coincide. We therefore speak of the label of an edge and assume the set of labels to be $\{0,1,2\}$.

Moreover, we call the sequence of labels $l_0, l_1, l_2, \ldots $ of the edges traversed by an agent
in a 3-regular graph $G$ starting at a vertex $v_0$ a \emph{traversal sequence} and say that the agent \emph{follows} the traversal sequence $ l_0, l_1, l_2, \ldots$ in $G$ starting in $v_0$.
Note that traversal sequences specify absolute labels to follow, whereas exploration sequences give offsets to the previous label in each step.
  
The most important building block for our construction are \emph{barriers}.
Intuitively, a barrier is a subgraph that cannot be crossed by a subset of the given set of agents.
To define barriers formally, we need to describe how to connect two 3-regular graphs.
Let $B$ be a 3-regular graph with two distinguished edges $\{u,v\}$ and $\{u',v'\}$ both labeled $0$, as shown in Figure~\ref{fig_barrier_example}.
An arbitrary 3-regular graph $G$ with at least two edges labeled $0$ can be connected to $B$ as follows:
We remove the edges $\{u,v\}$ and $\{u',v'\}$ from $B$ and two edges labeled 0 from $G$. 
We then connect each vertex of degree 2 in $G$ with a vertex of degree~2 in $B$ via an edge labeled 0.

\begin{figure}
\centering
\tikzstyle{graphnode}=[circle,draw,minimum size=2em,scale=0.8]
\tikzstyle{dnode}=[scale=0.8]
\newcommand*{\diffY}{0.6}
\newcommand*{\diffX}{0.5}
\newcommand*{\xOff}{6}
\begin{tikzpicture}
\node(g1) at (0,-2*\diffY) [ellipse,draw,minimum width=6em,minimum height=8em] {$B$};
\node (v1) at (\diffX,-3*\diffY) [graphnode] {};
\node (d1) at (\diffX,-3*\diffY) [dnode] {$v'$};
\node (v2) at (\diffX,-\diffY)  [graphnode] {$v$};

\node (v3) at (-\diffX,-3*\diffY) [graphnode] {};
\node (d3) at (-\diffX,-3*\diffY) [dnode] {$u'$};
\node (v4) at (-\diffX,-\diffY) [graphnode] {$u$};
\draw (v4)--node [dnode,above] {$0$} (v2);
\draw (v3)--node [dnode,below] {$0$} (v1);

\node(g2) at (0+\xOff,-2*\diffY) [ellipse,draw,minimum width=6em,minimum height=8em] {$B$};
\node(g2) at (0+0.54*\xOff,-2*\diffY) [ellipse,draw,minimum width=6em,minimum height=9em] {$G$};
\node (w1) at (\diffX+\xOff,-3*\diffY) [graphnode] {};
\node (d1) at (\diffX+\xOff,-3*\diffY) [dnode] {$v'$};
\node (w2) at (-\diffX+\xOff,-3*\diffY) [graphnode] {};
\node (d2) at (-\diffX+\xOff,-3*\diffY) [dnode] {$u'$};
\node (w3) at (\diffX+\xOff,-\diffY) [graphnode] {};
\node (d3) at (\diffX+\xOff,-\diffY) [dnode] {$v$};
\node (w4) at (-\diffX+\xOff,-\diffY) [graphnode] {$u$};
\node (u1) at (\diffX+0.52*\xOff,-3*\diffY) {};
\node (u2) at (0.45*\xOff,-3.3*\diffY) {};
\node (u3) at (\diffX+0.5*\xOff,-1*\diffY) {};
\node (u4) at (0.45*\xOff,-0.6*\diffY) {};
\draw (w1) to [bend left=50,looseness=1]  node [dnode,above] {$0$} (u2);
\draw (w2) to [bend left=50,looseness=1]  node [dnode,above] {$0$} (u1);
\draw (w3) to [bend right=50,looseness=1] node [dnode,above] {$0$} (u4);
\draw (w4) to [bend right=50,looseness=1] node [dnode,below] {$0$} (u3);
\end{tikzpicture}
\caption{The $r$-barrier $B$ on the left with two distinguished edges $\{u,v\}$, $\{u',v'\}$ can be connected to an arbitrary graph $G$, as shown on the right.}
\label{fig_barrier_example}
\end{figure}
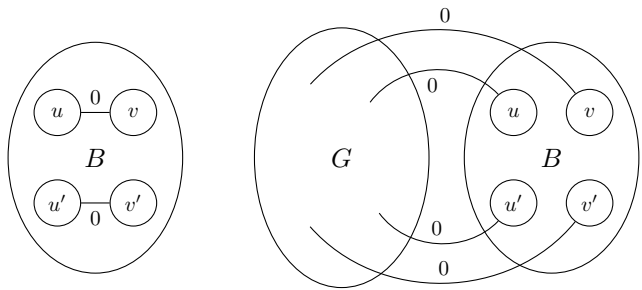 

\begin{definition}[$r$-barrier]
For $1 \leq r \leq k$, the graph $B$ is an $r$\emph{-barrier} for a set of $k$ $s$-state agents~$\agentset$ if for all graphs $G$ connected to $B$ as above, the following two properties hold:
\begin{enumerate}
\item For all subsets of agents $\agentset' \subseteq \agentset$ with $|\agentset'| \leq r$ and every pair $(a,b)$ in $ \{u,v\} \times \{u',v'\}$ the following holds: If initially all agents $\agentset$ are at vertices of~$G$, then no agent in the set $\agentset'$ can traverse $B$ from $a$ to $b$ or vice versa when only agents in $\agentset'$ enter the subgraph $B$ at any time during the traversal. We equivalently say that no subset of $r$ agents can traverse $B$ from $a$ to~$b$ or vice versa. 
\item Whenever a subset of agents $\agentset' \subseteq \agentset$ with $|\agentset'| = r+1$ enters the subgraph~$B$ during the traversal, all agents in $\agentset'$ leave $B$ either via $u$ and $v$ or via $u'$ and $v'$ if no other agents visit $B$ during this traversal. In other words, the set of agents~$\agentset'$ cannot split up such that a part of the agents leaves $B$ via~$u$ or~$v$ and the other part via~$u'$ or~$v'$.
\end{enumerate}

\end{definition}

A $k$-barrier immediately yields a trap for a set of agents.

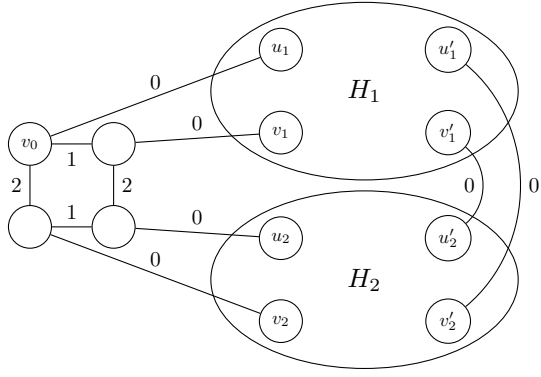
\begin{figure}
\centering
%\tikzstyle{graphnode}=[circle,draw,minimum size=1.8em,scale=0.7]
\tikzstyle{graphnode}=[circle,draw,minimum size=23pt,scale=0.7]
%\tikzstyle{dnode}=[scale=0.8]
\usetikzlibrary{shapes}
\newcommand*{\x}{1.1}
\newcommand*{\y}{0.55}
\newcommand*{\yOff}{2.5}
\begin{tikzpicture}[auto=left]
\node (a1) at (-6*\x, -\y+0.5*\yOff) [graphnode] {};
\node (a2) at (-5*\x, -\y+0.5*\yOff) [graphnode] {};
\node (a3) at (-6*\x, \y+0.5*\yOff) [graphnode] {$v_0$};
\node (a4) at (-5*\x, \y+0.5*\yOff) [graphnode] {};

\draw (a1) to  node [dnode,above] {$1$} (a2);
\draw (a1) to node [dnode,left] {$2$} (a3);
\draw (a2) to   node [dnode,right] {$2$} (a4);
\draw (a3) to node [dnode,below] {$1$} (a4);

\node(g2) at (-2*\x,0) [ellipse,draw,minimum width=115pt,minimum height=67pt] {$H_2$};
\node(g2) at (-2*\x,\yOff) [ellipse,draw,minimum width=115 pt,minimum height=67pt] {$H_1$};
\node (v1) at (-1*\x, \y) [graphnode] {$u_2'$};
\node (v2) at (-1*\x, -\y) [graphnode] {$v_2'$};
\node (v3) at (-3*\x,\y) [graphnode] {$u_2$};
\node (v4) at (-3*\x,-\y) [graphnode] {$v_2$};

\node (w1) at (-1*\x,\y+\yOff) [graphnode] {$u_1'$};
\node (w2) at (-1*\x,-\y+\yOff) [graphnode] {$v_1'$};
\node (w3) at (-3*\x,\y+\yOff) [graphnode] {$u_1$};
\node (w4) at (-3*\x,-\y+\yOff) [graphnode] {$v_1$};
\draw (w1) to [bend left=50,looseness=1]  node [dnode] {$0$} (v2);
\draw (w2) to [bend left=50,looseness=1]  node [dnode,left] {$0$} (v1);

\draw (w3) tonode [dnode,above] {$0$} (a3);
\draw (w4) to  node [dnode,above] {$0$} (a4);
\draw (a1) to  node [dnode,above] {$0$} (v4);
\draw (a2) to node [dnode, above] {$0$} (v3);
\end{tikzpicture}
\caption{Constructing a trap given two $k$-barriers $H_1$ and $H_2$.}
\label{fig_trap_construction}
\end{figure}

\begin{lemma}
\label{lem_barrier_to_trap}
Given a $k$-barrier with $n$ vertices for a set of $k$ agents $\agentset$, we can construct a trap with $2n+4$ vertices for $\agentset$.
\end{lemma}

\begin{proof}
Let $H_1$ and $H_2$ be two copies of a $k$-barrier for the set of agents $\agentset$ with distinguished edges $\{u_i,v_i\}$, $\{u_i',v_i'\}$ of $H_i$. We connect the two graphs and four additional vertices, as shown in Figure~\ref{fig_trap_construction}. If the agents start in the vertex $v_0$, then none of the agents can reach~$u_1'$ or $v_1'$ via the $k$-barrier
$H_1$ or via the $k$-barrier $H_2$.
Thus the agents $\agentset$ do not explore the graph. The constructed trap for the set of agents $\agentset$ contains $2n+4$ vertices.
\end{proof}

Our goal for the remainder of the section is to construct a $k$-barrier for a given set of $k$ agents $\agentset$ and to give a good upper bound on the number of vertices it contains. This will give an upper bound on the number of vertices
of a trap by Lemma~\ref{lem_barrier_to_trap}. The construction of the $k$-barrier is recursive. We start with a 1-barrier which builds on the following
 useful result by Fraigniaud et al.~\cite{fraigniaud06} stating that, for any set of non-cooperative agents, there is a graph containing an edge which is not traversed by any of them.
A set of agents is \emph{non-cooperative} if the transition function $\delta_i$ of every agent $A_i$ is completely independent of the state and location of the other agents, i.e., $\delta_i$  is independent of $\vec \agentstate_{-i}$, see \S~\ref{sec:prelim:agents:set}.

\begin{theorem}[{\cite[Theorem~4]{fraigniaud06}}]\label{fraigniaud_trap_q_agents}
For any $k$ non-cooperative $s$-state agents, there exists a 3-regular graph $G$ on $\bigO(k s)$ vertices with the following property: There are two edges $\{v_1,v_2\}$ and $\{v_3,v_4\}$ in~$G$, the former labeled 0, such that none of the $k$ agents traverses the edge $\{v_3,v_4\}$ when starting in~$v_1$ or~$v_2$.
\end{theorem}

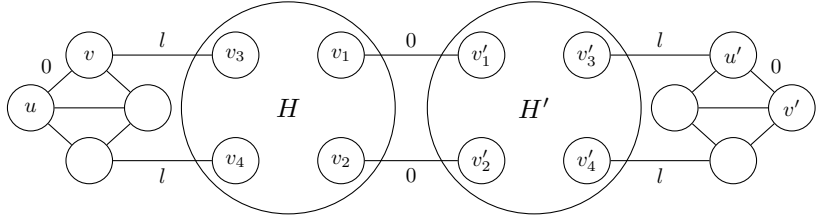
\begin{figure}
\centering
\tikzstyle{graphnode}=[circle,draw,minimum size=2em,scale=0.8]
\tikzstyle{dnode}=[scale=0.8]
\newcommand*{\diffX}{1.1}
\newcommand*{\diffY}{1}
\begin{tikzpicture}[scale=.7,auto=left]

\node (a1) at (-5.5*\diffX,1*\diffY) [graphnode] {$v$};
\node (a2) at (-5.5*\diffX,-1*\diffY) [graphnode] {};
\node (a3) at (-4.5*\diffX,0) [graphnode] {};
\node (a4) at (-6.5*\diffX,0) [graphnode] {$u$};
\draw(a1)--(a3);
\draw(a2)--(a3);
\draw(a3)-- (a4);
\draw(a1)--node [dnode,above left]{$0$} (a4);
\draw(a4)--(a2);

\node (b1) at (5.5*\diffX,1*\diffY) [graphnode] {$u'$};
\node (b2) at (5.5*\diffX,-1*\diffY) [graphnode] {};
\node (b3) at (4.5*\diffX,0) [graphnode] {};
\node (b4) at (6.5*\diffX,0) [graphnode] {$v'$};
\draw(b1)--(b3);
\draw(b2)--(b3);
\draw(b3)--(b4);
\draw(b4)--(b2);
\draw(b1)--node [dnode,above right] {$0$}(b4);

\node(g1) at (-2.1*\diffX,0) [ellipse,draw,minimum width=80pt,minimum height=80pt] {$H$};
\node(g2) at (2.1*\diffX,0) [ellipse,draw,minimum width=80pt,minimum height=80pt] {$H'$};
\node (v1) at (-3*\diffX,\diffY) [graphnode] {$v_3$};
\node (v2) at (-3*\diffX,-\diffY) [graphnode] {$v_4$};
\node (v3) at (-1.2*\diffX,\diffY) [graphnode] {$v_1$};
\node (v4) at (-1.2*\diffX,-\diffY) [graphnode] {$v_2$};
\node (vv1) at (3*\diffX,\diffY) [graphnode] {};
\node (vv2) at (3*\diffX,-\diffY) [graphnode] {};
\node (vv3) at (1.2*\diffX,\diffY) [graphnode] {};
\node (vv4) at (1.2*\diffX,-\diffY) [graphnode] {};
\node (lvv1) at (3*\diffX,\diffY) [dnode] {$v'_3$};
\node (lvv2) at (3*\diffX,-\diffY) [dnode] {$v'_4$};
\node (lvv3) at (1.2*\diffX,\diffY) [dnode] {$v'_1$};
\node (lvv4) at (1.2*\diffX,-\diffY) [dnode] {$v'_2$};

\draw(a1)--node [dnode,above] {$l$} (v1);
\draw(a2)--node [dnode,below] {$l$} (v2);
\draw(b1)--node [dnode,above] {$l$} (vv1);
\draw(b2)--node [dnode,below] {$l$} (vv2);

\draw (v3)--node [dnode] {$0$} (vv3);
\draw (v4)--node [dnode,below] {$0$} (vv4);
\end{tikzpicture}
\caption{A 1-barrier $B$ for $\agentset$ for the case that $l\in \{1,2\}$.}
\label{fig_0_barrier_construction}
\end{figure}

We proceed to generalize this construction towards arbitrary starting states and collaborating agents.

\begin{lemma}
\label{lemma_1_barrier}
For every set of $k$ collaborating $s$-state agents $\agentset$, there exists a \mbox{$1$-barrier $B$} with $\bigO(k s^2)$ vertices. Moreover, $B$ remains a~$1$-barrier even if for all $i \in \{1,\ldots,k\}$ agent~$A_i$ starts in an arbitrary state~$\agentstate \in \Agentstate_i$ instead of the starting state~$\agentstartstate_i$.
\end{lemma}

\begin{proof}
Let $\agentset=\{A_1,\ldots, A_k\}$, let $\Agentstate_i$ be the set of states of $\agent_i$ and let $\smash{\agentstartstate_i}$ be its starting state. 
For all $i \in \{1,\ldots, k\}$ and all $\agentstate \in \Agentstate_i$, we define agent $\smash{\agent_i^{(\agentstate)}}$ to be the agent with the same behavior as~$\agent_i$, but starting in state $\agentstate$ instead of $\agentstartstate_i$. That is, $\smash{\agent_i^{(\agentstate)}}$ has the same set of states $\Agentstate_i$ as $\agent_i$ and it transitions according to the function~$\delta_i$ of $\agent_i$.
Moreover, let $S:= \{ \smash{\agent_i^{(\agentstate)}} \mid i \in \{1,\ldots, k\},\ \agentstate \in \Agentstate_i\}$.

Applying Theorem~\ref{fraigniaud_trap_q_agents} for the set of agents $S$ yields a graph $H$ with an edge $\{v_1,v_2\}$ labeled 0 and an edge $\{v_3,v_4\}$ labeled $l \in \{0,1,2\}$ so that any agent $\smash{\agent_i^{(\agentstate)}}$ that starts  in $v_1$ or $v_2$ does not traverse the edge $\{v_3,v_4\}$. 
Let $B$ be the graph consisting of two connected copies of $H$ and 8 additional vertices, as illustrated in Figure~\ref{fig_0_barrier_construction}. 
The edges $\{v_1,v_2\}$ and $\{v'_1,v'_2\}$ are replaced by $\{v_1,v'_1\}$ and $\{v_2,v'_2\}$, which are also labeled $0$. 
The edges $\{v_3,v_4\}$ and $\{v'_3,v'_4\}$ with label $l$ are deleted and $v_3$ and $v_4$ are connected each to one of the two two-degree vertices of a diamond graph by an edge with label $l$. The same connection to a diamond graph is added for $v'_3$ and $v'_4$ as shown in Figure~\ref{fig_0_barrier_construction}.
The edge labels of the two diamond graphs are arbitrary. Since each diamond graph has two vertices of degree three, each diamond graph has at least one edge with label $0$. We choose one edge with label $0$ and call the end vertices $u$ and $v$ (resp.~$u',v'$).
Note that in Figure~\ref{fig_0_barrier_construction} we have $l \in \{1,2\}$; for the case that $l=0$ the edge $\{u,v\}$ is the unique edge between the two vertices that are not adjacent to $v_3$ or $v_4$.

We claim that $B$ is a $1$-barrier for $\agentset$ with the distinguished edges $\{u,v\}$ and $\{u',v'\}$.
Assume for the sake of contradiction, that the first property does not hold, i.e.,  there is a graph~$G$ that can be connected to $B$ via the pairs of vertices $\{u,v\}$ and $\{u',v'\}$ so that if the agents $\agentset$ start in $G$ in an arbitrary state, there is an agent~$\agent_j$ that walks (without loss of generality) from $u$ to $u'$ in $B$ while there are no other agents in~$B$.
Then~$\agent_j$ in particular walks from $v'_1$ or $v'_2$ to $v'_3$ or $v'_4$ in $H'$ and starts this walk in a state $\agentstate\in \Agentstate_j$. 
But the traversal sequence of
$\agent_j$ in $H'$ is the same as that of $\smash{\agent_j^{(\agentstate)}}$ that starts at $v'_1$ or $v'_2$. This would imply that $\smash{\agent_i^{(\agentstate)}}$ traverses the edge $\{v_3,v_4\}$ in the original graph $H$ when starting in $v_1$ or $v_2$, which contradicts Theorem~\ref{fraigniaud_trap_q_agents}.

To prove the second property of a 1-barrier, assume that there is a set of two agents,
such that both enter $B$ during the traversal and one of them exits $B$ via~$u$ or~$v$ and the other via~$u'$ or~$v'$. 
But then again one of the agents must have traversed $H$ starting in $v_1$ or $v_2$ in a state $\agentstate$ and finally traversed the edge with label $l$ incident to $v_3$ or $v_4$ or similarly in~$H'$ with $v'_1,v'_2,v'_3,v'_4$. This leads to the same contradiction as above.

The whole proof does not use the specific starting states of the agents $\agentset$ and, in particular, the definition of $S$ is independent of the starting states of the agents. 
Consequently, $B$ is a 1-barrier for $\agentset$ even if we change the starting states of the agents.

Since every agent has $s$ states, we obtain that the cardinality of $S$ is bounded by $\bigO(ks)$ and, hence, the graph $B$ has $\bigO(ks^2)$ vertices by Theorem~\ref{fraigniaud_trap_q_agents}.
\end{proof}

The proof of Theorem~\ref{fraigniaud_trap_q_agents} in \cite{fraigniaud06}
uses the fact that when traversing a 3-regular graph the next state of an $s$-state agent only depends on the previous state and the label $l \in \{0,1,2\}$ of the edge leading back to the previous vertex.
Thus, after at most $3 s$ steps, the state of the agent and therefore also the next label chosen
need to repeat with a period of length at most $3 s$.
For cooperative agents, however, the next state and label that are chosen may also depend on the positions and states of the other agents.
We therefore need to account for the positions of all agents when forcing them into a periodic behavior. To this end, we will consider the relative positions of the agents with respect to a given vertex $v$. For our purposes, it is sufficient to define the relative position of an agent $A_i$ by the shortest traversal sequence leading from $v$ to the location of~$A_i$. 
This motivates the following definition.

\begin{definition}
The \emph{configuration} of a set of~$k$~agents~$\agentset=\{A_1,\ldots, A_k\}$ in a graph $G$ \emph{with respect to a vertex}~$v$ is a $(3k)$-tuple $(\agentstate_1,l_1,r_1,\agentstate_2,l_2,r_2,\ldots,\agentstate_k,l_k,r_k)$, where $\agentstate_i$ is the current state of $\agent_i$, $l_i$ is the label of the edge leading back to the previous vertex visited by $\agent_i$ and~$r_i$ is the
shortest traversal sequence from $v$ to $A_i$, where ties are broken in favor of lexicographically smaller sequences and where we set~$r_i = \bot$ if the location of~$\agent_i$ is~$v$.
\end{definition}

In order to limit the number of possible configurations, we will force the agents to stay close together. Intuitively, we can achieve this for any graph~$G$ by replacing all edges with $(k-1)$-barriers.
This way, only all agents together can move between neighboring vertices of the original graph~$G$. To formalize this, we first need to explain how edges of a graph can be replaced by barriers.
Since our construction may not be 3-regular, we need a way to extend it to a 3-regular graph.

\begin{definition}
Given a graph $G$, with vertices of degrees 2 and 3, we define the \emph{3-regular extension}~$\regExt{G}$ as the graph resulting from copying $G$ and connecting every vertex $v$ of degree 2 to its copy~$v'$. As the edges incident to $v$ and $v'$ have the same labels, it is possible to label the new edge $\{v,v'\}$ with a locally unique label in $\{0,1,2\}$.
\end{definition} 
 
Note that the 3-regular extension only increases the number of vertices of the graph by a factor of 2. 
Given a 3-regular graph $G$ and an $r$-barrier $B$ for a set of~$k$ agents~$\agentset$  with $k \geq r $, we replace edges of~$G$ using the following construction.
First, for every $l \in \{0,1,2\}$ we replace every edge $\{a,b\}$ labeled $l$ with the gadget $B(l)$ shown in Figure~\ref{fig_replace_edge_with_barrier}, and we call the resulting graph~$G_1(B)$. 
By construction, the labels of the edges incident to the same vertex in $G_1(B)$ are distinct. 
However, certain vertices only have degree 2. 
We take the 3-regular extension of $G_1(B)$ and define the resulting graph as $G(B):=\regExt{G_1(B)}$.

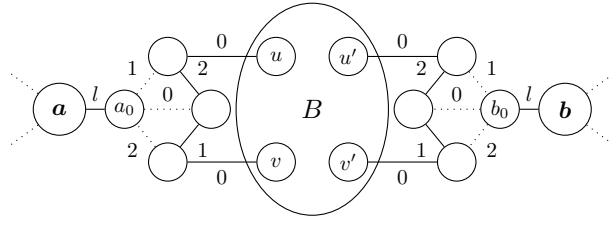
\begin{figure}
\centering
\tikzstyle{graphnode}=[circle,draw,minimum size=18pt,scale=0.8]
\usetikzlibrary{shapes}
\tikzstyle{dnode}=[scale=0.8]
\newcommand*{\diffX}{0.95}
\newcommand*{\diffY}{0.7}
\begin{tikzpicture}[auto=left]
\node(g1) at (0,0) [ellipse,draw,minimum width=57pt,minimum height=80pt] {$B$};
\node (v1) at (-0.5*\diffX,\diffY) [graphnode] {$u$};
\node (v2) at (-0.5*\diffX,-\diffY) [graphnode] {};
\node (d2) at (-0.5*\diffX,-\diffY) [dnode] {$v$};
\node (v3) at (0.5*\diffX,\diffY) [graphnode] {};
\node (lv3) at (0.5*\diffX,\diffY) [dnode] {$u'$};
\node (v4) at (0.5*\diffX,-\diffY) [graphnode] {};
\node (d4) at (0.5*\diffX,-\diffY) [dnode] {$v'$};

\node (w1) at (-2*\diffX,\diffY) [graphnode] {};
\node (w3) at (-2*\diffX,-\diffY) [graphnode] {};
\node (w4) at (-1.4*\diffX,0) [graphnode] {};
\node (w2) at (-2.6*\diffX,0) [graphnode] {};
\node (d2) at (-2.6*\diffX,0) [dnode] {$a_0$};

\draw (v1)--node [dnode,above] {$0$} (w1);
\draw (v2)--node [dnode] {$0$} (w3);
\draw[dotted] (w2)--node [dnode] {$1$} (w1);
\draw[dotted] (w3)--node [dnode] {$2$} (w2);
\draw (w4)--node [dnode] {$1$} (w3);
\draw (w1)--node [dnode] {$2$} (w4);
\draw[dotted] (w2)--node [dnode] {$0$} (w4);

\node (u1) at (2*\diffX,\diffY) [graphnode] {};
\node (u3) at (2.0*\diffX,-\diffY) [graphnode] {};
\node (u4) at (1.4*\diffX,0) [graphnode] {};
\node (u2) at (2.6*\diffX,0) [graphnode] {};
\node (d2) at (2.6*\diffX,0) [dnode] {$b_0$};

\draw (v3)--node [dnode,above] {$0$} (u1);
\draw (v4)--node [dnode,below] {$0$} (u3);
\draw[dotted] (u1)--node [dnode] {$1$} (u2);
\draw[dotted] (u2)--node [dnode] {$2$} (u3);
\draw (u3)--node [dnode] {$1$} (u4);
\draw (u4)--node [dnode] {$2$} (u1);
\draw[dotted] (u2)--node [dnode,above] {$0$} (u4);

\node (a) at (-3.5*\diffX,0)[circle,draw,minimum size=2em,scale=0.9] {$\bm{a}$};
\node (b) at (3.5*\diffX,0)[circle,draw,minimum size=2em,scale=0.9] {$\bm{b}$};
\draw(w2)--node [dnode,above] {$l$}(a);
\draw (u2)--node [dnode] {$l$}(b);
\draw[dotted](a)-- (-4.2*\diffX,0.7*\diffY);
\draw[dotted](a)-- (-4.2*\diffX,0.7*-\diffY);
\draw[dotted](b)-- (4.2*\diffX,0.7*\diffY);
\draw[dotted](b)-- (4.2*\diffX,0.7*-\diffY);
\end{tikzpicture}
\caption{An edge $\{a,b\}$ labeled $l$ is replaced with the gadget $B(l)$ containing an $r$-barrier $B$. Only the dotted edges incident to $a_0$ and $b_0$ that are not labeled $l$ are part of the gadget. Consequently, the gadget contains two vertices of degree 2. The vertices $a$ and $b$ are macro vertices of the graph $G(B)$.}
\label{fig_replace_edge_with_barrier}
\end{figure}

The graph $G(B)$ contains two copies of $G_1(B)$. 
To simplify exposition, we identify each vertex~$v$ with its copy $v'$ in $G(B)$. 
Then, there is a canonical bijection between the vertices in $G$ and the vertices in $G(B)$ which are not part of a gadget $B(l)$. 
These vertices can be thought of as the original vertices of~$G$, and we call them \emph{macro vertices}.

We now establish that the agents always stay close to each other in the graph $G(B)$.

\begin{lemma}
\label{lem_need_all_pebbles_to_cross_barrier}
Let $G$ be a connected $3$-regular graph and let~$B$ be a $(k-1)$-barrier for a set of $k$ agents~$\agentset$ with $s$ states each. Then, the following statements hold for the graph $G(B)$:
\begin{enumerate}
\item For all edges $\{v,v'\}$ in $G$ no strict subset $\agentset' \subsetneq \agentset$ of the agents can get from macro vertex~$v$ to macro vertex~$v'$ in $G(B)$ without all other agents also entering the gadget $B(l)$ between $v$ and $v'$, where $l \in \{0,1,2\}$.
\item At each step of the walk of $\agentset$ in $G$, there is some macro vertex $v$ such that all agents are at $v$ or in one of the surrounding gadgets  $B(0)$, $B(1)$ and $B(2)$.
\end{enumerate}
\end{lemma}

\begin{proof}
For the sake of contradiction, assume that there is a strict subset of agents $\agentset' \subsetneq \agentset$ that walks from a macro vertex~$v$ in $G(B)$ to a distinct macro vertex $v'$ without the other agents entering the gadget between $v$ and $v'$ at any time during the traversal. The graph $G(B)$ contains two copies of $G_1(B)$, but all vertices in the \mbox{$(k-1)$-barriers} within $G_1(B)$ have degree 3. Thus, $\agentset'$ must have traversed some $(k-1)$-barrier $B$ while only agents
in $\agentset'$ enter $B$ at any time of the traversal. This is a contradiction, as $|\agentset'|\leq k-1$ and $B$ is a $(k-1)$-barrier. 
Therefore, the agents $\agentset$ need to all enter the gadget between $v$ and $v'$ to  to get from a macro vertex $v$ to a distinct macro vertex $v'$. This shows the first claim.

For the second part of the claim, note that because of Property~2 for the barrier $B$ the agents cannot split up into two groups such that after the traversal of the gadget between $v$ and $v'$ one group is at $v$ (or one of the vertices at distance at most 4 from $v$ which are not part of the barrier $B$) and the other group is at $v'$ (or one of the vertices at distance 4 from $v'$ which are not part of the barrier $B$). This implies that if we consider the positions of the agents after an arbitrary number of steps and let $v$ be the macro vertex last visited by an agent in $\agentset$, then all agents must be located at~$v$ or one of the three surrounding gadgets.
\end{proof}

\begin{figure}
\centering
\tikzstyle{graphnode}=[circle,draw,minimum size=18pt,scale=0.8]
%\tikzstyle{dnode}=[scale=0.7]
\tikzstyle{gadget}=[rectangle,draw,minimum height=21pt,minimum width=32pt,scale=0.8]
\newcommand*{\x}{0.9}
\begin{tikzpicture}
\node(v0) at (0,0) [graphnode]{$v$};
\node(g1) at (0:\x) [gadget]{};
\node(d1) at (0:\x)[dnode] {$B(0)$};
\node(v1) at (0:2*\x) [graphnode]{};
\node(g2) at (120:\x) [gadget,rotate around={120:(0,0)}]{};
\node(d2) at (120:\x) [dnode]{$B(1)$};
\node(v2) at (120:2*\x) [graphnode]{};
\node(g3) at (240:\x) [gadget,rotate around={240:(0,0)}]{};
\node(d3) at (240:\x) [dnode]{$B(2)$};
\node(v3) at (240:2*\x) [graphnode]{};

\draw(v0)--(g1);
\draw(v0)--(g2);
\draw(v0)--(g3);
\draw(v1)--(g1);
\draw(v2)--(g2);
\draw(v3)--(g3);

\node(v4) at (10:2.7*\x){};
\node(v5) at (-10:2.7*\x) {};
\node(v6) at (110:2.7*\x){};
\node(v7) at (130:2.7*\x){};
\node(v8) at (230:2.7*\x){};
\node(v9) at (250:2.7*\x){};

\draw[dotted] (v1)--(v4);
\draw[dotted] (v1)--(v5);
\draw[dotted] (v2)--(v6);
\draw[dotted] (v2)--(v7);
\draw[dotted] (v3)--(v8);
\draw[dotted] (v3)--(v9);
\end{tikzpicture}
\caption{A macro vertex $v$ in a graph $G(B)$ surrounded by the three gadgets $B(0),B(1)$ and $B(2)$. }
 \label{fig_macro_vertex}
\end{figure}
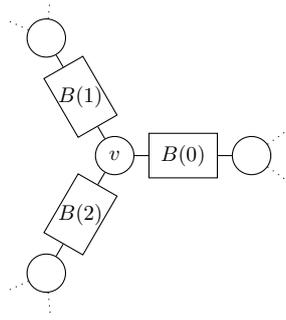

We will frequently consider the configuration of~$\agentset$ in a graph of the from $G(B)$ with respect to some macro vertex $v$. Recall from the definition that the graph $G(B)$ contains two copies of the graph $G_1(B)$ and actually there exists a macro vertex $v$ and a copy $v'$. Thus, when we talk about configurations of~$\agentset$ in $G(B)$ with respect to some macro vertex $v$, we mean that we consistently choose one of the copies $G_1(B)$ and consider the configuration of~$\agentset$ with respect to the macro vertex in this copy.

Let $B$ be a $(k-1)$-barrier for a set of $k$ cooperative $s$-state agents $\agentset=\{A_1,\ldots,A_k\}$ that all start in some macro vertex $v_0$ of $G(B)$. 
Iteratively, define $t_0=0$ and $t_i$ to be the first point in time after $t_{i-1}$, when one of the agents in $\agentset$ visits a macro vertex $v_i$ distinct from $v_{i-1}$. Then $v_i$ is a neighbor of $v_{i-1}$ in $G$ and by Lemma~\ref{lem_need_all_pebbles_to_cross_barrier}, all agents are at $v_i$ or one of the incident gadgets.
The sequence of macro vertices~$v_0, v_1, \ldots$, which is a sequence of neighboring vertices in $G$, yields a unique sequence of labels $l_0,l_1,\ldots$
of the edges between the neighboring vertices in $G$, which we call the \emph{macro traversal sequence} of~$\agentset$ starting in vertex $v_0$ in $G(B)$. Note that the macro traversal sequence may be finite. 

Consider the traversal sequence $l_0,l_1, \ldots$ of a single agent in a 3-regular graph~$G$ and the traversal sequence $l'_0,l'_1, \ldots$ of the same agent in another 3-regular graph~$G'$. If the state of the agent and label of the edge to the previous vertex 
in~$G$ after~$i$ steps is the same as the state in~$G'$ 
after~$j$~steps, then the traversal sequences coincide from that point on, i.e., $l_{i+h}=l'_{j+h}$ holds for all $h \in \N$. 
The reason is that the graphs we consider are 3-regular and the label of every edge $\{u,v\}$ is the same at $u$ and at~$v$. 
Therefore, once the state and label to the previous vertex are the same, the agent makes the same transitions as it can gain no new information while traversing the graph.
We want to obtain a similar result for a set of agents. 
However, in general it is not true that if the configuration of a set of agents in a graph $G$ after $i$ steps is the same as after $j$ steps in $G'$, then the next configurations and chosen labels of each agent coincide. 
This is because an agent can be used to mark a particular vertex and this can be used to detect differences in two 3-regular graphs~$G$ and $G'$. 
For instance, one agent could remain at a vertex $v$ while the other one walks in a loop that is only part of one of the graphs and this may lead to different configurations. 
That is why we consider graphs of the form $G(B)$. 
In these graphs, all macro vertices look the same, as they are surrounded by the same gadgets,
and the agents have to stay close together, making it impossible for the agents to detect a loop that is part of one of the graphs, but not the other. 
This intuition is formally expressed in the following technical lemma.

\begin{lemma}
\label{lemma_same_conf_same_labels}
 Let $B$ be a~$(k-1)$-barrier for a set of $k$ $s$-state agents~$\agentset$, and let $G$ and~$G'$ be two 3-regular graphs.
 Let $v_0,v_1,\ldots$ be the sequence of macro vertices visited by~$\agentset$ in~$G(B)$, let $l_0,l_1,\ldots$ be the corresponding macro traversal sequence, let $t_0=0$, and let~$t_i$ be the first time after~$t_{i-1}$ that an agent in~$\agentset$ visits~$v_i$. 
Let~$v'_0,v'_1,\ldots$ and $l'_0,l'_1,\ldots$ and $t'_i$ be defined analogously with respect to~$G'(B)$.
If there are $t \in \{t_i,\ldots, t_{i+1}-1\}$ and $t' \in \smash{\{t'_j,\ldots, t'_{j+1}-1\}}$ for some $i,\ j \in \N$, such that after $t$ steps in~$G(B)$ the configuration of~$\agentset$ with respect to~$v_i$ is the same as after $t'$ steps in~$G'(B)$ with respect to~$v'_j$,
  then:
\begin{itemize}
\item ~$l_{i+h} = l'_{j+h}$ holds for all~$h \in \N$,
\item the configuration of~$\agentset$ in $G(B)$ after $t_{i+h}$ steps with respect to~$v_{i+h}$ is the same as configuration of~$\agentset$ in $G'(B)$ after $t'_{j+h}$ steps with respect to~$v'_{j+h}$
 for all~$h \in \N$, $h>0$.
\end{itemize}  
\end{lemma}

\begin{proof}
In order to simplify the notation of the proof, we abuse notation and overwrite the definition of $t_i$ and $t'_j$ by setting $t_i:=t$, $t'_j:=t'$.
By induction on $h\in \N$, we show that the configuration of~$\agentset$ after~$t_{i+h}$ steps in $G(B)$ with respect to~$v_{i+h}$ is the same as the configuration of~$\agentset$ after~$t'_{j+h}$ steps in $G'(B)$ with respect to~$v'_{j+h}$. The induction step also shows that we have~$l_{i+h} = l'_{j+h}$ for all $h \in \N$. 

For~$h=0$ we have by assumption (and as we redefined $t_i$ and $t'_j$) that after $t_i$ steps in $G(B)$ the configuration of~$\agentset$ with respect to~$v_i$ is the same as after $t'_j$ steps in $G'(B)$ with respect to~$v'_j$.

Now, assume that the statement holds for some~$h \in \N$.
The idea of the proof is that, in between visits to macro vertices, the agents behave the same in the two graphs and, in particular, they traverse the same gadget~$B(l)$ in both settings in such that~$l_{i+h} = l'_{j+h}$.

The graphs~$G(B)$ and~$G'(B)$ locally look the same to the agents in~$v_{i+h}$ and~$v'_{j+h}$ as both macro vertices are surrounded by the same gadgets, as shown in Figure~\ref{fig_macro_vertex}. Formally, there is a canonical graph isomorphism~$\gamma$ from the induced subgraph of~$G(B)$ containing~$v_{i+h}$ and all surrounding gadgets to the induced subgraph of~$G'(B)$ containing~$v'_{j+h}$ and all surrounding gadgets. Moreover,
$\gamma$ respects the labeling and maps~$v_{i+h}$ to~$v'_{j+h}$.
As the configuration of~$\agentset$ after~$t_{i+h}$ steps
with respect to~$v_{i+h}$ is the same as the configuration of~$\agent$ after~$t'_{j+h}$ steps with respect to~$v'_{j+h}$, the
isomorphism also respects the positions of all the agents. 
As~$v_{i+h+1}$ is the first macro vertex visited after~$v_{i+h}$, all agents are
at~$v_{i+h}$ or any of the surrounding gadgets until the agents~$\agentset$ reach~$v_{i+h+1}$ by  Lemma~\ref{lem_need_all_pebbles_to_cross_barrier}. The same holds for~$v'_{j+h}$ and~$v'_{j+h+1}$. Iteratively, for~$c=0,1,\ldots $ the following holds until the agents reach the next macro vertex~$v_{i+h+1}$ or~$v'_{j+h+1}$:
\begin{enumerate}
\item For every agent $\agent \in \agentset$, the state of~$\agent$ and the edge label to the previous vertex after~$t_{i+h}+c$ steps in~$G(B)$ is the same as
the state of~$\agent$ and the edge label to the previous vertex after~$t'_{j+h}+c$ steps in~$G'(B)$.
\item The isomorphism~$\gamma$ maps the position of every agent $\agent \in \agentset$ after~$t_{i+h}+c$ steps in~$G(B)$ to the
position of~$\agent$ after~$t'_{j+h}+c$ steps in~$G'(B)$.
\end{enumerate}
This implies that macro vertices $v_{i+h}$ and~$v_{i+h+1}$  are connected with the same gadget as~$v'_{j+h}$ and~$v'_{j+h+1}$,
i.e.,~$l_{i+h} = l'_{j+h}$. Furthermore, there is~$\bar{c}$ such that~$t_{i+h+1}=t_{i+h}+\bar{c}$ and~$t'_{j+h+1}=t'_{j+h}+\bar{c}$.
Moreover, the configuration of~$\agentset$ with respect to~$v_{i+h+1}$ after~$t_{i+h+1}$ steps is the same as with respect to~$v'_{j+h+1}$ after~$t'_{j+h+1}$ steps.
\end{proof}

Let $2 \leq r \leq k$.
In order to construct an~$r$-barrier~$B'$ for a set~$\agentset$ of~$k$ cooperative~$s$-state agents given an~$(r-1)$-barrier~$B$, we need to examine the behavior of all subsets of~$r$ agents. There are~$\binom{k}{r}$ subsets of~$r$ agents and the behavior of two different subsets of~$r$ agents may be completely different. We denote these~$\smash{\binom{k}{r}}$ subsets of~$r$ agents by~$\agentset^{(r)}_{1},\ldots,  {\agentset^{(r)}_{\binom{k}{r}}}$.

Assume, we have an $(r-1)$-barrier $B$ for a set of $k$ agents~$\agentset$.
For $1 \leq j \leq \binom{k}{r}$, consider the behavior of only the subset of agents $\smash{\agentset^{(r)}_{j}}$ in a graph of the form~$G(B)$.
Let~$v_0,v_1,\ldots$ be the sequence of macro vertices,~$l_0,l_1,\ldots$ the corresponding macro label sequence, $t_0 = 0$, and~$t_i$ be the first time after~$t_{i-1}$ that an agent in~$\smash{\agentset^{(r)}_{j}}$ visits~$v_i$. 
Between steps~$t_{i-1}$ and~$t_i$ all agents are located at~$v_{i-1}$ or
one of the  surrounding gadgets~$B(0), B(1), B(2)$ by Lemma~\ref{lem_need_all_pebbles_to_cross_barrier}. 
Thus, the number of possible locations of the agents can be bounded in terms of the size of the gadgets $B(0)$, $B(1)$, and $B(2)$. In addition, every agent has at most $s$ states. Therefore the number of configurations of~$\smash{\agentset^{(r)}_{j}}$ with respect to~$v_i$ between steps~$t_{i-1}$ and~$t_i$ can also be bounded in terms of $s$ and the size of the gadgets. In particular, this bound is independent of the specific subset of agents~$\smash{\agentset^{(r)}_{j}}$.
For a sufficiently large number of steps, a configuration must repeat and, by applying Lemma~\ref{lemma_same_conf_same_labels} for $G=G'$, the macro label sequence becomes periodic.
The other crucial property that follows from Lemma~\ref{lemma_same_conf_same_labels} is that the macro label sequence is independent of the underlying 3-regular graph~$G$. As a consequence, we may denote by $\alpha_B$ the maximum over all  $j \in \bigl\{1  \ldots \binom{k}{r}\bigr\}$ of the number of steps in the macro label sequence  until~$\smash{\agentset^{(r)}_{j}}$ is twice in the same configuration in $G(B)$ with respect to two macro vertices, i.e.,
there are $a,\ b \leq \alpha_B$ such that
the configuration of $\smash{\agentset^{(r)}_{j}}$ at $t_a$ with respect to $v_a$ is the same as at $t_b$ with respect to $v_b$. Note that the value of $\alpha_B$ depends on the size of the barrier $B$ and thus also on the values of $s$ and $r$.

Given the definition of $\alpha_B$, we are now in position to present the construction of an~$r$-barrier given an~$(r-1)$-barrier. We will later bound $\alpha_B$ and, thus, the size of the $r$-barrier in Lemma~\ref{lemma_bound_configurations2}.

\begin{theorem}
\label{theo_induction_barrier}
Given an~$(r-1)$-barrier~$B$ with~$n$ vertices for a set~$\agentset$ of~$k$~agents with~$s$ states each, we can construct an~$r$-barrier~$B'$ for~$\agentset$ with the following properties:
\begin{enumerate}
\item We have $B'=H(B)$ for a suitable 3-regular graph $H$.
\item If $\{u,v\}$ and $\{u',v'\}$ are the two distinguished edges of $B'$, then any path from $u$ or $v$ to $u'$ or $v'$ contains at least 3 distinct barriers $B$.
\item The~$r$-barrier~$B'$ contains at most $\bigO(\binom{k}{r} \cdot n \cdot \alpha_B^2)$ vertices.
\end{enumerate}
\end{theorem}

\begin{proof} 
For~$j \in \bigl\{1,2,\ldots, \smash{\binom{k}{r}}\bigr\}$, consider a subset of~$r$ agents
$\smash{\agentset^{(r)}_j}$ starting at a vertex~$v_0$ in a graph~$G(B)$. 
Let~$t_0=0$ and for~$i=1,2,\dots$ iteratively define~$t_i$ to be the first point in time after~$t_{i-1}$, when an agent in~$\smash{\agentset^{(r)}_j}$ visits a macro vertex~$v_i$ distinct from~$v_{i-1}$. 
Then~$v_0,v_1,\ldots$ is the macro label sequence of~$\smash{\agentset^{(r)}_j}$ in~$G(B)$ with a corresponding macro label sequence~$l_0,l_1,\ldots$. 
After at most~$\alpha_B$ steps, the agents in~$\smash{\agentset^{(r)}_j}$ are twice in the same configuration with respect to two macro vertices, i.e., there are~$a,b \in \N$ with~$a< b \leq \alpha_B$ such that after~$t_a$ steps the configuration of~$\smash{\agentset^{(r)}_j}$ with respect to~$v_a$ is the same as after~$t_b$ steps with respect to~$v_b$. 
Note that~$\alpha_B$ is a bound on the maximum possible number of steps until the configuration repeats and therefore independent of the specific subset of agents~$\smash{\agentset^{(r)}_j}$. 
The possible configurations of~$\agentset$ at times $t_0, t_1, \ldots$ can hence be enumerated~$x_1,\ldots, x_{\alpha_B}$.

By Lemma~\ref{lemma_same_conf_same_labels}, the configuration of the set of agents~$\smash{\agent^{(r)}_j}$ uniquely determines the next label in the macro label sequence of~$\agent^{(r)}_j$, independently of the underlying graph~$G$.
We can therefore define a single agent~$\smash{\bar{A}_j}$
whose state corresponds to the configuration of the set of agents~$\smash{\agent^{(r)}_j}$ and whose label sequence is the macro label sequence of~$\smash{\agent^{(r)}_j}$. More precisely we define~$\smash{\bar{A}_j}$ as follows:
The set of states of~$\smash{\bar{A}_j}$ is~$\{\agentstate_1,\ldots,\agentstate_{\alpha_B}\}$. Moreover, in state~$\agentstate_h$ the agent~$\smash{\bar{A}_j}$ traverses the edge labeled~$l$ and transitions to~$\agentstate_{h'}$ if the set of agents~$\agentset^{(r)}_j$ in configuration~$x_{h}$ at a time~$t_i$ will traverse the gadget~$B(l)$ to the next vertex $v_{i+1}$ in the macro vertex sequence where it arrives in configuration~$x_{h'}$ at time~$t_{i+1}$ (this means that~$l=l_i$ is the next label in the macro label sequence of~$\smash{\agentset^{(r)}_j}$ in configuration~$x_h$). 
The starting state of~$\smash{\bar{A}_j}$ corresponds to the 
configuration, where all the agents in~$\smash{\agentset^{(r)}_j}$ are in their starting states and located at the same vertex. 
Note that the transition function~$\smash{\bar{\delta}}$ of $\smash{\bar{A}_j}$ described above is well-defined because, by Lemma~\ref{lemma_same_conf_same_labels}, the next label $l_i$ in the macro label sequence of~$\smash{\agentset^{(r)}_j}$ only depends on the configuration of~$\smash{\agentset^{(r)}_j}$ at $t_i$ and is independent of the underlying graph~$G$. 
By construction, the macro traversal sequence of~$\smash{\agentset^{(r)}_j}$ in~$G(B)$ is exactly the same as the traversal sequence of~$\bar{A}_j$ in~$G$, independently of the graph~$G$.
Applying Lemma~\ref{lemma_1_barrier} for the single agent~$\bar{A}_j$, we obtain a~$1$-barrier~$H_j$ with~$\bigO(\alpha_B^2)$ vertices that cannot be traversed by~$\bar{A}_j$, irrespective of its starting state.

\begin{figure}
\centering
\tikzstyle{graphnode}=[circle,draw,minimum size=16pt,scale=0.7]
\tikzstyle{dnode}=[scale=0.7]
\tikzstyle{dnode2}=[scale=0.9]
\newcommand*{\diffX}{0.35}
\newcommand*{\diffY}{0.55}
\newcommand*{\xOff}{1.8}
\begin{tikzpicture}
\node(spacingnode) at (0,2.5*\diffY) {};  
\node(g1) at (0,0) [ellipse,draw,minimum width=45pt,minimum height=57pt] {};
\node(1) at (0,0) [dnode2]{$H_1$};
\node (v1) at (-\diffX,\diffY) [graphnode] {};
\node (v2) at (-\diffX,-\diffY) [graphnode] {};
\node (v3) at (\diffX,\diffY) [graphnode] {};
\node (v4) at (\diffX,-\diffY) [graphnode] {};

\node (a1) at (-\diffX-0.4*\xOff,0) [graphnode] {};
\node (d0) at (-\diffX-0.4*\xOff,0)[dnode] {$v$};
\node (a2) at (-\diffX-0.65*\xOff,-\diffY) [graphnode] {};
\node (a3) at (-\diffX-0.65*\xOff,\diffY) [graphnode] {};
\node (a4) at (-\diffX-0.9*\xOff,0) [graphnode] {$u$};

\draw(a1) --node [dnode,above] {$0$}(a4);
\draw(a2) --node [dnode,below] {$0$}(v2);
\draw(a3) --node [dnode,above] {$0$}(v1);
\draw(a1) --node [dnode,above] {}(a2);
\draw(a1) --node [dnode,above] {}(a3);
\draw(a4) --node [dnode,above] {}(a2);
\draw(a4) --node [dnode,above] {}(a3);

\node(gg1) at (\xOff,0) [ellipse,draw,minimum width=45pt,minimum height=57pt] {};
\node(1) at (\xOff,0) [dnode2]{$H_2$};
\node (vv1) at (-\diffX+\xOff,\diffY) [graphnode] {};
\node (vv2) at (-\diffX+\xOff,-\diffY) [graphnode] {};
\node (vv3) at (\diffX+\xOff,\diffY) [graphnode] {};
\node (vv4) at (\diffX+\xOff,-\diffY) [graphnode] {};

\node(ggg1) at (2.2*\xOff,0) [ellipse,draw,minimum width=45pt,minimum height=57pt] {};
\node(1) at (2.2*\xOff,0) [dnode2]{$H_{\binom{k}{r}}$};
\node (vvv1) at (-\diffX+2.2*\xOff,\diffY) [graphnode] {};
\node (vvv2) at (-\diffX+2.2*\xOff,-\diffY) [graphnode] {};
\node (vvv3) at (\diffX+2.2*\xOff,\diffY) [graphnode] {};
\node (vvv4) at (\diffX+2.2*\xOff,-\diffY) [graphnode] {};

\node (b1) at (\diffX+2.6*\xOff,0) [graphnode] {};
\node (lb1) at (\diffX+2.6*\xOff,0) [dnode] {$u'$};
\node (b2) at (\diffX+2.85*\xOff,-\diffY) [graphnode] {};
\node (b3) at (\diffX+2.85*\xOff,\diffY) [graphnode] {};
\node (b4) at (\diffX+3.1*\xOff,0) [graphnode] {};
\node (d0) at (\diffX+3.1*\xOff,0) [dnode] {$v'$};

\draw(b1) --node [dnode,above] {$0$}(b4);
\draw(b2) --node [dnode,below] {$0$}(vvv4);
\draw(b3) --node [dnode,above] {$0$}(vvv3);
\draw(b1) --node [dnode,above] {}(b2);
\draw(b1) --node [dnode,above] {}(b3);
\draw(b4) --node [dnode,above] {}(b2);
\draw(b4) --node [dnode,above] {}(b3);

\node(u1) at (1.6*\xOff,\diffY){};
\node(u2) at (1.6*\xOff,-\diffY){};

\draw (v3)--node [dnode,above] {$0$}(vv1);
\draw (v4)--node [dnode,below] {$0$}(vv2);

\draw[dotted] (vv3)--node [dnode,above] {$0$}(u1);
\draw[dotted] (vv4)--node [dnode,below] {$0$}(u2);
\draw[dotted] (u1)--node [dnode,above] {$0$}(vvv1);
\draw[dotted] (u2)--node [dnode,below] {$0$}(vvv2);

\end{tikzpicture}
\caption{Connecting the graphs~$H_1,H_2,\ldots, \smash{H_{\binom{k}{r}}}$ to a graph~$H$, yields the~$r$-barrier~$H(B)$.}
\label{fig_row_of_barriers}
\end{figure}
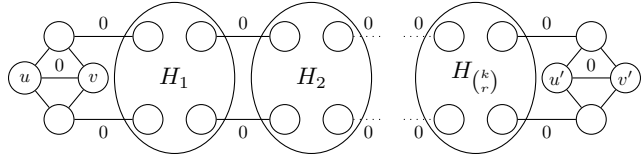

We now connect the graphs~$\smash{H_1,\ldots, H_{\binom{k}{r}}}$ as shown in Figure~\ref{fig_row_of_barriers}, and we let~$H$ denote the resulting graph.
We first show that the graph~$B':=H(B)$ is an~$r$-barrier for~$\agentset$ and and then show the three additional properties in the claim.

For the first property of an $r$-barrier, assume, for the sake of contradiction, that there is a subset of~$r$~agents~$\agentset^{(r)}_j$ and some graph~$G$ connected to~$H(B)$ such that the agents~$\agentset^{(r)}_j$ can traverse~$H(B)$ from~$u$ to~$u'$. Then there must be a consecutive subsequence~$w_0,w_1,\ldots, w_h$  of the macro vertex sequence of ~$\agentset^{(r)}_j$ during the traversal of~$H(B)$ with the following properties: The vertices~$w_1,\ldots, w_{h-1}$ are contained in~$H_j(B)$, $w_0$ and $w_{h}$ are not contained in $H_j(B)$, $w_1$ and $w_{h-1}$ (as vertices in the 1-barrier $H_j$) are incident to different distinguished edges (i.e., $\{u,v\}$ or $\{u',v'\}$ in Fig.~\ref{fig_0_barrier_construction}) of the 1-barrier $H_j$.
Thus, the set of agents~$\agentset^{(r)}_j$ starting in~$w_0$ in a suitable configuration~$x_i$ traverses the graph~$H_j(B)$ from $w_1$ to $w_{h-1}$. This means that for a suitable graph~$G'$ and starting state~$\agentstate_i$ the agent~$\bar{A}_j$ can traverse~$H_j$. But this is a contradiction as we constructed~$H_j$ as a 1-barrier for~$\bar{A}_j$ using Lemma~\ref{lemma_1_barrier} and the 1-barrier~$H_j$ is independent of the starting state of $\bar{A}_j$.

For the second property of an $r$-barrier, let $\agentset' \subseteq \agentset$ be a set of agents with $|\agentset'|=r+1$. 
Assume, for the sake of contradiction, that there is some graph~$G$ connected to~$H(B)$ such that after the agents of~$\agentset'$ (and no other agents) enter $H(B)$ a subset $\emptyset\neq \agentset'_1 \subsetneq \agentset'$ leaves $H(B)$ via $u$ or $v$ and the other agents $\agentset'_2:=\agentset' \setminus   \agentset'_1$ via $u'$ or $v'$. 
Since $B$ is an $(r-1)$-barrier, no set of at most $r-1$ agents can get from a macro vertex to a distinct macro vertex in~$H(B)$. 
Thus, we must have  $|\agentset'_1|\geq r$ or  $|\agentset'_2|\geq r$. 
Without loss of generality, we assume that the first case occurs, which implies $|\agentset'_1 |=r$ and $|\agentset'_2 |=1$. 
For the single agent in~$\agentset'_2$ to leave $H(B)$ via $u'$ or $v'$ at least $r-1$ agents from $\agentset'_1$ must be in a gadget adjacent to $u'$ or $v'$. 
But all these $r-1$ agents afterwards leave $H(B)$ via $u$ or $v$ and they need the remaining agent in $\agentset'_1$ to even get to a distinct macro vertex. 
But then the set of $r$ agents $\agentset'_1$ traverses the subgraphs $H_j(B)$ for all $j \in \{1, \ldots, \binom{k}{r} \}$, which again leads to a contradiction as in the proof for the first property (for $j$ such that $\agentset'_1 = \smash{\agentset_j^{(r)}}$). 

Finally, we obviously have $B'=H(B)$ for a 3-regular graph $H$ by construction and
the second additional property follows from the fact
that any path from $u$ or $v$ to $u'$ or $v'$ in $H$ has length at least 3.
Further, each~$H_j$ contains~$\bigO(\alpha_B^2)$ vertices and therefore~$H$ has at most~$\bigO(\binom{k}{r} \cdot \alpha_B^2)$ vertices. As~$B$ has~$n$ vertices, the number of
vertices of~$B'=H(B)$ is at most~$\smash{\bigO(\binom{k}{r} \cdot  n \cdot \alpha_B^2)}$, where we use that~$H$ is 3-regular and therefore the number of edges of $H$ that are replaced by a copy of $B$ is~$3/2$ times the number of its vertices.
\end{proof}

We now fix a set of~$k$ agents $\agentset$ with $s$ states each and let $B_1$ be the \mbox{$1$-barrier} given by Lemma~\ref{lemma_1_barrier} and~$B_r$ for $1<r\leq k$ be the $r$-barrier constructed recursively using Theorem~\ref{theo_induction_barrier}. 
Moreover, we let~$n_r$ be the number of vertices of $B_r$ and~$\alpha_r:=\alpha_{B_{r-1}}$ be the maximum number of steps in the macro label sequence that a set of $r$ agents from $\agentset$ can execute in a graph of the form  $G(B_{r-1})$ until their configuration repeats.

We want to bound the number of vertices $n_k$ of $B_k$ and thus, according to Lemma~\ref{lem_barrier_to_trap}, also the number of vertices of the trap for $\agentset$. By Theorem~\ref{theo_induction_barrier}, there is a constant $c\in \N$ such that $n_r  \leq  c\smash{\binom{k}{r} n_{r-1}  \alpha_r^2}$. 
In order to bound~$n_r$, we therefore need to bound~$\alpha_r$.

One possible way to obtain an upper bound on $\alpha_r$ is to use Lemma~\ref{lem_need_all_pebbles_to_cross_barrier} stating that there always is a macro vertex $v$ such that all agents are located at $v$ or inside one of the surrounding gadgets. 
Counting the number of possible positions within these three gadgets and states of the agents then gives an upper bound on $\alpha_r$.
For the tight bound in our main result, however, we need a more careful analysis of the recursive structure of our construction and also need to consider the configurations of the agents at specific times. 
We start with the following definition and a technical lemma.

\begin{figure}
\centering
\tikzstyle{graphnode}=[circle,draw,minimum size=16pt,scale=0.7]
\tikzstyle{graphnode2}=[circle,draw,minimum size=8pt,scale=0.7]
\usetikzlibrary{shapes}
\tikzstyle{dnode}=[scale=0.9]
\newcommand*{\offX}{7.5}
\newcommand*{\diffX}{0.3}
\newcommand*{\diffY}{0.4}
\newcommand*{\offA}{2}
\newcommand*{\offB}{0.7}
\newcommand*{\shrink}{0.7}
\begin{tikzpicture}[auto=left]
\node(b1) at (0.5*\offX,0) [ellipse,draw,minimum width=130pt,minimum height=90pt] {};
\node(d1) at  (0.5*\offX,-1) [dnode]{$B_{r-1}$};
\node (w1) at (3*\diffX,\diffY) [graphnode] {};
\node (w3) at (3*\diffX,-\diffY) [graphnode] {};
\node (w4) at (4*\diffX,0) [graphnode] {};
\node (w2) at (2*\diffX,0) [graphnode] {};
\node (v1) at (6.5*\diffX,\diffY) [graphnode] {};
\node (v2) at (6.5*\diffX,-\diffY) [graphnode] {};

\draw (v1)-- (w1);
\draw (v2)--(w3);
\draw (w2)--(w1);
\draw (w3)-- (w2);
\draw (w4)--(w3);
\draw (w1)--(w4);
\draw[dotted](v1)--(6.5*\diffX+0.5,\diffY-0.25);
\draw[dotted](v2)--(6.5*\diffX+0.5,-\diffY+0.25);
\draw[dotted](v2)--(6.5*\diffX+0.5,-\diffY-0.25);

\node (w11) at (\shrink*3*\diffX+\offA,\shrink*\diffY+\offB) [graphnode2] {};
\node (w12) at (\shrink*3*\diffX+\offA,\shrink*-\diffY+\offB) [graphnode2] {};
\node (w13) at (\shrink*4*\diffX+\offA,0+\offB) [graphnode2] {};
\node (w14) at (\shrink*2*\diffX+\offA,0+\offB) [graphnode2] {};
\node (w15) at (\shrink*6*\diffX+\offA+0.1,\shrink*\diffY+\offB) [graphnode2] {};
\node (w16) at (\shrink*6*\diffX+\offA+0.1,\shrink*-\diffY+\offB) [graphnode2] {};
\node(b2) at (\offA+1.83,\offB) [ellipse,draw,minimum width=45pt,minimum height=35pt] {};
\node(d1) at   (\offA+1.8,\offB)[dnode]{$B_{r-2}$};

\draw (w14)-- (v1);
\draw (w14)-- (w11);
\draw (w14)-- (w12);
\draw (w12)-- (w13);
\draw (w11)-- (w13);
\draw (w11)-- (w15);
\draw (w12)-- (w16);

\node (u1) at (-3*\diffX+\offX,\diffY) [graphnode] {};
\node (u3) at (-3*\diffX+\offX,-\diffY) [graphnode] {};
\node (u4) at (-2*\diffX+\offX,0) [graphnode] {};
\node (u2) at (-4*\diffX+\offX,0) [graphnode] {};
\node (v3) at (-6.5*\diffX+\offX,\diffY) [graphnode] {};
\node (v4) at (-6.5*\diffX+\offX,-\diffY) [graphnode] {};

\draw[dotted](v3)--(-6.5*\diffX+\offX-0.5,\diffY+0.25) ;
\draw[dotted](v3)--(-6.5*\diffX+\offX-0.5,\diffY-0.25) ;
\draw[dotted](v4)--(-6.5*\diffX+\offX-0.5,-\diffY+0.25);
\draw[dotted](v4)--(-6.5*\diffX+\offX-0.5,-\diffY-0.25);

\draw (v3)--node [dnode,above] {} (u1);
\draw (v4)--node [dnode,below] {} (u3);
\draw (u1)--node [dnode] {} (u2);
\draw (u2)--node [dnode] {} (u3);
\draw (u3)--node [dnode] {} (u4);
\draw (u4)--node [dnode] {} (u1);

\node (v0) at (0,0)[graphnode] {$v$};
\node (u0) at (\offX,0)[scale=0.7] {};
\node (u0) at (\offX,0)[graphnode] {};
\draw (w2)--node [dnode,above] {}(v0);
\draw (u4)--node [dnode] {}(u0);
\end{tikzpicture}
\caption{Recursive structure of $B(l)$ containing $i$-barriers for $i \in \{1,\ldots, r-1\}$.}
\label{fig_barrier_recursive}
\end{figure}
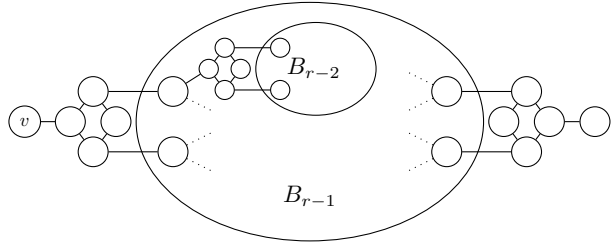

For $j \in \{1,\dots, r-1\}$, we say that a vertex~$w'$ is \emph{$j$-adjacent} to some other vertex~$w$ if there is a path $P$ from~$w$ to~$w'$ that does not traverse a~$j$-barrier~$B_j$, i.e., $P$ does not contain a subpath leading from one vertex of the distinguished edge $\{u,v\}$ to a vertex of the other distinguished edge $\{u',v'\}$ in $B_j$. 
As a convention, every vertex~$w$ is~$j$-adjacent to itself for all $j \in \{1,\dots,r-1\}$. 
Note that a vertex~$w'$ contained inside a~$j$-barrier may be~$j$-adjacent to some vertex $w$ outside the barrier if there is a path from~$w$ to~$w'$ that does not traverse a distinct~$j$-barrier.

\begin{lemma}
\label{lemma-bound-i-adjacent-vertices}
Let $v$ be a macro vertex in $G(B_{r-1})$. Then for $j \in \{1,\ldots, r-1\}$ the number of vertices that are  $j$-adjacent to $v$ is bounded by $2^{4(r-j)} n_j$.
\end{lemma}

\begin{proof}
In order to bound the number of~$j$-adjacent vertices, we examine the recursive structure of one of the gadgets~$B(l)$ incident to~$v$, as shown in Figure~\ref{fig_barrier_recursive}. By Theorem~\ref{theo_induction_barrier} 
an $(r-1)$-barrier $B'$ for $r\geq 3$ is constructed from a 3-regular graph $H$ and an $(r-2)$-barrier $B$ such that $B'=H(B)$. Hence, the gadget $B(l)$, which contains the barrier~$B_{r-1}$, also contains many copies of the barrier~$B_{r-2}$, which again
contain many copies of the barrier~$B_{r-3}$ (if $r\geq 4$) and so on. 

We first observe that the distance from~$v$ to any~$j$-adjacent vertex, which is not contained in a barrier~$B_{j}$, is at most~$3 (r-j)+1$. This observation is clear for~$j=r-1$ and follows for~$r-2, r-3, \ldots$ by examining the recursive structure given in Figure~\ref{fig_barrier_recursive}. As~$G(B_{r-1})$ is 3-regular, there are at most~$2^{3(r-j)+1}$ such vertices. Moreover, any~$j$-barrier~$B_{j}$ containing vertices that are~$j$-adjacent to~$v$,
in particular contains a vertex with a distance of exactly~$3(r-j)$ to~$v$. 
As~$G(B_{r-1})$ is 3-regular, there are at most~$2^{3(r-j)}$ vertices of distance exactly~$3(r-j)$ from~$v$ and therefore at most~$2^{3(r-j)}$ different~$j$-barriers, with~$n_j$ vertices each, containing~$j$-adjacent vertices.
Thus, there are at most~$2^{3(r-j)} n_{j}$ vertices that are~$j$-adjacent to~$v$
and contained in a barrier~$B_{j}$. Overall, the number of~$j$-adjacent vertices to~$v$ can therefore be bounded by
\begin{align*}
2^{3(r-j)} n_{j} + 2^{3(r-j)+1} \leq 2^{4(r-j)} n_j,
\end{align*}
where we used~$n_j\geq 2$ and~$j \leq r-1$.
\end{proof} 
 
The idea now is to consider the configuration of the agents with respect to a macro vertex $v_i$ exactly at the time $t$ when at least $\ceil{r/2}+1$ agents are $\ceil{r/2}$-adjacent to $v_i$. We then further use the fact that it is not possible to partition the agents $\agentset$ into two groups $\agentset'$ and $\agentset''$ with at most $i \geq \ceil{r/2}$ agents each that are separated on any path by at least two $i$-barriers. This yields the following bound on $\alpha_r$.

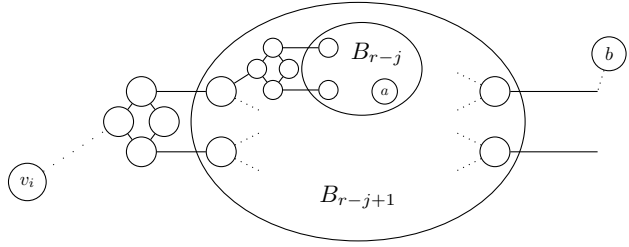
\begin{figure}
\centering

\tikzstyle{graphnode}=[circle,draw,minimum size=16pt,scale=0.7]
\tikzstyle{graphnode2}=[circle,draw,minimum size=8pt,scale=0.7]
\usetikzlibrary{shapes}
\tikzstyle{dnode}=[scale=0.9]
\newcommand*{\offX}{7.5}
\newcommand*{\diffX}{0.3}
\newcommand*{\diffY}{0.4}
\newcommand*{\offA}{2}
\newcommand*{\offB}{0.7}
\newcommand*{\shrink}{0.7}
\begin{tikzpicture}[auto=left]
\node(b1) at (0.5*\offX,0) [ellipse,draw,minimum width=125pt,minimum height=90pt] {};
\node(d1) at  (0.5*\offX,-1) [dnode]{$B_{r-j+1}$};
\node (w1) at (3*\diffX,\diffY) [graphnode] {};
\node (w3) at (3*\diffX,-\diffY) [graphnode] {};
\node (w4) at (4*\diffX,0) [graphnode] {};
\node (w2) at (2*\diffX,0) [graphnode] {};
\node (v1) at (6.5*\diffX,\diffY) [graphnode] {};
\node (v2) at (6.5*\diffX,-\diffY) [graphnode] {};

\draw (v1)-- (w1);
\draw (v2)--(w3);
\draw (w2)--(w1);
\draw (w3)-- (w2);
\draw (w4)--(w3);
\draw (w1)--(w4);
\draw[dotted](v1)--(6.5*\diffX+0.5,\diffY-0.25);
\draw[dotted](v2)--(6.5*\diffX+0.5,-\diffY+0.25);
\draw[dotted](v2)--(6.5*\diffX+0.5,-\diffY-0.25);

\node (w11) at (\shrink*3*\diffX+\offA,\shrink*\diffY+\offB) [graphnode2] {};
\node (w12) at (\shrink*3*\diffX+\offA,\shrink*-\diffY+\offB) [graphnode2] {};
\node (w13) at (\shrink*4*\diffX+\offA,0+\offB) [graphnode2] {};
\node (w14) at (\shrink*2*\diffX+\offA,0+\offB) [graphnode2] {};
\node (w15) at (\shrink*6*\diffX+\offA+0.1,\shrink*\diffY+\offB) [graphnode2] {};
\node (w16) at (\shrink*6*\diffX+\offA+0.1,\shrink*-\diffY+\offB) [graphnode2] {};
\node(b2) at (\offA+1.8,\offB) [ellipse,draw,minimum width=45pt,minimum height=35pt] {};
\node(d1) at   (\offA+2,\offB+0.2)[dnode]{$B_{r-j}$};

\draw (w14)-- (v1);
\draw (w14)-- (w11);
\draw (w14)-- (w12);
\draw (w12)-- (w13);
\draw (w11)-- (w13);
\draw (w11)-- (w15);
\draw (w12)-- (w16);

\node (u) at (\offA+2.1,\offB-0.3)[graphnode,scale=0.8] {$a$};
\node (v) at (-1.5*\diffX+\offX,\diffY+0.5)[graphnode] {$b$};
\node (v0) at (-2*\diffX,-2*\diffY)[graphnode] {$v_i$};

\node (v3) at (-6.5*\diffX+\offX,\diffY) [graphnode] {};
\node (v4) at (-6.5*\diffX+\offX,-\diffY) [graphnode] {};

\draw[dotted](v3)--(-6.5*\diffX+\offX-0.5,\diffY+0.25) ;
\draw[dotted](v3)--(-6.5*\diffX+\offX-0.5,\diffY-0.25) ;
\draw[dotted](v4)--(-6.5*\diffX+\offX-0.5,-\diffY+0.25);
\draw[dotted](v4)--(-6.5*\diffX+\offX-0.5,-\diffY-0.25);

\draw (v3)-- (-2.0*\diffX+\offX,\diffY);
\draw (v4)-- (-2.0*\diffX+\offX,-\diffY);
\draw[dotted] (v) -- (-2.0*\diffX+\offX,\diffY);

\draw[loosely dotted](v0) -- (w2);
\end{tikzpicture}
\caption{An $(r-j+1)$-barrier adjacent to $v_i$ containing $(r-j)$-barriers.}
\label{fig-barriers-separating-agents}
\end{figure}

\begin{lemma}
\label{lemma_bound_configurations2}
Let~$\agentset$ be a set of~$k$ agents, $s\geq 2$ and $r \in \{2,\dots,k\}$. 
We then have
 $$\alpha_r \leq s^{7 r^2} \cdot n_{\ceil{r/2}}^{\ceil{r/2}} \cdot n_{r-1} \cdot \prod_{j=\ceil{r/2}+1}^{r-1} n_j.$$
\end{lemma}

\begin{proof}
Let~$\agentset^{(r)} \subseteq \agentset$ be an arbitrary subset of $r$ agents. In order to bound $\alpha_r$,
we consider the behaviour of this subset $\agentset^{(r)}$ of agents in a graph of the form~$G(B_{r-1})$. 
We let~$v_0$ be the starting vertex of the set of agents~$\agentset^{(r)}$ in~$G(B_{r-1})$ and let~$t_0=0$. 
Again, we iteratively define~$t_i$ be the first point in time after~$t_{i-1}$, when an agent in $\agentset^{(r)}$ visits a macro vertex~$v_i$ distinct from~$v_{i-1}$.

Because of the recursive structure of the barriers, see Figure~\ref{fig_barrier_recursive}, every macro vertex
is surrounded by $\lceil r/2 \rceil$-barriers and any path between two consecutive macro vertices~$v_{i-1}$ and~$v_{i}$ contains at least one barrier~$B_{\lceil r/2 \rceil}$ (note that $r\geq 2$ by assumption). 
In order to reach the vertex $v_i$ after visiting~$v_{i-1}$, at least $\lceil r/2 \rceil+1$ agents from $\agentset^{(r)}$
are necessary to traverse such an $\lceil r/2 \rceil$-barrier. Thus, at some time~$t \in \{t_{i-1},\dots,t_i-1\}$ at least~$\lceil r/2 \rceil+1$ agents must be at a vertex that is~$\lceil r/2 \rceil$-adjacent to~$v_{i}$, as otherwise the agents would not be able to reach~$v_i$.

The crucial observation at this point is that by Lemma~\ref{lemma_same_conf_same_labels} the number of possible configurations at this time $t$ also bounds $\alpha_r$, the number of possible steps in the macro label sequence after which a configuration of $\agentset^{(r)}$ with respect to a macro vertex must repeat. The reason is that whenever the set of agents $\agentset^{(r)}$ traverse a gadget $B(l)$ there has to be a time $t$ with the properties described above. 

Let~$\agentset_1$ denote the set of agents that are at a vertex that is~$\lceil r/2 \rceil$-adjacent to~$v_{i}$ at time~$t$, and let~$\agentset_2:= \agentset^{(r)} \setminus \agentset_1$.
We claim the following: For~$j \in \{1,\ldots, |\agentset_2|\}$, there are at least~$(r-j)$ agents that are located at a vertex which is~$(r-j)$ adjacent to~$v_i$.

For~$j= |\agentset_2|$, we have~$r-j = |\agentset_1| > \lceil r/2 \rceil$. Thus, the claim holds by definition of~$\agentset_1$, since there are~$r-j$ agents, namely the set of agents~$\agentset_1$, which are located at vertices which are~$\lceil r/2 \rceil$-adjacent to~$v_i$ and thus also $(r-j)$-adjacent to~$v_i$ because ~$r-j > \lceil r/2 \rceil$. 

Now, assume for the sake of contradiction that the claim holds for~$j$, but not for~$j-1$. This means that there is a subset of agents~$\agentset' \subset \agentset^{(r)}$ with~$|\agentset'|=r-j$ such that all agents in~$\agentset'$ are located at vertices which are~$(r-j)$-adjacent to~$v_i$. 
But for~$j-1$ the claim does not hold, which implies that all other agents~$\agentset'':=\agentset^{(r)} \setminus \agentset'$ are at vertices which are not~$(r-j+1)$-adjacent: 
If there was an agent~$A \in \agentset''$ at a vertex which is~$(r-j+1)$-adjacent, then~$\agentset' \cup \{A\}$ would be a set of~$(r-j+1)$ agents which are all at~$(r-j+1)$-adjacent vertices, which is a contradiction to the choice of~$j$.

But the path between any pair of vertices~$(a,b)$, such that~$a$ is~$(r-j)$-adjacent to~$v_i$ and~$b$ is not~$(r-j+1)$-adjacent to~$v_i$, contains at least two~$(r-j)$-barriers, see also Figure~\ref{fig-barriers-separating-agents}.
The reason is that $r-j+1 > \lceil r/2 \rceil\geq 1$ and, by Theorem~\ref{theo_induction_barrier}, any path from $u$ or $v$ to $u'$ or $v'$ contains at least three~$(r-j)$ barriers.
 Thus the set of agents~$\agentset'$ and~$\agentset''$ are separated by at least two~$(r-j)$-barriers on any path and~$|\agentset'|\leq r-j$ as well as~$|\agentset''|=j < r-j$ since~$j \leq \lceil r/2 \rceil-1$. But then  a set of at most $r-j$ agents must have traversed a barrier $B_{r-j}$ or a set of at most $r-j-1$ agents must have traversed the gadget between two macro vertices in $B_{r-j}$, which both is a contradiction.
 
We now use the bound on the number of~$j$-adjacent vertices from Lemma~\ref{lemma-bound-i-adjacent-vertices} together with the claims to bound~$\alpha_r$.
By the claim above, we can enumerate the agents in~$\agentset^{(r)}$ as~$A_1,A_2,\ldots, A_r$ so that:
\begin{enumerate}
\item For~$j \in \{1,\dots,|\agentset_1|\}$, $A_j \in \agentset_1$ and the location of~$A_j$ is~$ \lceil r/2 \rceil$-adjacent to~$v_i$.
\item For~$j \in \{|\agentset_1|+1,\dots,r-1\}$, $A_j \in \agentset_2$ and  the location of~$A_j$ is~$j$-adjacent to~$v_i$.
\item Agent~$A_r\in \agentset_2$ is at~$v_i$ or one of the surrounding gadgets by Lemma~\ref{lem_need_all_pebbles_to_cross_barrier}.
\end{enumerate}
There are~$r!$ possible permutations of the agents and each agent has~$s$ possible states. Using Lemma~\ref{lemma-bound-i-adjacent-vertices}, we can bound the number of possible locations at time $t$ of the agents in $\agentset_1$ by $\smash{( 2^{4(r- \ceil{r/2})} n_{\ceil{r/2}})^{|\agentset_1|}}$, the number of possible locations of the agents $\{A_{|\agentset_1|+1},\ldots, A_{r-1}\}$ by
$ \prod_{j=|\agentset_1|+1}^{r-1} 2^{4(r-j)} n_j$ and the number of possible locations of $A_r$ by $2^4 n_{r-1}$. 
Overall, we can thus bound the number of possible locations of the agents~$\agentset^{(r)}$ at $t$ with respect to~$v_i$ by
\begin{align*}
& \quad r! \cdot \left( 2^{4(r- \ceil{r/2})} n_{\ceil{r/2}}\right)^{|\agentset_1|} \left ( \prod_{j=|\agentset_1|+1}^{r-1} 2^{4(r-j)} n_j  \right)\, 2^4 n_{r-1} \\
& \leq r! \cdot \left(2^{4r}\right)^r \cdot  n_{\ceil{r/2}}^{|\agentset_1|} \cdot n_{r-1} \cdot  \prod_{j=|\agentset_1|+1}^{r-1} n_j
\leq  2^{5 r^2} \cdot  n_{\ceil{r/2}}^{\ceil{r/2}} \cdot n_{r-1} \cdot  \prod_{j=\ceil{r/2}+1}^{r-1} n_j,
\end{align*}
where we used $r! \leq r^r \leq 2^{r^2}$ and $n_{j-1} \leq n_{j}$ for all $j \in \{2,\ldots, r-1\}$.

In order to bound the number of configurations of the agents~$\agentset^{(r)}$
note that there are $s^r$~possible states of the agents and for each agent~$3$ possible edge labels to the previous vertex. 
Combining these bounds with the above bound on the number of locations of the agents, we obtain
the following bound on the number of configurations  of~$\agentset^{(r)}$ at $t$ with respect to~$v_i$:
\begin{align*}
 s^r \cdot 3^r \cdot 2^{5 r^2} \cdot  n_{\ceil{r/2}}^{\ceil{r/2}} \cdot n_{r-1} \cdot \prod_{j=\ceil{r/2}+1}^{r-1} n_j 
 \leq s^{7 r^2} \cdot n_{\ceil{r/2}}^{\ceil{r/2}} \cdot n_{r-1} \cdot \prod_{j=\ceil{r/2}+1}^{r-1} n_j.
\end{align*}
Here we used $s \geq 2$ and $r \geq 2$.
By the observation at the beginning of the proof, the number of possible configurations of~$\agentset^{(r)}$ at $t$ with respect to~$v_i$ also bounds $\alpha_r$.
\end{proof}

Using the bound on~$\alpha_r$ from Lemma~\ref{lemma_bound_configurations2}, we can bound the number of vertices of the barriers.

\begin{theorem}\label{theo_size_of_barrier}
For every set of~$k$ agents~$\agentset$ with~$s$ states each and every $r \leq k$, there is an $r$-barrier with at most~$\bigO(s^{k \cdot 2^{4 \cdot r}})$ vertices.
\end{theorem}

\begin{proof}
The existence of an~$r$-barrier follows from Lemma~\ref{lemma_1_barrier} and Theorem~\ref{theo_induction_barrier} and we further have the following bound on the number of vertices~$n_r$ of~$B_r$ for a sufficiently large constant~$c \in \N$:
\begin{align*}
n_1 \leq c k s^2 \quad \text{and} \quad n_{r} \leq c \binom{k}{r} n_{r-1} \alpha_r^2.
\end{align*}
It is without loss of generality to assume~$s \geq 2$ since otherwise a trap of constant size can trivially be found. Hence, we can plug in the bound on~$\alpha_r$ from Lemma~\ref{lemma_bound_configurations2}. For the asymptotic bound, we may assume~$c\leq s^k$ and we further have~$\binom{k}{r} \leq 2^k$. We therefore get
\begin{align}
 n_{r} & \leq s^k \cdot 2^k \cdot n_{r-1} \cdot \left(s^{7 \cdot r^2}\right)^2 
 \cdot  \left(n_{\ceil{r/2}}^{\ceil{r/2}}\right)^2 \cdot  n_{r-1}^{2} 
 \prod_{j=\ceil{r/2}+1}^{r-1} n_j^2 \notag \\
 & \leq s^{2k + 14 r^2} \cdot n_{\ceil{r/2}}^{(r+1)}\cdot  n_{r-1}^{3} \prod_{j=\ceil{r/2}+1}^{r-1} n_j^2. \label{ineq_nr_barrier}
\end{align}
We proceed to show inductively that~$n_r \leq s^{k \cdot 2^{4 \cdot r}}$ holds for all~$r \in \{1,\ldots, k\}$. For~$r=1$, we have~$n_1 \leq c k s^2 \leq s^{4 k} \leq s^{k \cdot 2^4}$.
Let us assume the claim holds for~$1,\dots,r-1$. From Inequality~\eqref{ineq_nr_barrier} we obtain
\begin{align*}
 n_{r} & \leq s^{2k + 14 r^2}  \cdot 
	 \left( s^{k \cdot  2^{4 \cdot \ceil{r/2}}} \right)^{r+1} \cdot
	  \left(s^{k \cdot 2^{4(r-1)}} \right)^3 \cdot
	\prod_{j=\ceil{r/2}}^{r-1} \left( s^{k \cdot 2^{4\cdot j}} \right)^2 \\
& = s^ {2k + 14 r^2  +  k \cdot (r+1) \cdot  2^{4 \cdot \ceil{r/2}} + 3 \cdot k \cdot 2^{4(r-1)}+ 2 k  \sum_{j=\ceil{r/2}+1}^{r-1} 2^{4 \cdot j}}.	
\end{align*}
Thus, it is sufficient to bound the exponent. As $r \geq 2$, we have~$\sum_{i=0}^{r-1} 2^{4 \cdot i} = (2^{4r}-1)/(2^4-1)  \leq  2 \cdot 2^{4(r-1)}$ as well as  $(r+1) \cdot  2^{4\ceil{r/2}} \leq 4 \cdot 2^{4(r-1)}$ and~$2k + 14 r^2 \leq 2 \cdot k \cdot 2^{4(r-1)}$. Hence, we obtain
\begin{align*}
& \quad 2k + 14 r^2 +  k \cdot (r+1) \cdot  2^{4 \cdot \ceil{r/2}} + 3 \cdot k\cdot 2^{4(r-1)} +  2 k  \sum_{j=\ceil{r/2}+1}^{r-1} 2^{4 \cdot j} \\ 
& \leq k \cdot  \left( 2 \cdot 2^{4(r-1)} +  4 \cdot 2^{4(r-1)} + 3 \cdot 2^{4(r-1)} + 4 \cdot 2^{4(r-1)} \right) \leq k \cdot  2^{4 \cdot r}.
\end{align*}
This shows~$n_r \leq s^{k \cdot 2^{4 r}}$, as desired.
\end{proof}

The bound for the barriers above immediately yields the bound for the trap for~$k$ agents.
 
\begin{theorem}
\label{theo_size_of_trap}
For any set~$\agentset$ of~$k$ agents with at most~$s$ states each, there is a trap with at most~$\bigO(s^{2^{5 k}})$ vertices.
\end{theorem}

\begin{proof}
We can always add additional unreachable states to all agents so that all of them have~$s$ states. Theorem~\ref{theo_size_of_barrier} yields a $k$ barrier for a given set of $k$ agents~$\agentset$ with  $\bigO(s^{k \cdot 2^{4 \cdot k}})$ vertices. The claim follows from the fact that $k \cdot 2^{4 \cdot k} \leq 2^{5\cdot k}$ and that a $k$-barrier with $n$~vertices yields a trap with~$\bigO(n)$~vertices for~$\agentset$ by Lemma~\ref{lem_barrier_to_trap}.
\end{proof}

Finally, we derive a bound on the number of agents~$k$ that are needed for exploring every graph on at most~$n$ vertices. 

\begin{theorem}
\label{theo_pebbles_lower_bound}
The number of agents needed to explore every graph on at most~$n$ vertices is at least~$\Omega( \log \log n )$, if we allow~$\bigO( (\log n)^{1-\varepsilon})$ bits of memory 
for an arbitrary constant~$\varepsilon >0$ for every agent.
\end{theorem}

\begin{proof}
Let~$\agentset$ be a set of~$k$ agents with~$\bigO( (\log n)^{1-\varepsilon})$ bits of memory that explores any graph on at most~$n$ vertices. By otherwise adding some unused memory, we may assume that~$0<\varepsilon<1$ and that there is a constant~$c \in \N$  such that all agents in~$\agentset$ have~$\smash{s:=2^{c \cdot (\log n)^{1-\varepsilon}}}$ states. We apply Theorem~\ref{theo_size_of_trap} and obtain a trap for~$\agentset$ containing~$\bigO(s^{2^{5 \cdot k}})$ vertices.
As the set of agents~$\agentset$ explore any graph on at most~$n$ vertices, we have~$\smash{n \leq \bigO(1) s^{2^{5 \cdot k}}}$. By taking logarithms on both sides of this inequality, we obtain
\begin{align*}
\log n \leq \bigO(1) + 2^{5 k}  \log s =  \bigO(1) + 2^{5  k}  \cdot c \cdot (\log n)^{1-\varepsilon}.
\end{align*}
Multiplication by~$(\log n)^{\varepsilon-1}$ on both sides and taking logarithms yields the claim.
\end{proof}

As an additional agent is more powerful than a pebble (Lemma~\ref{lem:agent>pebble}), we obtain the following result as a direct corollary of Theorem~\ref{theo_pebbles_lower_bound}.

\begin{corollary}
\label{cor_pebbles_lower_bound}
An agent with $\bigO((\log n)^{1-\varepsilon})$ bits of memory for an arbitrary constant $\varepsilon >0$ needs $\Omega(\log \log n)$ pebbles to explore every graph with at most $n$ vertices.
\end{corollary}

\bibliographystyle{abbrv}
\bibliography{mrabbrev2012,journabbrev,references,confabbrev}

\end{document}